\documentclass[a4paper,11pt]{amsart} 
\usepackage{amsmath,amsxtra,amssymb,latexsym, amscd,amsthm}
\usepackage[mathscr]{eucal}
\usepackage{mathrsfs}
\usepackage{bbm}
\usepackage{enumerate}
\usepackage{pict2e}
\usepackage{tikz}
\usetikzlibrary{arrows.meta, positioning, calc}
\usepackage{graphicx}
\usepackage[a4paper]{geometry}
\usepackage{color}
\usepackage{hyperref}
\numberwithin{equation}{section}

\theoremstyle{plain}
\newtheorem{thm}{Theorem}[section]

\newtheorem{prp}[thm]{Proposition}
\newtheorem{cor}[thm]{Corollary}
\newtheorem{lem}[thm]{Lemma}
\newtheorem*{ego*}{Egorov's Theorem}
\theoremstyle{definition}
\newtheorem{defa}[thm]{Definition}

\newtheorem{rem}[thm]{Remark}

\newtheorem{exa}[thm]{Example}
\newtheorem*{rem*}{Remark}

\newtheorem{claim}{Claim}
\newtheorem*{claim*}{Claim}
\newtheorem{case}{Case}
\newcommand{\dd}{\mathrm{d}}
\newcommand{\ee}{\mathrm{e}}
\newcommand{\ii}{\mathrm{i}}
\renewcommand{\Im}{\operatorname{Im}}
\renewcommand{\Re}{\operatorname{Re}}

\newcommand{\N}{\mathbb{N}}
\newcommand{\Z}{\mathbb{Z}}
\newcommand{\Q}{\mathbb{Q}}
\newcommand{\R}{\mathbb{R}}
\newcommand{\C}{\mathbb{C}}

\newcommand{\T}{\mathbb{T}}

\newcommand{\LL}{\mathbb{L}}

\newcommand{\cC}{\mathcal{C}}

\newcommand{\cO}{\mathcal{O}}
\newcommand{\cP}{\mathcal{P}}

\newcommand{\cX}{\mathcal{X}}

\newcommand{\cS}{\mathcal{S}}

\newcommand{\br}{\mathbf{r}}
\newcommand{\bm}{\mathbf{m}}
\newcommand{\bk}{\mathbf{k}}
\newcommand{\bp}{\mathbf{p}}

\newcommand{\bv}{\mathbf{v}}
\newcommand{\bx}{\mathbf{x}}

\newcommand{\btheta}{\boldsymbol{\theta}}
\newcommand{\balpha}{\boldsymbol{\alpha}}
\newcommand{\Utheta}{\widehat{U}(\theta)}

\DeclareMathOperator{\prob}{\mathbb{P}}

\DeclareMathOperator{\supp}{supp}

\DeclareMathOperator{\opn}{Op_N}

\DeclareMathOperator{\Cay}{Cay}
\DeclareSymbolFont{extraup}{U}{zavm}{m}{n}
\DeclareMathSymbol{\varheart}{\mathalpha}{extraup}{86}
\DeclareMathSymbol{\vardiamond}{\mathalpha}{extraup}{87}

\let\nto\nrightarrow

\makeatletter
\newcommand*{\ovF}[1]{%
  $\m@th\overline{\raisebox{0pt}[\dimexpr\height+0.3mm\relax]{#1}}$%
	}
\makeatother
\title{Ergodicity in discrete-time quantum walks}
\author{Kiran Kumar and Mostafa Sabri}
\address{Science Division, New York University Abu Dhabi, Saadiyat Island, Abu Dhabi, UAE.}
\email{kas10285@nyu.edu}
\email{mostafa.sabri@nyu.edu}
\usepackage{calc}
\usepackage{graphicx}

\makeatletter
\newlength{\temp@wc@width}
\newlength{\temp@wc@height}
\newcommand{\widecheck}[1]{%
  \setlength{\temp@wc@width}{\widthof{$#1$}}%
  \setlength{\temp@wc@height}{\heightof{$#1$}}%
  #1\hspace{-\temp@wc@width}%
  \raisebox{\temp@wc@height+2pt}[\heightof{$\widehat{#1}$}]%
     {\rotatebox[origin=c]{180}{\vbox to 0pt{\hbox{$\widehat{\hphantom{#1}}$}}}}%
}
\makeatother

\begin{document}

\begin{abstract}
We undertake a detailed analysis of ergodicity for homogeneous discrete-time quantum walks on the integer lattice. The most significant result of our paper holds in dimension one, and gives a complete equivalence between the absolutely continuous spectrum of the unitary operator encoding the walk, and the equidistribution  of its dynamics in position space, which appears for the first time in the context of large-volume quantum ergodicity. In higher dimensions, we give a criterion for full and partial ergodicity in terms of a finer property of the spectrum which we dub ``No Repeating Graphs'', and we distinguish how strongly the equidistribution is taking place (weak convergence vs total variation). Many examples are included to illustrate the criterion and to distinguish between the types of ergodicity.  
\end{abstract}

\maketitle

\tableofcontents

\section{Introduction}

In this paper we consider homogeneous quantum walks over $\Z^d$ with a finite number of spins. We prove that under the action of the walk, localized initial states tend to equidistribute over the lattice as time grows large, provided certain spectral assumptions are met. 
This is illustrated by several examples and non-examples. 
Our treatment is especially complete in one dimension, where we identify all semiclassical measures by a careful analysis and prove a complete equivalence between the absolutely continuous spectrum and a form of ergodicity in position space. Let us give some background.

The problem of equidistribution of the quantum dynamics on large graphs is very active and has an interesting history. Some of the earliest papers on quantum walks already considered this question \cite{Aha,NayakVishwanath2000,Bednarska2003}. In these papers, one starts with an initial state $\psi$ of compact support (e.g. a qubit), considers a unitary operator $U$ encoding the walk, and studies whether the vector $U^k\psi$ becomes equidistributed in space as time goes on. Technically, it is necessary to consider the time-averaged evolution when working on a finite graph, otherwise the distribution will not converge \cite{Aha}. 

In more detail, given a walk with $\nu$ spins on a finite graph $G$, denote $\|\phi(r)\|_{\C^{\nu}}^2 = \sum_{i=1}^\nu |\phi_i(r)|^2$ for $r\in G$ and consider $\mu_{T,\psi}^G(r) = \frac{1}{T} \sum_{k=0}^{T-1} \| (U^k \psi)(r)\|_{\C^{\nu}}^2$. If the initial state $\psi$ is normalized, $\mu_{T,\psi}^G(r)$ gives (an averaged) probability that the walk is at position $r$ at time $T$. It is easy to see that this quantity converges as $T\to\infty$ to $\mu_\psi^G(r)=\sum_{i\le m} \|(P_{\lambda_i} \psi)(r)\|_{\C^\nu}^2$, where $P_{\lambda_i}$ are the orthogonal projections onto the $m$ distinct eigenvalues of $U$ (for this simply write $U^k\psi = \sum_i \lambda_i^kP_{\lambda_i}\psi$ and expand the square norm, noting that $|\lambda_i|=1$, the diagonal contribution is $\mu_\psi^G(r)$, the off-diagonal part vanishes in the limit as a geometric sum divided by $T$). 

The question is then whether $\mu_\psi^G$ is the uniform measure over $G$. This measure is difficult to analyze unless the walk is very simple, as one needs to understand the structure of $P_{\lambda_i}$, which is generally problematic for large graphs. The paper \cite{Aha} shows that $\mu_\psi^G$ is uniform on $G$ if $G$ is the Cayley graph of an abelian group and if all the eigenvalues of $U$ are distinct. The eigenvalue assumption is restrictive (and in general unnecessary, Proposition~\ref{prp:nrgstro}), in fact the result is only applied to Hadamard walks on odd cycles in \cite{Aha}. The paper \cite{Bednarska2003} later showed that $\mu_\psi^G$ is not uniform on even cycles. The discrepancy between $\mu_\psi^G$ and the uniform measure vanishes however as the size of the cycle goes to infinity. See Remark~\ref{rem:hadcyc} for further discussion on this particular walk.

This framework where the size of the graph grows large has been explored in much greater generality in the realm of quantum chaos. In that context, given a sequence of graphs $G_N$ of size $N$, one looks instead at the orthonormal eigenvectors $(\psi_j^{(N)})$ of the adjacency matrix of $G_N$, and studies the limiting behaviour of the probability density $|\psi_j^{(N)}(v)|^2$. Quantum ergodicity is the statement that this density equidistributes over $G_N$ when $N\to\infty$, for most eigenvectors corresponding to a spectral interval $I$. The general picture established in \cite{ALM,BLL,AS,McKSa,BLS} is that this should hold for any orthonormal basis $(\psi_j^{(N)})$ of eigenfunctions, provided the graphs $G_N$ converge to some infinite graph which has a ``strong form'' of absolutely continuous spectrum. One also needs some technical assumptions on $G_N$ for this. The papers \cite{MT,Naor,BL} take a different point of view, wherein there is no need to look at the limiting graph (if it exists), but only certain orthonormal bases $(\psi_j^{(N)})$ can be controlled.

If the above suggests that spectral delocalization implies ergodicity at the spatial level (eigenvectors), one may expect that it implies ergodicity at the dynamical level as well. This is precisely what is being assessed when we speak of ergodicity of quantum walks. The papers \cite{BS,BKS} answer this question positively for \emph{continuous-time quantum walks} on $\Z^d$-periodic graphs (crystals) such as the hexagonal lattice, strips and cylindrical hypercubes. Such walks are induced by the Schr\"odinger evolution $\ee^{-\ii tA_{G}}$ on $G$, the question being whether $\ee^{-\ii tA_{G_N}}\psi$ spreads out uniformly on $G_N$ as the time $t$ and size $N$ grow large. In the present paper, we study the same question for \emph{discrete-time quantum walks} on $\Z^d$, which are encoded by iterations of a unitary operator $U$ which is not a priori associated to any nice Schr\"odinger operator on $G$. Such walks are closer analogs to classical random walks. Our aim is to show a similar link between spectrum and dynamics, and to push the analysis further to discern different types of ergodicity and to obtain converse statements.

For a discrete-time quantum walk on $\Z^d$, the Hilbert space is $\mathscr{H} = \ell^2(\Z^d)\mathop\otimes \C^\nu$, representing motion along $\Z^d$ with $\nu$ possible spins. We first prove a general criterion saying that if the operator $U$ on $\mathscr{H}$ has ``nice'' Floquet eigenvalues, then the walk restricted to cubes $G_N$ converging to $\Z^d$, with periodic boundary conditions, is ergodic. In particular, if $\mu_{T,\psi}^N(\br) = \frac{1}{T}\sum_{k=0}^{T-1}\|(U^k_N\psi)(\br)\|_{\C^\nu}^2$ with $\psi$ of compact support and $\|\psi\|=1$, and if $\mu_\psi^N(\br)=\lim_{T\to\infty}\mu_{T,\psi}^N(\br)$, then $\mu_\psi^N$ approaches the uniform measure $\mu^N(\br) = \frac{1}{N^d}$ on $G_N$, as $N\to\infty$.

``Nice'' Floquet eigenvalues should not be flat (we assume absolutely continuous spectrum) and should not have a ``high frequency'' (no part of the graph of the eigenvalue, as a function of the quasimomentum, should repeat itself on sets of positive measure).

Despite its generality, this abstract criterion is not the main purpose of this paper. Technically speaking, its proof is quite similar to the case of continuous-time quantum walks previously established in \cite{BS}. Our main purpose here is to investigate how far this criterion applies to concrete quantum walks considered in the literature, illustrating examples and non-examples, and to see how far it can be improved. The strongest results we obtain are in one dimension, where we establish for the first time a complete equivalence between the absolutely continuous spectrum of $U$ (without any additional assumptions on the Floquet eigenvalues or resolvent) and the ergodicity of $U_N$ in position space. This is obtained as a consequence of analyzing all semiclassical limits of $\mu_{T,\psi}^N$ that arise when testing against a broader (nonregular) family of observables.

Different notions of ergodicity have also been explored in the very different framework of \emph{open quantum walks} in \cite{LS,KJ}, we refer to \cite{APSS} for a background on that model.

\subsection{Model and definitions}
We consider a quantum walk $U$ on $\Z^d$ with $\nu$ degrees of freedom (i.e. $\nu$ spins). The Hilbert space is thus $\mathscr{H} = \ell^2(\Z^d)\otimes \C^{\nu}$, which we identify with $\ell^2(\Z^d)^{\nu} = \ell^2(\Z^d,\C^\nu)$. To define $U$, let us first introduce the shift operators
\begin{align*}
S_{\bp}:\ell^2(\Z^d)&\to \ell^2(\Z^d)\\
\delta_{\bk} &\mapsto \delta_{\bk+\bp}
\end{align*}
where we use boldface indices for points in $\Z^d,\R^d$. Our quantum walk $U$ on $\ell^2(\Z^d,\C^\nu)$ acts on vector functions $\psi=\begin{pmatrix}\psi_1&\cdots&\psi_\nu\end{pmatrix}^T$ with $\psi_j\in \ell^2(\Z^d)$ and is defined by

\begin{equation}\label{e:ugen}
U = \begin{pmatrix} U_{1,1}&U_{1,2}&\cdots& U_{1,\nu}\\ \vdots& \ddots &&\vdots\\ U_{\nu,1}&\cdots& &U_{\nu,\nu}\end{pmatrix}
\end{equation}
where $U$ is unitary and each $U_{i,j}$ is an operator on $\ell^2(\Z^d)$. We assume the walk is homogeneous and of finite range, i.e. evolves at a finite distance at each step. For coined walks, homogeneity means that the coin is independent of the position. More generally, our assumptions mean that there exists a finite set $F\subset \Z^d$ and $U_{i,j}(\bp)\in \C$ such that
\begin{equation}\label{eqn:homogen_unitary}
	U_{i,j} = \sum_{\bp\in F} U_{i,j}(\bp)S_{\bp} \,.
\end{equation}

If $\psi\in \ell^2(\Z^d,\C^\nu)$ is an initial state, its subsequent time evolution is given by $U^n \psi$.

In bra-ket notation, the walk acts on a basis element $\mid j,\bk\rangle\in \C^\nu \mathop\otimes \ell^2(\Z^d)$ by 
\[
U: \,\, \mid j,\bk\rangle \mapsto \sum_{i=1}^\nu\sum_{\bp\in F} U_{i,j}(\bp) \mid i,\bk+\bp\rangle \,.
\]
This means that a qubit at position $\bk$ and spin $j$ is mapped to a superposition of states at positions $\bk+\bp$ and spins $i$. We do not use this notation much in the article, however we include a detailed translation in Appendix~\ref{app} for readers who prefer it.

Simple examples of walks of the form \eqref{e:ugen}-\eqref{eqn:homogen_unitary} are the Hadamard and Grover walks on $\Z$, given by the operators
\begin{equation}\label{e:uhagro}
U_{\mathrm{Had}} = \frac{1}{\sqrt{2}}\begin{pmatrix} S_{-1}&S_{-1}\\ S_{1}&-S_{1}\end{pmatrix} \qquad \text{and} \qquad U_{\mathrm{Gro}} = \frac{1}{3} \begin{pmatrix} -S_{-1}& 2S_{-1}& 2S_{-1}\\ 2&-1&2\\2S_1&2S_1&-S_1\end{pmatrix}
\end{equation}
i.e. each $U_{i,j}$ reduces to a single shift. See Appendix~\ref{app} for the more common shift-coin definition, which is equivalent to \eqref{e:uhagro}. Our framework is a lot more general and includes homogeneous \emph{split-step} quantum walks, the \emph{shunt decomposition model}, the \emph{arc reversal model}, the higher dimensional \emph{PUTO model} and more, see Appendix~\ref{app} for details. 

The operator $U:\mathscr{H}\to\mathscr{H}$ given in \eqref{e:ugen}-\eqref{eqn:homogen_unitary} is unitarily equivalent to the operator $M_{\widehat{U}}$ on $L^2(\T^d)^{\nu}$ of multiplication by the unitary Floquet matrix function $\widehat{U}(-\btheta)$, where
\begin{equation}\label{e:uhatheta}
\widehat{U}(\btheta) = \begin{pmatrix} \widehat{U}_{1,1}(\btheta)&\cdots &\widehat{U}_{1,\nu}(\btheta)\\ & \ddots\\ \widehat{U}_{\nu,1}(\btheta) & \cdots & \widehat{U}_{\nu,\nu}(\btheta) \end{pmatrix} \quad \text{and}\quad \widehat{U}_{i,j}(\btheta)=\sum_{\bp\in F} U_{i,j}(\bp)\ee^{-2\pi\ii\btheta\cdot \bp} \,.
\end{equation}

Here $\T^d \equiv [0,1)^d$. This result is known; see Lemma~\ref{lem:foufin} and the comment thereafter. Our convention to take $\ee^{-2\pi\ii\btheta\cdot \bp}$ rather than $\ee^{2\pi\ii\btheta\cdot \bp}$ in $\widehat{U}(\btheta)$ is just to simplify notations that will arise later.

Let $\LL_N^d := \{0,1,\dots,N-1\}^d\subset \Z^d$ be a large box and let $U_N$ be the restriction of $U$ to $\ell^2(\LL_N^d,\C^\nu)$ with periodic boundary conditions. We fix an initial state $\psi\in \ell^2(\LL_N^d)^\nu$ of compact support (for example a qubit $ \delta_\bk \mathop\otimes f$) and follow its evolution under $U_N$. 
Assume $\|\psi\|=1$.
As in \cite{Aha,Bednarska2003,Konno}, we define the probability that the quantum walk on $\LL_N^d$, at time $k$, is at location $\br$, by
\[
\prob_N(X_k=\br) = \|(U^k_N\psi)(\br)\|^2_{\C^{\nu}} = \sum_{i=1}^{\nu} |[U^k_N\psi]_i(\br)|^2 
\]
and consider the time-average
\[
\mu_{T,\psi}^N(\br) = \frac{1}{T}\sum_{k=0}^{T-1} \prob_N(X_k=\br)= \frac{1}{T}\sum_{k=0}^{T-1} \sum_{i=1}^{\nu} |[U_N^k\psi]_i(\br)|^2 \,.
\]
This $\mu_{T,\psi}^N$ is a probability measure on $\LL_N^d$ as $\mu_{T,\psi}^N(\LL_N^d) = \frac{1}{T}\sum_{k=1}^T \|U_N^k\psi\|^2= \|\psi\|^2=1$, as $U_N$ is unitary. At time $T=0$, $\mu_{0,\psi}^N =  \psi$ is fully localized (think of $\psi$ as a qubit). Our aim is to show that under some assumptions of spectral delocalization, $\mu_{T,\psi}^N$ approaches the uniform measure on $\LL_N^d$ as $T$ and $N$ get large.
This will be done by comparing the averages of a function $\phi$ on $\LL_N^d$ with respect to $\mu_{T,\psi}^N$ and the uniform measure $\mu^N$. The choice of these observables is important, so let us introduce the following definition.

\begin{defa}\label{def:rego}
We say that $\phi^{(N)}\in \ell^2(\LL_N^d)$ is a \emph{regular observable} if it satisfies one of the following conditions:
\begin{enumerate}[(i)]
\item $\phi^{(N)}(\bk) = f(\bk/N)$ for some fixed $f\in H^s(\T^d)$, with $s>d/2$,
\item or $\phi^{(N)}$ is the restriction to $\LL_N^d$ of a summable function $\phi\in \ell^1(\Z^d)$.
\end{enumerate}
Recall that by the Sobolev embedding theorem, under assumption (i), $f$ can be chosen to be continuous, and we are implicitly making this choice here.

Similarly, we say that $a\in \ell^2(\LL_N^d,\C^\nu)$ is a \emph{regular observable} if each coordinate $a_j\in \ell^2(\LL_N^d)$ is regular.
\end{defa}

Our strongest ergodicity results will hold if the following spectral property is satisfied :

\begin{defa}
Let $\widehat{U}(\btheta)$ be the Floquet matrix \eqref{e:uhatheta} with eigenvalues $\{E_s(\btheta)\}_{s=1}^\nu$. We say that the Floquet eigenvalues have \emph{No Repeating Graphs} if 
\begin{equation}\label{e:flo}
\sup_{\bm\neq \mathbf{0}} \frac{\#\{\br\in \LL_N^d:E_s(\frac{\br+\bm}{N})-E_w(\frac{\mathbf{r}}{N})=0\}}{N^d}\longrightarrow 0 \qquad\qquad \forall\,s,w\le \nu\tag{\textbf{NRG}}
\end{equation}
as $N\to\infty$.
\end{defa}

\eqref{e:flo} also stands for e\textbf{n}e\textbf{rg}y, as it is an assumption on the energies of $\widehat{U}(\btheta)$.

When we ask for \eqref{e:flo} to be satisfied, we mean the full sequence $N=1,2,3,\dots$.

Condition \eqref{e:flo} appeared before in \cite{McKSa,BS} in the context of ergodicity of eigenbases and continuous-time quantum walks. It implies that the spectrum of $U$ is absolutely continuous with no eigenvalues, as $U$ has an eigenvalue iff $E_s(\btheta)=\lambda_0$ is constant \cite{Tate2019}, in which case $E_s(\frac{\br+\bm}{N})=E_s(\frac{\br}{N})$ $\forall \,\br,\bm$ and \eqref{e:flo} is violated.
For quantum walks on $\Z$, we prove in Lemma \ref{lem:1D_subseq_2conditions} that \eqref{e:flo} fails iff $E_s(\theta+\varphi) \equiv E_w(\theta)$ for some rational $0<\varphi<1$. If $s=w$, this means that $E_s$ has a short period (high frequency). For higher-dimensional walks, \eqref{e:flo} seems to be equivalent to the following statement: for any non-zero $ \balpha \in [0,1)^d \cap \Q^d$, the zero set of $E_s(\btheta+\balpha)-E_w(\btheta)$ has Lebesgue measure zero, see \S\,\ref{rem:higher_dim_equiv}. Note that having a Floquet eigenvalue of higher multiplicity, $E_s(\btheta)\equiv E_w(\btheta)$, does not violate \eqref{e:flo}. For example, the rather trivial walk $\begin{pmatrix} S_1&0\\0&S_1\end{pmatrix}$ with Floquet eigenvalues $E_1(\theta)=E_2(\theta)=\ee^{-2\pi\ii\theta}$ satisfies \eqref{e:flo}.

\subsection{A simple example}\label{sec:simpexa}
Before we state our main results, let us explain why an assumption like \eqref{e:flo} may be natural. Consider the toy walk $U = \begin{pmatrix} S_2&0\\0&S_{-2}\end{pmatrix}$ over $\Z$. If $\psi = \begin{pmatrix} \delta_0\\0\end{pmatrix}$, then $U_N^k \psi = \begin{pmatrix} \delta_{2k}\\0\end{pmatrix}$. If $N$ is even, we see that this walk only visits even integers, so it does not equidistribute over $\LL_N$. Here $\LL_N$, due to the periodic conditions, amounts to an $N$-cycle. The issue here is partly caused by the large hoping, which is spectrally translated into the high frequency eigenvalues of the Floquet matrix, namely $\ee^{\pm 4\pi\ii\theta}$, something that \eqref{e:flo} aims to exclude. This walk has purely absolutely continuous spectrum, but lacks perfect ergodicity.

It is also worthwhile to note that if $N=2n+1$ is odd, then the sequence does visit all vertices of the cycle, because it hops from $0$ until reaching $2n$ and then to $1$. This is something that we will see more generally in one-dimensional walks.

\subsection{Main results}
Our first result is the following criterion for ergodicity.

\begin{thm}\label{thm:cri0}
Let $U$ be a quantum walk \eqref{e:ugen}-\eqref{eqn:homogen_unitary} satisfying \eqref{e:flo}. Let $\psi = \sum_{\bk\in \Lambda} \psi(\bk)\delta_{\bk}$ be an initial state of compact support, i.e. $\Lambda\subset \Z^d$ is finite and independent of $N$, and assume $\|\psi\|=1$. 

Then for any regular observable $\phi=\phi^{(N)}\in \ell^2(\LL_N^d)$,
\begin{equation}\label{e:mainlim}
\lim_{N\to\infty}\Big|\lim_{T\to\infty} \langle \phi \rangle_{T,\psi} - \langle \phi\rangle \Big|=0 \,, \tag{\textbf{PQE}}
\end{equation}
where $\langle \phi \rangle_{T,\psi} = \sum_{\br\in \LL_N^d} \phi(\br) \mu_{T,\psi}^N(\br)$, while $\langle \phi\rangle = \frac{1}{N^d} \sum_{\br\in \LL_N^d} \phi(\br)$ is the uniform average.
\end{thm}

Here \eqref{e:mainlim} stands for \emph{position quantum ergodicity}. Theorem~\ref{thm:cri0} is a direct consequence of Theorem~\ref{thm:pergra}, which will yield a more precise \emph{full quantum ergodicity} \eqref{eqn:ergodic_general} in both position and spin spaces. Before we introduce it, let us explain \eqref{e:mainlim} further.

Theorem~\ref{thm:cri0} shows that $\mu_{T,\psi}^N(\br) \approx \frac{1}{N^d}$ in the sense of measures on $\LL_N^d$. As the space changes with $N$, we cannot directly formulate this as a weak convergence. This can be somehow fixed. First recall that quite trivially, for fixed $N$ and any $\br\in \LL_N^d$,
\begin{equation}\label{e:limtonly}
\mu_{T,\psi}^N(\br) \xrightarrow{T\to\infty} \mu_\psi^N(\br) \,,
\end{equation}
where $\mu_{\psi}^N(\br) = \sum_{i\le m}\|(P_{\lambda_i^{(N)}}\psi)(\br)\|_{\C^\nu}^2$ (see the introduction). This is equivalent to the weak convergence $\mu_{T,\psi}^N\xrightarrow{w} \mu_\psi^N$. Now recall that we are considering periodic boundary conditions; in one dimension this amounts to studying the walk on a large cycle. To fix the space, we may embed this cycle into the unit circle. More generally, we define the measure $\widetilde{\mu}_{T,\psi}^N$ on the torus $\T^d$ by
\[
\widetilde{\mu}_{\psi}^N := \sum_{\bk\in\LL_N^d}  \mu_{\psi}^N(\bk) \delta_{\frac{\bk}{N}}
\]
so that $\widetilde{\mu}_{\psi}^N$ is concentrated on $\{\frac{\bk}{N}:\bk\in\LL_N^d\}$, with $\widetilde{\mu}_{\psi}^N(\frac{\bk}{N})=\mu_{\psi}^N(\bk)$. Then we have:

\begin{prp}\label{prp:wea}
Let $\psi = \sum_{\bk\in \Lambda} \psi(\bk)\delta_{\bk}$ be an initial state of finite support independent of $N$, and assume $\|\psi\|=1$. Then \eqref{e:mainlim} for regular observables implies the weak convergence
\begin{equation}\label{e:weauni}
\widetilde{\mu}_{\psi}^N \xrightarrow{w} \mu
\end{equation}
where $\dd \mu(\mathbf{x})=\dd \mathbf{x}$ is the uniform measure on $\T^d$. In particular, \eqref{e:weauni} holds for any quantum walk \eqref{e:ugen}-\eqref{eqn:homogen_unitary} satisfying \eqref{e:flo}.
    \end{prp}

Recall that assumption \eqref{e:flo} implies that $U$ has no flat bands, i.e. $U$ has no eigenvalue. 
It is natural to exclude flat bands. In fact,

\begin{prp}\label{prp:flanoqe}
If $U$ has a flat band, there exists some initial state $\psi$ of compact support, and a regular observable $\phi=\phi_N$, such that $\langle \phi\rangle_{T,\psi}= 1$ and $\langle \phi\rangle=  \frac{c}{N^d}$ for all $T$ and $N$. In particular, \eqref{e:mainlim} fails.
\end{prp}

In specific examples, we can take the initial state $\psi$ to be a qubit.

\begin{lem}\label{lem:inui}	
The Grover walk violates \eqref{e:mainlim} for the initial state $\psi=\delta_0\otimes f$ and the observable $\phi=\delta_0$, for any $f\in \C^3$, $\|f\|=1$, such that $f\notin \mathrm{span}(\frac{1}{\sqrt{6}},\frac{-2}{\sqrt{6}},\frac{1}{\sqrt{6}})$.
\end{lem}

Let us now discuss full ergodicity:

\begin{thm}\label{thm:pergra}
Let $U$ be a quantum walk \eqref{e:ugen}-\eqref{eqn:homogen_unitary} satisfying \eqref{e:flo}.
Suppose the observable $a=a^{(N)}\in \ell^2(\LL_N^d)^{\nu}$ is regular. 
Then for any initial state $\psi$ of compact support,
\begin{equation}\label{eqn:ergodic_general}
 \lim_{N\to\infty}\Big|\lim_{T\to\infty}\frac{1}{T}\sum_{k=0}^{T-1} \langle U_N^k\psi, a U_N^k\psi\rangle
- \langle a\rangle_\psi\Big|=0\,, \tag{\textbf{FQE}}  
\end{equation}
where, denoting $\langle a_j\rangle:= \frac{1}{N^d} \sum_{\bk\in \LL_N^d} a_j(\bk)$ and $P_{E_s}(\btheta)$ the spectral projections corresponding to the $\nu'\le \nu$ distinct eigenvalues $E_s(\btheta)$ of $\widehat{U}(\btheta)$,
\begin{equation}\label{e:avp}
\langle a\rangle_\psi = \sum_{j=1}^\nu \langle a_j\rangle \sum_{\br\in\LL_N^d} \sum_{s=1}^{\nu'} \Big|\Big[P_{E_s}\Big(\frac{\br}{N}\Big)\widehat{\psi}(\br)\Big]_j\Big|^2\,.
\end{equation}
\end{thm} 

This theorem is more precise than Theorem~\ref{thm:cri0} because it shows how the mass spreads out both in position space $\LL_N^d$ and in spin space $\C^\nu$. This spreading is generally not uniform in spin space. The reason that we get the uniform measure in Theorem~\ref{thm:cri0} is that we sum over the spin contributions to assess the probability of being at a position $\br\in\LL_N^d$.

Still, the presence of $\langle a_j\rangle$ implies uniformity in position space. To make this more precise, we may embed the walk in the torus as before by defining
\[
\widetilde{\mu}_{\psi,j}^N = \sum_{\bk\in\LL_N^d} \mu_{\psi,j}^N(\bk)\delta_{\frac{\bk}{N}}
\]
where $\mu_{\psi,j}^N(\br) = \lim_{T\to\infty} \mu_{T,\psi,j}^N(\br)$ and $\mu_{T,\psi,j}^N(\br) = \frac{1}{T}\sum_{k=0}^{T-1}|[U_N^k\psi]_j(\br)|^2$. Then we have:

\begin{prp}\label{prp:weaj}
Let $\psi = \sum_{\bk\in \Lambda} \psi(\bk)\delta_{\bk}$ be an initial state of finite support $\Lambda\subset \Z^d$ independent of $N$, and assume $\|\psi\|=1$. Then \eqref{eqn:ergodic_general} for regular observables implies the weak convergence
\begin{equation}\label{e:wea}
\widetilde{\mu}_{\psi,j}^N \xrightarrow{w} c_{\psi,j}\mu
\end{equation}
where $\dd \mu(\mathbf{x})=\dd \mathbf{x}$ is the uniform measure on $\T^d$ and $c_{\psi,j}$ is the explicit constant
\[
c_{\psi,j}= \int_{\T^d} \sum_{s=1}^{\nu'} \bigg|\sum_{\bm\in\Lambda} [P_{E_s}(\btheta)\psi(\bm)]_j\ee^{-2\pi\ii \bm\cdot\btheta}\bigg|^2\,\dd\btheta \,.
\]
In particular, \eqref{e:wea} holds for any quantum walk \eqref{e:ugen}-\eqref{eqn:homogen_unitary} satisfying \eqref{e:flo}.
    \end{prp}

    If the initial state $\psi=\delta_{\mathbf{0}}\mathop\otimes \delta_\ell$ is a qubit, we get $c_{\psi,j} = \int_{\T^d} \sum_{s=1}^{\nu'} | P_{E_s}(\btheta)(j,\ell)|^2\,\dd\btheta$.

    So far we showed that for general quantum walks \eqref{e:ugen}-\eqref{eqn:homogen_unitary}, \eqref{e:flo}$\implies$\eqref{eqn:ergodic_general}$\implies$\eqref{e:mainlim} for regular observables. We now specialize to $d=1$ and study homogeneous quantum walks over $\Z$ of finite range, where we push the machinery to the limit and obtain a lot more precise results summarized in the following.
    
\begin{thm}[One-dimensional walks]\label{thm:main1d}
Let $U$ be a quantum walk \eqref{e:ugen}-\eqref{eqn:homogen_unitary} over $\Z$ with an initial state $\psi$ of compact support, $\|\psi\|=1$. 
\begin{enumerate}[\rm(1)]
    \item If \eqref{e:flo} holds, then \eqref{eqn:ergodic_general} holds for all observables $a=a^{(N)}$ with $\|a^{(N)}\|_\infty\le 1$.
    \item If $U$ has no flat bands, then there is always a subsequence $N_n$ such that \eqref{e:flo} holds on $\LL_{N_n}$ and \eqref{eqn:ergodic_general} holds for all $a^{(N_n)}$, $\|a^{(N_n)}\|_\infty\le 1$.
    \item \eqref{e:flo} fails iff $\limsup_N\sup_{m\neq 0}\#\{r\in \LL_N:E_s(\frac{r+m}{N})-E_w(\frac{r}{N})\}= \infty$.
    \item If $U$ has no flat bands, \eqref{e:mainlim} can fail for observables $\phi^{(N)}$ with $\|\phi^{(N)}\|_\infty\le 1$. Still, we find all ``semiclassical measures'' in this case, and exhibit a phenomenon of equidistribution on subsets
    \item $U$ has no flat bands iff \eqref{e:mainlim} holds for regular observables, for the full sequence.
    \item Coined walks $U=\cS\cC$ with a $2\times 2$ coin and shifts to the right and left by distances $\alpha,\beta\in\N^\ast$ exhibit all kinds of dynamical behaviors, and we fully characterize them in terms of $\cC$, $\alpha$ and $\beta$. Split-step walks are also treated.
    \item There exist walks satisfying \eqref{e:mainlim} for regular observables, but not for arbitrary $\|\phi^{(N)}\|_\infty\le 1$. There exist walks satisfying \eqref{e:mainlim} for $\|\phi^{(N)}\|_\infty\le 1$, but not \eqref{eqn:ergodic_general} for arbitrary $\|a^{(N)}\|_\infty \le 1$. A walk may satisfy \eqref{eqn:ergodic_general} and not \eqref{e:flo}.
\end{enumerate}
\end{thm}

In (1),(2), we can actually take observables $a^{(N)}$ with sup norm growing not too fast. We summarize most of these findings in Figure~\ref{fig:sum}. 

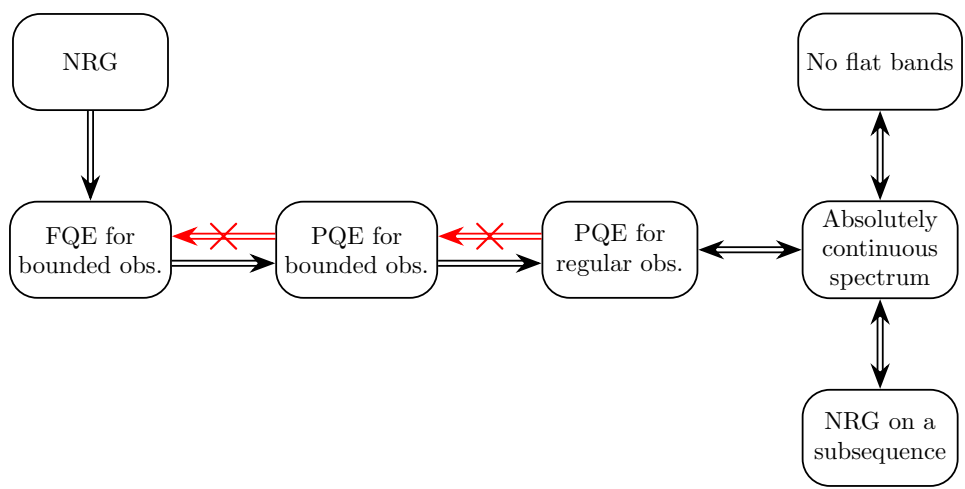
\begin{figure}[h!]
\scalebox{0.85}{ 
\begin{tikzpicture}[
    box/.style={
        draw,
        rounded corners=12pt,
        minimum width=2.4cm,
        minimum height=1.5cm,
        align=center,
        thick
    },
    impl/.style={double distance=1.5pt, -{Stealth[length=4mm]}, thick},
    implrev/.style={double distance=1.5pt, {Stealth[length=4mm]}-, thick},
    iff/.style={double distance=1.5pt, {Stealth[length=4mm]}-{Stealth[length=4mm]}, thick},
]

\node[box] (NRG) {NRG};
\node[box, below=1.4cm of NRG] (FQE) {FQE for\\ bounded obs.};

\node[box, right=1.6cm of FQE] (PQE1) {PQE for\\ bounded obs.};
\node[box, right=1.6cm of PQE1] (PQE2) {PQE for\\ regular obs.};

\node[box, right=1.6cm of PQE2] (ACS) {Absolutely\\ continuous\\ spectrum};
\node[box, above=1.4cm of ACS] (NoFlat) {No flat bands};
\node[box, below=1.4cm of ACS] (NRGsub) {NRG on a\\ subsequence};

\draw[impl] ([xshift=-9pt]NRG) -- ([xshift=-9pt]FQE);
\draw[implrev,red]
    ([xshift=9pt]NRG.south) -- ([xshift=9pt]FQE.north)
    node[midway]{\textcolor{red}{\Huge $\times$}};

\draw[impl] ([yshift=-6pt]FQE.east) -- ([yshift=-6pt]PQE1.west);
\draw[implrev,red]
    ([yshift=6pt]FQE.east) -- ([yshift=6pt]PQE1.west)
    node[midway]{\textcolor{red}{\Huge $\times$}};

\draw[impl] ([yshift=-6pt]PQE1.east) -- ([yshift=-6pt]PQE2.west);
\draw[implrev,red]
    ([yshift=6pt]PQE1.east) -- ([yshift=6pt]PQE2.west)
    node[midway]{\textcolor{red}{\Huge $\times$}};

\draw[iff] (PQE2) -- (ACS);
\draw[iff] (ACS) -- (NoFlat);
\draw[iff] (ACS) -- (NRGsub);
\end{tikzpicture}
}
\caption{Summary of results for walks on $\Z$.}\label{fig:sum}
\end{figure}

The equivalence between absolutely continuous spectrum and ergodicity of quantum dynamics in position space for regular observables is significant. To our knowledge, it appears for the first time in the literature on large-graph quantum chaos. This is obtained after a detailed analysis of the semiclassical measures in that context.

The figure also suggests that \eqref{e:flo} is markedly stronger than absolutely continuous spectrum in general, so the fact that the latter always guarantees a subsequence satisfying \eqref{e:flo} seems quite strong.

On the other hand, the equivalence between absolutely continuous spectrum and the absence of flat bands (eigenvalues) is classical, as singularly continuous spectrum cannot exist in periodic settings of finite range, see e.g. \cite{Tate2019}. The fact that \eqref{e:flo} on a subsequence implies absolutely continuous spectrum is also immediate; as we noted previously, the presence of a flat band would lead to $E_s(\frac{\br+\bm}{N_n})=E_s(\frac{\br}{N_n})$ $\forall\,\br,\bm$.

To better appreciate the relevance of assessing the  ergodicity of walks using  bounded observables, recall the simple example \S\,\ref{sec:simpexa}. There, if $\phi$ is regular, then $\sum_{r\in\LL_N} \phi(r)\mu_{\psi}^N(r)$ and $\frac{1}{N}\sum_{r\in\LL_N}\phi(r)$ share the same limit (either a Riemann integral or zero). In other words, such $\phi$ cannot distinguish between $\mu_{\psi}^N$, which is uniformly spread on the even numbers of $\LL_{2n}$, and $\mu^N$, which is uniform on all $\LL_{2n}$. In some sense, such observables are not fine enough to detect that the ergodicity is imperfect. In contrast, if we take $\phi$ to be the indicator function of even integers in $\LL_{2n}$, then these two sums approach $1$ and $\frac{1}{2}$, respectively, showing that bounded observables can detect this subtlety.

In general, going from regular to bounded observables improves the comparison from weak convergence to convergence in total variation:

\begin{prp}[Convergence in total variation]\label{prp:tvd}
For walks \eqref{e:ugen}-\eqref{eqn:homogen_unitary} on $\Z^d$, \eqref{e:mainlim} holds for all observables $\|\phi^{(N)}\|_\infty \le 1$ if and only if $\delta(\mu_\psi^N,\mu^N)\to 0$, where $\delta$ is the total variation distance and $\mu^N(\br)=\frac{1}{N^d}$ is the uniform measure on $\LL_N^d$.

Similarly, \eqref{eqn:ergodic_general} holds for all observables $a=a^{(N)}$ with $\|a^{(N)}\|_\infty\le 1$ if and only if $\delta(\mu_{\psi,j}^N,c_{N,\psi,j}\mu^N)\to 0$, where $c_{N,\psi,j}\mu^N$ is a constant multiple of $\mu^N$.
\end{prp}

One may naturally ask if results like Theorem~\ref{thm:main1d} and Figure~\ref{fig:sum} hold in higher dimensions. Quite interestingly, we show that certain implications fail:

\begin{prp}\label{prp:failhi}
There exists a quantum walk \eqref{e:ugen}-\eqref{eqn:homogen_unitary} over $\Z^2$ which has purely absolutely continuous spectrum, but has no subsequence satisfying \eqref{e:flo}.

The same walk also violates \eqref{e:mainlim} for regular observables.
\end{prp}

Finally, our results in one dimension also have interesting applications for continuous-time quantum walks and eigenvectors ergodicity of 1d periodic Schr\"odinger operators. We collect these in \S\,\ref{sec:schro} to avoid overcrowding the present section. 

\subsection{Open questions}
In view of Figure~\ref{fig:sum} and Proposition~\ref{prp:failhi}, a natural question is whether \eqref{e:flo} implies \eqref{eqn:ergodic_general} for all bounded observables $\|a^{(N)}\|_\infty\le 1$ in dimension $d>1$, not only regular observables. We have no conjecture about this.

Another question concerns expressing \eqref{e:flo} in terms of the Lebesgue measure of the zero set of $E_s(\btheta+\balpha)-E_w(\btheta)$, without referring to a sequence of approximating graphs $\LL_N^d$. We almost succeed in establishing this characterization in \S\,\ref{rem:higher_dim_equiv} using arguments from \cite{Wen1}, except for a subtlety concerning the rationality of $\balpha$. Settling this would be desirable.

As stated in Theorem~\ref{thm:main1d}, ergodicity by itself does not imply the spectral property of no repeating graphs. The counterexample we give consists of a very trivial walk; it might be that \eqref{e:flo} is a characterization of full ergodicity in more specific walks. Towards this question, we show in Theorem~\ref{thm:fqenrg} that full ergodicity essentially implies, either \eqref{e:flo}, which is a statement on the \emph{difference} between  eigenvalues at different quasimomenta, or the vanishing of the corresponding \emph{products} of Floquet eigenprojectors. 

There is also interest in exploring ergodicity beyond the present homogeneous setting. In particular, walks on $\Z$ with non-constant periodic coins, or with periodic electric fields \cite{AbSt,ACSW,CRWAGW} are physically interesting. We expect a general criterion like Theorem~\ref{thm:pergra} to hold in such periodic frameworks, however in Theorem~\ref{thm:cri0}, one does not expect $\mu_{T,\psi}^N$ 
 to be perfectly uniform anymore, as its profile will be affected by the field or the structure of the graph, if one goes beyond $\Z^d$ and studies general crystals. 

\subsection{Organization of the paper} We start by proving our general criteria in Section~\ref{sec:mainproof}, namely the implications of \eqref{e:flo} on ergodicity in general dimension, the implications on weak convergence and convergence in total variation distance, and the fact that flat bands preclude \eqref{e:mainlim}. In the following Section~\ref{sec:proof1d}, we specialize to dimension one and prove most of the results collected as Theorem~\ref{thm:main1d}. This is, in our view, the most important contribution of the paper. Section~\ref{sec:2state} specializes further to coined and split-step walks, as they are quite popular models in the literature, and we bring some generalizations to these models, where we play with the step sizes of the walks and notice interesting phenomena. In Section~\ref{sec:higher_dim_QW} we explore various models of quantum walks in higher dimension, some ergodic, some not. In the case of the Fourier coin, we prove ergodicity by combining our criterion with the irreducibility theory of Bloch varieties. Section~\ref{sec:conclu} concludes with a discussion on how far \eqref{e:flo} characterizes ergodicity. Finally, in Appendix~\ref{app} we review notations used in the literature for the reader's convenience, and some background on models studied here, while in Appendix~\ref{app:rage}, we adapt the RAGE theorem to our quantum walks. We include this here for completeness, as a different more basic criterion linking the spectrum and dynamics, which is significantly easier to prove.

\subsection*{Acknowledgments} 
We warmly thank Houssam Abdul-Rahman, Wencai Liu and Armin Rainer for interesting discussions and clarifications on their respective works.

\section{Proof of the general criterion}\label{sec:mainproof}

In this section, after some basic Fourier analysis in \S\,\ref{sec:basic}, we prove Theorem~\ref{thm:pergra} in 4 steps in \S\,\ref{sec:step1}--\S\,\ref{subsec:step4}, and show that it implies Theorem~\ref{thm:cri0} as a special case in \S\,\ref{sec:fp}. We then prove Propositions~\ref{prp:wea}, \ref{prp:weaj}, \ref{prp:tvd} and \ref{prp:flanoqe} in \S\,\ref{sec:weak}--\ref{sec:compact}.

\subsection{Basic Fourier Analysis.}\label{sec:basic} 
Consider a cube $\LL_N^d\subset \Z^d$ and restrict $U$ with periodic boundary conditions: $S_{\bp} \delta_{\bk} = \delta_{\bk+\bp}$ with $k_i+p_i$ $\mathrm{mod}\ N$. The restricted operator is denoted by $U_N$. We have
\begin{equation}\label{e:mainscal}
\frac{1}{T}\sum_{k=0}^{T-1}\langle U_N^k\psi, aU_N^k\psi\rangle = \langle \psi,A_T\psi\rangle = \sum_{i=1}^\nu \sum_{\bv\in \LL_N^d} \overline{\psi_i(\bv)}(A_T\psi)_i(\bv)
\end{equation}
for $A_T = \frac{1}{T}\sum_{k=0}^{T-1} U_N^{\ast\,k} aU_N^k$. Throughout the section, we focus on finding the asymptotics of $(A_T\psi)_i(\bv)$ by moving to Fourier space. The starting point is to expand the function $(A_T\psi)_i\in\ell^2(\LL_N^d)$ as a Fourier series
\begin{equation}\label{e:fseries}
(A_T\psi)_i(\bv) = \sum_{\br\in\LL_N^d} (\widehat{A_T\psi})_i(\br)e_{\br}^{(N)}(\bv)
\end{equation}
with $e_{\br}^{(N)}(\bv) := \frac{1}{N^{d/2}}\ee^{2\pi\ii \br\cdot\bv/N}$, the basis of $\ell^2(\LL_N^d)$ with Fourier coefficients $\widehat{\varphi}(\br) = \frac{1}{N^{d/2}}\sum_{\bk\in \LL_N^d} \ee^{-2\pi\ii \bk\cdot \br/N} \varphi(\bk)$.
We define
\begin{align*}
\mathscr{F}\otimes \mathrm{id}:\,&\ell^2(\LL_N^d)^{\nu}\to \ell^2(\LL_N^d)^{\nu}\\
& \begin{pmatrix} \psi_1\\\vdots\\\psi_{\nu}\end{pmatrix} \mapsto \begin{pmatrix} \widehat{\psi_1}\\\vdots\\\widehat{\psi_{\nu}}\end{pmatrix}\,,
\end{align*}

It follows from direct computation that for $\varphi,\psi \in \ell^2(\LL_N^d)^{\nu}$
\begin{equation}\label{e:convpro}
	\widehat{\varphi \cdot \psi}=\frac{1}{N^{d/2}}\,\widehat{\varphi} * \widehat{\psi},
\end{equation}
where $*$ denotes convolution of functions, $\widehat{\varphi} * \widehat{\psi}(\br)= \sum_{\bm
\in \LL_N^d}\widehat{\varphi}(\bm)\widehat{\psi}(\br-\bm)$.

\begin{lem}\label{lem:foufin}
The operator $U_N$ on $\ell^2(\LL_N^d)^{\nu}$ is unitarily equivalent to 
\[
\mathop\oplus_{\br\in \LL_N^d} \widehat{U}\Big(\frac{-\br}{N}\Big)
\]
on $\ell^2(\LL_N^d)^{\nu}$, where $\widehat{U}(\btheta)$ is the matrix \eqref{e:uhatheta}.
\end{lem}
\begin{proof}
We observe that for $\varphi\in\ell^2(\LL_N^d)^\nu$,
\[
U_N(\mathscr{F}\otimes \mathrm{id})\varphi =\begin{pmatrix} \sum_{j=1}^{\nu} \sum_{\bp} U_{1,j}(\bp)S_{\bp} \widehat{\varphi}_j\\ \vdots \\ \sum_{j=1}^{\nu} \sum_{\bp} U_{\nu,j}(\bp)S_{\bp} \widehat{\varphi}_j\end{pmatrix}
\]
and due to periodic boundary conditions,
\begin{equation} \label{e:sphat}
S_{\bp} \widehat{\varphi}_j = S_{\bp} \sum_{\bk\in \LL_N^d} \widehat{\varphi}_j(\bk) \delta_{\bk} = \sum_{\bk\in \LL_N^d} \widehat{\varphi}_j(\bk) \delta_{\bk+\bp} = \sum_{\bk\in \LL_N^d} \widehat{\varphi}_j(\bk-\bp)\delta_{\bk} \,.
\end{equation}

Let $\epsilon_{\bp}(\bk)=\ee^{2\pi\ii \bp\cdot \bk/N}=N^{d/2}e_{\bp}^{(N)}(\bk)$. Since 
\[
\widehat{\varphi}(\br-\bp)=\frac{1}{N^{d/2}} \sum_{\bk\in\LL_N^d} \ee^{-2\pi\ii(\br-\bp)\cdot \bk/N}\varphi(\bk) = (\widehat{\epsilon_{\bp}\varphi})(\br) \,,
\]
we get $S_{\bp} \widehat{\varphi}_j = \widehat{\epsilon_{\bp}\varphi_j}$, and thus
\[
[(\mathscr{F}^{-1}\otimes \mathrm{id})U_N(\mathscr{F}\otimes \mathrm{id})\varphi](\br) = \begin{pmatrix} \sum_j\sum_{\bp} U_{1,j}(\bp)\ee^{2\pi\ii \br\cdot \bp/N} \varphi_j(\br)\\ \vdots\\ \sum_j\sum_{\bp} U_{\nu,j}(\bp)\ee^{2\pi\ii \br\cdot \bp/N}\varphi_j(\br)\end{pmatrix} = \widehat{U}\Big(\frac{-\br}{N}\Big)\varphi(\br) \,. \qedhere
\]
\end{proof}

The same proof shows that the full operator $U$ on $\ell^2(\Z^d)^\nu$ is unitarily equivalent to multiplication by the matrix function $\widehat{U}(-\btheta)$ on $L^2(\T^d)^\nu$. Here we use the Fourier transform $\mathscr{F}\otimes \mathrm{id} :L^2(\T^d)\otimes \C^\nu \to \ell^2(\Z^d)\otimes \C^\nu$, taking $f\otimes \delta_i \mapsto \widehat{f} \otimes \delta_i$, $\widehat{f}(\bk) = \int_{\T^d} \ee^{-2\pi\ii\btheta\cdot \bk}f(\btheta)\dd\btheta$, with inverse $(a_{\bk})\otimes \delta_i \mapsto a\otimes \delta_i$, $a(\btheta) = \sum_{\bk} a_{\bk} \ee^{2\pi\ii \bk\cdot \btheta}$.

For example, for the Hadamard and Grover walks in \eqref{e:uhagro}, this gives
\[
\widehat{U}_{\mathrm{Had}}(\theta) = \frac{1}{\sqrt{2}}\begin{pmatrix} \ee^{2\pi \ii\theta}&\ee^{2\pi \ii\theta}\\ \ee^{-2\pi \ii\theta}&-\ee^{-2\pi \ii\theta}\end{pmatrix} \qquad \text{and} \qquad \widehat{U}_{\mathrm{Gro}}(\theta) = \frac{1}{3} \begin{pmatrix} -\ee^{2\pi\ii\theta} &2\ee^{2\pi\ii\theta}& 2\ee^{2\pi\ii\theta}\\ 2&-1&2\\ 2\ee^{-2\pi\ii\theta}& 2\ee^{-2\pi\ii\theta}& -\ee^{-2\pi\ii\theta}\end{pmatrix}.
\]
	
\subsection{Step 1. From the Heisenberg picture to a phase space operator.}\label{sec:step1} 
Back to \eqref{e:mainscal}, given an observable $a\in \ell^2(\LL_N^d)^{\nu}$, we need to find for $k\in\N$,
\[
(U^{\ast k}_N a U_N^k\psi)(\br) \,.
\]

Recall \eqref{e:ugen}. We have
\begin{align*}
		(U^k)_{i,j}&=\sum_{j_1,j_2,\ldots,j_{k-1}}U_{i,j_1}U_{j_1,j_2}\cdots U_{j_{k-1},j} \\
		&=\sum_{j_1,j_2,\ldots,j_{k-1}}\sum_{\bp_1,\bp_2,\ldots,\bp_k}U_{i,j_1}(\bp_1)U_{j_1,j_2}(
        \bp_2)\cdots U_{j_{k-1},j}(\bp_k)S_{\bp_1+\bp_2+\cdots+\bp_k}.
	\end{align*}
	
	In \eqref{e:fseries} we are interested in Fourier coefficients. Similarly to \eqref{e:sphat},
\[
\widehat{S_{\bp} g}(\br)=\sum_{\bk} g(\bk)\widehat{S_{\bp}\delta_{\bk}}(\br)=\sum_{\bk} g(\bk)\widehat{\delta_{\bk+\bp}}(\br)=\sum_{\bk}g(\bk)\frac{\ee^{\frac{-2\pi\ii (\bk+\bp)\cdot \br}{N}}}{N^{d/2}}= \ee^{-\frac{2\pi\ii \bp\cdot \br}{N}}\widehat{g}(\br)
\]
	
Therefore,
\begin{align*}
(\widehat{U_N^k\psi})_i(\br) &= \sum_{j=1}^{\nu} \widehat{U_{i,j}^k\psi_j}(\br) \\
&= \sum_{j=1}^{\nu}\sum_{j_1,\dots,j_{k-1}}\sum_{\bp_1,\dots,\bp_k} U_{i,j_1}(\bp_1)U_{j_1,j_2}(\bp_2)\cdots U_{j_{k-1},j}(\bp_k)\ee^{-2\pi\ii \br\cdot (\bp_1+\dots+\bp_k)/N} \widehat{\psi}_j(\br) \\
&=\sum_{j=1}^{\nu}\sum_{j_1,\dots,j_{k-1}} \widehat{U}\Big(\frac{\br}{N}\Big)_{i,j_1}\widehat{U}\Big(\frac{\br}{N}\Big)_{j_1,j_2}\cdots \widehat{U}\Big(\frac{\br}{N}\Big)_{j_{k-1},j} \widehat{\psi}_j(\br)= \Big[\widehat{U}\Big(\frac{\br}{N}\Big)^k\widehat{\psi}(\br)\Big]_i \,.
\end{align*}

On the other hand, by \eqref{e:convpro},
\begin{align*}
\widehat{(aU_N^k\psi)}_i(\br)&=\frac{1}{N^{d/2}}[\widehat{a}*\widehat{U_N^k\psi}]_i(\br)=	\frac{1}{N^{d/2}}\sum_{\bm \in \LL_N^d }\widehat{a}_i(\bm)(\widehat{U_N^k\psi})_i(\br-\bm)\\
&= \frac{1}{N^{d/2}}\sum_{\bm \in \LL_N^d }\widehat{a}_i(\bm)\Big[\widehat{U}\Big(\frac{\br-\bm}{N}\Big)^k\widehat{\psi}(\br-\bm)\Big]_i
\end{align*}
where $\br-\bm$ is taken $\mathrm{mod}\ N$.

Now $(U^\ast)_{i,j}=\sum_{\bp} \overline{U_{j,i}(\bp)}S_{-\bp}$ and $(\widehat{U}(\btheta)^\ast)_{i,j} =\sum_{\bp} \overline{U_{j,i}(\bp)}\ee^{2\pi\ii\btheta\cdot \bp}$. So we similarly get $(\widehat{U^{\ast k}\psi})_i(\br)=\sum_j\sum_{j_1,\dots,j_{k-1}}\sum_{\bp_1,\dots,\bp_k} \overline{U_{j_1,i}(\bp_1)\dots U_{j,j_{k-1}}(\bp_k)}\ee^{2\pi\ii \br \cdot (\bp_1+\dots+\bp_k)/N}\widehat{\psi}_j(\br)=[\widehat{U}(\frac{\br}{N})^{\ast k}\widehat{\psi}(\br)]_i$. We may thus expand as in \eqref{e:fseries} to get
\begin{align*}
[U_N^{*k}aU_N^k\psi]_i(\bv) &=\sum_{\br\in\LL_N^d} \widehat{[U_N^{\ast k}aU_N^k\psi]}_i(\br)e_{\br}^{(N)}(\bv)=\sum_{\br\in\LL_N^d}\Big[\widehat{U}\Big(\frac{\br}{N}\Big)^{\ast k}(\widehat{aU_N^k\psi})(\br)\Big]_ie_{\br}^{(N)}(\bv)\\
&= \sum_{\br\in\LL_N^d} \sum_{j=1}^{\nu} \widehat{U}\Big(\frac{\br}{N}\Big)^{\ast k}_{i,j}\widehat{(aU_N^k\psi)_j}(\br)e_{\br}^{(N)}(\bv) \\
&= \frac{1}{N^{d/2}}\sum_{\br\in\LL_N^d} \sum_{j=1}^{\nu} \widehat{U}\Big(\frac{\br}{N}\Big)^{\ast k}_{i,j}\sum_{\bm \in \LL_N^d }\widehat{a}_j(\bm)\Big[\widehat{U}\Big(\frac{\br-\bm}{N}\Big)^k\widehat{\psi}(\br-\bm)\Big]_je_{\br}^{(N)}(\bv) \\
&= \frac{1}{N^{d/2}}\sum_{\br,\bm\in\LL_N^d} \sum_{j,\ell=1}^{\nu} \widehat{U}\Big(\frac{\br}{N}\Big)^{\ast k}_{i,j}\widehat{a}_{j}(\bm)\widehat{U}\Big(\frac{\br-\bm}{N}\Big)_{j,\ell}^k\widehat{\psi}_{\ell}(\br-\bm) e_{\br}^{(N)}(\bv)\,.
\end{align*}

Let $E_1(\btheta), \dots, E_{\nu'}(\btheta)$ denote the distinct eigenvalues of $\widehat{U}(\btheta)$, with respective eigenprojections $P_{E_1}(\btheta),\dots P_{E_{\nu'}}(\btheta)$. Note that $\nu'$ may depend on $\theta$, but this will not affect the following calculations. We denote $\nu_1'=\nu'(\frac{\br+\bm}{N})$ and $\nu_2'=\nu'(\frac{\br}{N})$.

Since $\widehat{U}^k(\btheta)=\sum_{s=1}^{\nu^\prime}E_s(\btheta)^k P_s(\btheta)$ and $\widehat{U}(\btheta)^{\ast k}=\sum_{s=1}^{\nu^\prime}\overline{E_s(\btheta)}^k P_s(\btheta)$, 
we get that
\begin{multline*}
	[U_N^{\ast k}aU_N^k\psi]_i(\bv)
	=\frac{1}{N^{d/2}} \sum_{\br,\bm\in\LL_N^d} \sum_{j,\ell=1}^{\nu} \sum_{s,t}\left(\overline{E_s\Big(\frac{\br}{N}\Big)}E_t\Big(\frac{\br-\bm}{N}\Big)\right)^k P_{E_s}\Big(\frac{\br}{N}\Big)(i,j)\\
	\cdot \widehat{a}_{j}(\bm) P_{E_t}\Big(\frac{\br-\bm}{N}\Big)(j,\ell)\widehat{\psi}_{\ell}(\br-\bm) e_{\br}^{(N)}(\bv)\\
	=\sum_{\br,\bm\in\LL_N^d} \sum_{j,\ell=1}^{\nu} \sum_{s=1}^{\nu^\prime_1}\sum_{t=1}^{\nu_2'}\left(\overline{E_s\Big(\frac{\br+\bm}{N}\Big)}E_t\Big(\frac{\br}{N}\Big)\right)^k P_{E_s}\Big(\frac{\br+\bm}{N}\Big)(i,j)\\
    \cdot \widehat{a}_{j}(\bm) P_{E_t}\Big(\frac{\br}{N}\Big)(j,\ell)\widehat{\psi}_{\ell}(\br) \frac{e_{\br+\bm}^{(N)}(\bv)}{N^{d/2}}.
\end{multline*}
Define 
\begin{multline*}
	F_T(\bv,\br;i,\ell)
	= \sum_{\bm\in\LL_N^d} \sum_{j=1}^{\nu} \sum_{s=1}^{\nu_1'}\sum_{t=1}^{\nu^\prime_2} \frac{1}{T}\sum_{k=0}^{T-1} \left(\overline{E_s\Big(\frac{\br+\bm}{N}\Big)}E_t\Big(\frac{\br}{N}\Big)\right)^k \\ \cdot  P_{E_s}\Big(\frac{\br+\bm}{N}\Big)(i,j)\widehat{a}_{j}(\bm) P_{E_t}\Big(\frac{\br}{N}\Big)(j,\ell)e_{\bm}^{(N)}(\bv),
\end{multline*}
and
\begin{equation}\label{e:opf}
	[\text{Op}_N(F)\psi]_i(\bv)=\sum_{\br \in \LL_N^d}\sum_{\ell=1}^\nu F(\bv,\br;i,\ell)\widehat{\psi}_{\ell}(\br) e_{\br}^{(N)}(\bv)
\end{equation}
Then using $\frac{1}{N^{d/2}}e_{\bm+\br}^{(N)}=e_{\bm}^{(N)}e_{\br}^{(N)}$, we get 
\[
	\frac{1}{T}\sum_{k=0}^{T-1} [U_N^{*k}aU_N^k\psi]_i(\bv)=\sum_{\br \in \LL_N^d} \sum_{\ell=1}^\nu F_T(\bv,\br;i,\ell)\widehat{\psi}_{\ell}(\br) e_{\br}^{(N)}(\bv)=	[\text{Op}_N(F_T)\psi]_i(\bv) \,,
\]
and so, our main quantity \eqref{e:mainscal} takes the form
\begin{equation}\label{e:pseudocalc}
\frac{1}{T}\sum_{k=0}^{T-1}\langle U_N^k\psi,aU_N^k\psi\rangle=\sum_{i=1}^{\nu}\sum_{\bv\in\LL_N^d} \overline{\psi_i(\bv)}[\text{Op}_N(F_T)\psi]_i(\bv) = \langle \psi,\text{Op}_N(F_T)\psi\rangle \,.
\end{equation}

\subsection{Step 2. The time limit in a finite spatial box.}

\begin{lem}\label{lem:limpse}
For any fixed $N$, we have
\[
\lim_{T\to\infty} \langle \psi, \opn(F_T)\psi\rangle = \langle \psi,\opn(b)\psi\rangle 
\]
for
\[
b(\bv,\br,i,\ell)=\sum_{\bm\in\LL_N^d} \sum_{j=1}^{\nu} \sum_{s=1}^{\nu^\prime_1}\sum_{t=1}^{\nu_2'} \mathbf{1}_{S_{\br}}(\bm,s,t) P_{E_s}\Big(\frac{\br+\bm}{N}\Big)(i,j)\widehat{a}_{j}(\bm) P_{E_t}\Big(\frac{\br}{N}\Big)(j,\ell)e_{\bm}^{(N)}(\bv)\, ,
\]
and $S_{\br}=\{(\bm,s,t):E_s(\frac{\br+\bm}{N})=E_t(\frac{\br}{N})\}$.
\end{lem}
\begin{proof}
    Recalling \eqref{e:opf}, if $\psi = \sum_{\bk\in \Lambda} \psi(\bk)\delta_{\bk}$ has compact support $\Lambda\subset \Z^d$ independent of $N$, then $\langle \psi,\opn(F)\psi\rangle = \sum_{i=1}^\nu\sum_{\bv\in \Lambda} \overline{\psi_i(\bv)}\sum_{\br\in\LL_N^d}\sum_{\ell=1}^\nu F(\bv,\br;i,\ell)\widehat{\psi}_\ell(\br)e_{\br}^{(N)}(\bv)$. Also, $\widehat{\psi}_\ell(\br) = \sum_{\bk'\in\Lambda} \overline{e_{\br}^{(N)}(\bk')}\psi_\ell(\bk')$. Thus, 
\begin{equation}\label{e:bdop}
|\langle \psi, \opn(F)\psi\rangle|\le C_{\Lambda}\max_{i,\ell,\bv,\bk'} \Big|\frac{1}{N^d}\sum_{\br\in \LL_N^d}F(\bv,\br;i,\ell)\ee^{2\pi\ii \br \cdot (\bv-\bk')/N}\Big|\,.
\end{equation}

We thus examine the difference of symbols. Let $\alpha_{s,t,\br,\bm,N}= (\overline{E_s\big(\frac{\br+\bm}{N}\big)}E_t\big(\frac{\br}{N}\big)$. If $(\bm,s,t)\in S_{\br}$, then $ \frac{1}{T}\sum_{k=0}^{T-1} \alpha_{s,t,\br,\bm,N}^k =1$. Otherwise, the sum is geometric. Thus,
\begin{multline*}
|\langle \psi, \opn(F_T-b)\psi\rangle| \le C_\Lambda \max_{i,\ell,\bv,\bk'} \frac{1}{N^d} \sum_{\br,\bm\in\LL_N^d}\sum_{j=1}^\nu\sum_{s,t}\mathbf{1}_{S^c_r}(\bm,s,t)\\
\cdot \Big|\frac{1-\alpha_{s,t,\br,\bm,N}^{T}}{T[1-\alpha_{s,t,\br,\bm,N}]} P_{E_s}\Big(\frac{\br+\bm}{N}\Big)(i,j)\widehat{a}_j(\bm)P_{E_t}\Big(\frac{\br}{N}\Big)(j,\ell)\Big| \le \frac{C_{N,a}}{T}\,,
\end{multline*}
where $C_{N,a}$ is finite for any $N$ and is independent of $T$. Taking $T\to \infty$ yields the result.

\end{proof}

\subsection{Step 3. The zero order symbol yields a weighted average.} We deduce from \eqref{e:pseudocalc} and Lemma~\ref{lem:limpse} that
\begin{equation}\label{e:limpieces}
\lim_{T\to\infty} \frac{1}{T} \sum_{k=0}^{T-1}\langle U_N^k\psi,aU_N^k\psi\rangle = \langle \psi, \opn(b)\psi\rangle =  \langle \psi, \opn(b_0)\psi\rangle+ \langle \psi, \opn(b')\psi\rangle \,,
\end{equation}
where $b_0$ is the contribution of $\bm=\mathbf{0}$, i.e.
\[
b_0(\bv,\br,i,\ell)=\sum_{j=1}^{\nu} \sum_{s,t=1}^{\nu^\prime_2}\mathbf{1}_{S_{\br}}(\mathbf{0},s,t) P_{E_s}\Big(\frac{\br}{N}\Big)(i,j)\widehat{a}_{j}(\mathbf{0}) P_{E_t}\Big(\frac{\br}{N}\Big)(j,\ell)e_{\mathbf{0}}^{(N)}(\bv)\, ,
\]
\[
b'(\bv,\br,i,\ell)=\sum_{\bm\neq\mathbf{0}} \sum_{j=1}^{\nu} \sum_{s=1}^{\nu^\prime_1}\sum_{t=1}^{\nu_2'} \mathbf{1}_{S_{\br}}(\bm,s,t) P_{E_s}\Big(\frac{\br+\bm}{N}\Big)(i,j)\widehat{a}_{j}(\bm) P_{E_t}\Big(\frac{\br}{N}\Big)(j,\ell)e_{\bm}^{(N)}(\bv)\, .
\]

Let $\nu':=\nu_2'$. Since $\widehat{a}_{j}(\mathbf{0}) = \frac{1}{N^{d/2}} \sum_{\bk\in\LL_N^d} a_{j}(\bk) = N^{d/2}\langle a_{j}\rangle$ and $e_{\mathbf{0}}^{(N)}(\bv) = N^{-d/2}$,
\[
b_0(\bv,\br,i,\ell)=\sum_{j=1}^{\nu} \sum_{s,t=1}^{\nu^\prime} \mathbf{1}_{S_{\br}}(\mathbf{0},s,t) P_{E_s}\Big(\frac{\br}{N}\Big)(i,j)\langle a_{j}\rangle P_{E_t}\Big(\frac{\br}{N}\Big)(j,\ell)\, ,
\]

By definition \eqref{e:opf}, we get
\begin{align*}
[\opn(b_0)\psi]_i(\bv)&=\sum_{\br\in\LL_N^d}\sum_{\ell=1}^\nu\sum_{\substack{s,t=1\\  E_s(\frac{\br}{N})=E_t(\frac{\br}{N})}}^{\nu^\prime}\sum_{\substack{j=1}}^{\nu}   P_{E_s}\Big(\frac{\br}{N}\Big)(i,j)\langle a_{j}\rangle P_{E_t}\Big(\frac{\br}{N}\Big)(j,\ell) \widehat{\psi}_{\ell}(\br)e_{\br}^{(N)}(\bv) \\
 &=\sum_{j=1}^\nu\langle a_{j}\rangle\sum_{\br\in\LL_N^d}\sum_{\ell=1}^{\nu}\sum_{s=1}^{\nu'} P_{E_s}\Big(\frac{\br}{N}\Big)(i,j)P_{E_s}\Big(\frac{\br}{N}\Big)(j,\ell)\widehat{\psi}_{\ell}(\br)e_{\br}^{(N)}(\bv)\\
&=\sum_{j=1}^{\nu}\langle a_{j}\rangle\sum_{\br\in\LL_N^d}\sum_{s=1}^{\nu'} P_{E_s}\Big(\frac{\br}{N}\Big)(i,j)\Big[P_{E_s}\Big(\frac{\br}{N}\Big)\widehat{\psi}(\br)\Big]_{j}e_{\br}^{(N)}(\bv)\,.
\end{align*}
where $\nu'$ is the number of distinct eigenvalues of $\widehat{U}(\frac{\br}{N})$. Thus, we showed that

\begin{align}\label{eqn:u=0_simpli}
\langle \psi,\opn(b_0)\psi\rangle &= \sum_{i,j=1}^{\nu} \langle a_{j}\rangle \sum_{\br\in\LL_N^d} \sum_{s=1}^{\nu'} \sum_{\bv\in\LL_N^d} \overline{P_{E_s}\Big(\frac{\br}{N}\Big)(j,i)\psi_i(\bv)}e_{\br}^{(N)}(\bv)\Big[P_{E_s}\Big(\frac{\br}{N}\Big)\widehat{\psi}(\br)\Big]_{j} \nonumber\\
&=\sum_{j=1}^{\nu} \langle a_{j}\rangle \sum_{\br\in\LL_N^d} \sum_{s=1}^{\nu'}  \overline{\Big[P_{E_s}\Big(\frac{\br}{N}\Big)\widehat{\psi}(\br)\Big]_{j}}\Big[P_{E_s}\Big(\frac{\br}{N}\Big)\widehat{\psi}(\br)\Big]_{j} = \langle a\rangle_\psi
\end{align}    
as given in \eqref{e:avp}, where we used that $\widehat{f}(\br) = \sum_{\bv} f(\bv) \overline{e_{\br}^{(N)}(\bv)}$. 

\begin{rem}\label{rem:step3}
    Observe that in Steps 1-3 we have not used the regularity of the observable, in particular, the previous results hold for any $a^{(N)}$ with $\|a^{(N)}\|_\infty \le 1$.
\end{rem}

\subsection{Step 4. The higher order modes vanish asymptotically.}\label{subsec:step4}
We now show that $\opn(b')$ is asymptotically negligible. Recalling \eqref{e:bdop}, taking $F=b'$, it suffices to show that the term in the modulus goes to zero. This term is
\[
\frac{1}{N^d}\sum_{\br\in\LL_N^d} \sum_{\bm\neq \mathbf{0}}\sum_{j,s,w=1}^\nu \mathbf{1}_{S_r}(\bm,s,w)P_s\Big(\frac{\br+\bm}{N}\Big)(i,j) \widehat{a}_j(\bm) P_w\Big(\frac{\br}{N}\Big)(j,\ell)e_{\bm}^{(N)}(\bv)\ee^{2\pi\ii \br \cdot (\bv-\bk')/N}\,.
\]
Here, if $(u_s(\btheta),E_s(\btheta))$ is an orthonormal eigensystem of $\widehat{U}(\btheta)$, then $P_s(\btheta)f=\langle u_s(\btheta),f\rangle u_s(\btheta)$. In other words, $P_{E_t} = \sum P_w$, where the sum runs over the $u_w(\btheta)$ with eigenvalue $E_t(\btheta)$.

Let $A_{\bm} = \{(\br,s,w):E_s(\frac{\br+\bm}{N})-E_w(\frac{\br}{N})=0\}$. Then $(\bm,s,w)\in S_{\br}\iff (\br,s,w)\in A_{\bm}$ so the above is
\[
\frac{1}{N^d} \sum_{\bm\neq \mathbf{0}}\sum_{j=1}^\nu \widehat{a}_j(\bm)e_{\bm}^{(N)}(\bv) \sum_{\br\in\LL_N^d}\sum_{s,w=1}^{
\nu} \mathbf{1}_{A_{\bm}}(\br,s,w)P_s\Big(\frac{\br+\bm}{N}\Big)(i,j) P_w\Big(\frac{\br}{N}\Big)(j,\ell)\ee^{2\pi\ii \br \cdot (\bv-\bk')/N}\,.
\]

We will show that for regular $a^{(N)}$, there is $C_a$ independent of $N$ such that \begin{equation}\label{e:regood}
    \sum_{\bm}\sum_{j=1}^\nu |\widehat{a}_j(\bm)e_{\bm}^{(N)}(\bv)|\le C_a \,.
\end{equation} 
Since by \eqref{e:flo}, $\sup_{\bm\neq \mathbf{0}} \frac{|A_{\bm}|}{N^d}\to 0$, then using $|P_s(\theta)(i,j)|\le 1$, we will get $\langle \psi,\opn(b')\psi\rangle\to 0$. In view of \eqref{e:limpieces} and \eqref{eqn:u=0_simpli}, this will complete the proof. 

Let's prove \eqref{e:regood}. If $a_j(\bk) = f_j(\bk/N)$ with $f_j\in H^s(\T^d)$, $s>d/2$, where $\|f\|_{H^s}^2 = \sum_{\bk\in \Z^d} |\hat{f}(\bk)|^2\langle \bk\rangle^{2s}$, $\hat{f}(\bk) = \int_{\T^d} \ee^{-2\pi\ii \bk\cdot \bx}f(\bx)\,\dd \bx$ and $\langle \bk\rangle = \sqrt{1+|\bk|^2}$, then $\|\hat{f}\|_1:=\sum_{\bk} |\hat{f}(\bk)| \le C_s \|f\|_{H^s}$, where $C_s^2=\sum_{\bk}\langle \bk\rangle^{-2s}<\infty$ since $2s>d$. On the other hand, $f = \sum_{\bk} \hat{f}(\bk) \mathfrak{e}_{\bk}$ with $\mathfrak{e}_{\bk}(\bx) = \ee^{2\pi\ii \bk\cdot \bx}$, so $\widehat{a}_j(\bm) = \langle e_{\bm}^{(N)}, f_j(\cdot/N)\rangle_{\ell^2(\LL_N^d)} = \sum_{\bk\in \Z^d} \hat{f}_{j}(\bk) \langle e_{\bm}^{(N)},\mathfrak{e}_{\bk}(\cdot/N)\rangle_{\ell^2(\LL_N^d)} = \hat{f}_j(\bm) N^{d/2}$, since $\mathfrak{e}_{\bk}(\mathbf{n}/N) = N^{d/2} e_{\bk}^{(N)}(\mathbf{n})$.
This implies $\sum_{m}\sum_j |\widehat{a}_j(\bm)e_{\bm}^{(N)}(\bv)|\le \sum_{j} \|\hat{f}_j\|_1$, which is finite.

The second class of $a^{(N)}$ are the restrictions of $a\in\ell^1(\Z^d)^\nu$. Here, $a_j = \sum_{\mathbf{n}\in \Z^d}\sum_{q=1}^\nu c_{\mathbf{n},j}\delta_{\mathbf{n}}$ with $\sum_{n}|c_{\mathbf{n},j}|<\infty$. Then $\widehat{a}_j(\bm) = \sum_{\mathbf{n}\in \Z^d} c_{\mathbf{n},j} \overline{e_{\bm}^{(N)}(\mathbf{n})}$. This implies $|\widehat{a}_j(\bm)e_{\bm}^{(N)}(\bv)|\le \frac{1}{N^d}\|a\|_1$ for all $j$ implying \eqref{e:regood}.

Combining the $4$ steps together, we have finally completed the proof of Theorem~\ref{thm:pergra}.

\subsection{From FQE to PQE}\label{sec:fp}
Take $a\in \ell^2(\LL_N^d)^{\nu}$ such that $a_j = \phi \ \forall j$. Then \eqref{e:avp} becomes
\[
\langle a\rangle_\psi= \langle \phi \rangle \sum_{\br\in \LL_N^d} \sum_{s=1}^{\nu'} \Big\| P_{E_s}\Big(\frac{\br}{N}\Big)\widehat{\psi}(\br)\Big\|^2_{\C^{\nu}} = \langle \phi\rangle \sum_{\br\in \LL_N^d} \|\widehat{\psi}(\br)\|_{\C^{\nu}}^2 = \langle \phi \rangle \|\widehat{\psi}\|^2_{\ell^2(\LL_N^d)^{\nu}} = \langle \phi \rangle \|\psi\|^2 = \langle \phi\rangle
\]
since $\|\psi\|=1$. On the other hand,
\[
\frac{1}{T}\sum_{k=0}^{T-1} \langle U_N^k \psi, a U_N^k\psi\rangle = \frac{1}{T}\sum_{k=0}^{T-1} \sum_{i=1}^{\nu} \sum_{\br\in\LL_N^d} a_i(\br) |[U_N^k\psi]_i(\br)|^2 = \sum_{\br\in \LL_N^d} \phi(\br) \mu_{T,\psi}^N(\br) = \langle \phi\rangle_{T,\psi} \,.
\]
Thus, Theorem~\ref{thm:pergra} implies Theorem~\ref{thm:cri0}.

\subsection{From QE to weak convergence}\label{sec:weak}
Let $a_j=\phi$ and $a_i=0$ for $i\neq j$. Then $\frac{1}{T}\sum_{k=0}^{T-1} \langle U_N^k \psi, aU_N^k\psi\rangle = \sum_{\bm \in \LL_N^d} \phi(\bm) \frac{1}{T}\sum_{k=0}^{T-1}|[(U_N^k\psi)(\bm)]_j|^2\to \sum_{\bm\in\LL_N^d} \phi(\bm)\mu_{\psi,j}^N(\bm)=:\langle\phi\rangle_{\psi,j}$ as $T\to\infty$. Thus, \eqref{eqn:ergodic_general} implies that $|\langle \phi\rangle_{\psi,j}-\langle\phi\rangle c_{N,\psi,j}|\to 0$ as $N\to\infty$, where $c_{N,\psi,j}:= \sum_{\br\in\LL_N^d} \sum_{s=1}^{\nu'} |[P_{E_s}(\frac{\br}{N})\widehat{\psi}(\br)]_j|^2$.

Note that $\widehat{\psi}(\br)=\frac{1}{N^{d/2}}\sum_{\bk \in \Lambda} \ee^{-\frac{2 \pi \ii \bk\cdot \br}{N}}\psi(\bk)$, where $\Lambda$ is the compact support of $\psi$, which is independent of $N$. 
So by linearity,
\[
\Big[P_{E_s}\Big(\frac{\br}{N}\Big)\widehat{\psi}(\br)\Big]_j =\frac{1}{N^{d/2}}\sum_{\bk \in \Lambda}\Big[P_{E_s}\Big(\frac{\br}{N}\Big)\psi(\bk)\Big]_j\ee^{-\frac{2 \pi \ii \bk \cdot \br}{N}} \,.
\]
Hence,
\[
c_{N,\psi,j}=\frac{1}{N^{d}}\sum_{\br\in\LL_N^d}\sum_{s=1}^{\nu'}  \Bigg|\sum_{\bk \in \Lambda}\Big[P_{E_s}\Big(\frac{\br}{N}\Big)\psi(\bk)\Big]_j\ee^{-\frac{2 \pi \ii \bk \cdot \br}{N}}\Bigg|^2\,,
\]
This is a Riemann sum of a Riemann integrable functions (see the proof of \eqref{e:limterm} where this is explained in detail in a more complicated scenario) and therefore,
\[
c_{N\psi,j}\xrightarrow{N\to\infty}  \int_{\T^d}\sum_{s=1}^{\nu'}\Bigg|\sum_{\bk \in \Lambda}\big[P_{E_s}(\btheta)\psi(\bk)\big]_j\ee^{-2 \pi \ii \bk\cdot \btheta}\Bigg|^2\ \dd \btheta = c_{\psi,j} \,.
\]
Now take $\phi(\bm)=f(\bm/N)$ with $f\in H^s(\T^d)$ continuous. Then $\langle \phi\rangle \to \int_{\T^d} f(\bx)\,\dd\bx$. Thus, \eqref{eqn:ergodic_general} implies that
\[
\int_{\T^d} f(\bx)\,\dd\widetilde{\mu}_{\psi,j}^N(\bx)  = \langle \phi\rangle_{\psi,j} - \langle\phi\rangle c_{N,\psi,j} + \langle\phi\rangle c_{N,\psi,j} \to c_{\psi,j}\int_{\T^d}f(\bx)\,\dd \bx 
\]
for any continuous $f\in H^s(\T^d)$, in particular, for any smooth $f$. Combining \cite[Cor 15.3, Thm 13.34]{Klenke}, we deduce that $\widetilde{\mu}_{\psi,j}^N \xrightarrow{w} c_{\psi,j}\mu$ as $N\to\infty$.
This proves Proposition~\ref{prp:weaj}. The proof of Proposition~\ref{prp:wea} is the same, with simpler notations.

\subsection{Convergence in total variation} 
We turn to the proof of Proposition~\ref{prp:tvd}. Recall that $\delta(\mu,\mu')=\frac{1}{2}\sum_x |\mu(x)-\mu'(x)|$.

Suppose \eqref{e:mainlim} holds for all observables $\|\phi^{(N)}\|_\infty \le 1$. This means that $|\sum_{\bk\in \LL_N^d} \phi^{(N)}(\bk)(\mu_\psi^N(\bk)-\frac{1}{N^d})|\to 0$. Choosing $\phi^{(N)}(\bk) = \mathrm{sgn}(\mu_\psi^N(\bk)-\frac{1}{N^d})$, this yields $\sum_{\bk\in\LL_N^d} |\mu_\psi^N(\bk)-\frac{1}{N^d}|\to 0$, and thus $\delta(\mu_\psi^N,\mu^N)\to 0$.

Conversely, if $\delta(\mu_\psi^N,\mu^N)\to 0$, then by the triangle inequality, $|\sum_{\bk} \phi^{(N)}(\bk)(\mu_\psi^{N}(\bk)-\frac{1}{N^d})| \le 2 \delta(\mu_\psi^{(N)},\mu)\to 0$.

\subsection{Compact support of eigenvectors}\label{sec:compact}
We next prove Proposition~\ref{prp:flanoqe} and Lemma~\ref{lem:inui}. We start with the following result.

\begin{lem}\label{lem:flaco}
If the full quantum walk operator $U$ in \eqref{e:ugen}-\eqref{eqn:homogen_unitary}  has an eigenvalue $\lambda_0$, i.e. a flat band, then there exists a corresponding eigenvector for $\lambda_0$ of compact support.
\end{lem}
\begin{proof}
If $U$ has a flat band $\lambda_0$, then its Floquet matrix \eqref{e:uhatheta} has a corresponding eigenvector $f(\btheta)$, with $\widehat{U}(\btheta)f(\btheta) = \lambda_0 f(\btheta)$ for all $\btheta$ \cite{Tate2019}.
Since $U$ is of finite range, each entry of $\widehat{U}(\btheta)$ is a trigonometric polynomial. The same proof as \cite[Lemma 2.5]{SY} shows that we can choose $f(\btheta)$ to be a trigonometric polynomial in each entry, say $f_p(\btheta) = \sum_{\bm\in\Lambda_p}\alpha_p(\bm)\ee^{2\pi\ii \bm\cdot \btheta}$ with each $\Lambda_p$ finite. On the other hand, $\widehat{U} = (\mathscr{F}^{-1}\otimes\mathrm{id})U(\mathscr{F}\otimes \mathrm{id})$. Thus, $U(\mathscr{F}\otimes \mathrm{id})f = \lambda_0 (\mathscr{F}\otimes \mathrm{id}) f$, so $(\mathscr{F}\otimes \mathrm{id})f = \begin{pmatrix} \widehat{f}_1\\ \vdots\\ \widehat{f}_\nu\end{pmatrix}$ is an eigenvector of $U$. Here, $\widehat{f}(\bk)=\int_{\T^d} \ee^{-2\pi\ii\btheta\cdot\bk}f(\btheta)\dd\btheta$. So we get $\widehat{f_p}(\bk) = \alpha_p(\bk)$ if $\bk\in \Lambda_p$ and $0$ otherwise. Thus, $(\mathscr{F}\otimes \mathrm{id})f$ is an eigenvector of $U$ of compact support.
\end{proof}

\begin{proof}[Proof of Proposition~\ref{prp:flanoqe}]
By Lemma~\ref{lem:flaco}, if $U$ has a flat band, then $U$ has a corresponding eigenvector $\psi$ of compact support. Say as a vector, $\psi(\bk)=\mathbf{0}$ if $\bk\notin \Lambda$. 
We normalize $\|\psi\|=1$. Taking $N$ big enough, we have $\Lambda\subset \LL_N^d$ and $U_N \psi = U\psi = \lambda_0\psi$. Hence, $\mu_{T,\psi}^N(\br) = \frac{1}{T}\sum_{k=0}^{T-1} \sum_{i=1}^{\nu} |(U_N^k \psi)_i(\br)|^2 = \frac{1}{T}\sum_{k=0}^{T-1}\sum_{i=1}^{\nu} |\psi_i(\br)|^2 = \|\psi(\br)\|^2_{\C^\nu}$, since $|\lambda_0|=1$. 
Choose $\phi=\phi^{(N)}\equiv 1$ on $\Lambda$ and $\phi=0$ on $\LL_N^d\setminus \Lambda$. Then $\langle \phi\rangle_{T,\psi} = \sum_{\br} \phi(\br)\mu_{T,\psi}^N(\br) = \|\psi\|^2=1$. On the other hand, $\langle \phi\rangle = \frac{|\Lambda|}{N^d} = \frac{c}{N^d}$ with $c$ independent of $N$. Here $\phi$ is regular as it has a compact support (so is in $\ell^1$ ).
\end{proof}

\begin{proof}[Proof of Lemma~\ref{lem:inui}]
Consider the initial state $\psi=\delta_0\mathop\otimes f$, where $f=(\alpha,\beta,\gamma)$.
It is shown in \cite[Section 3]{IKS} that 
\[
\lim_{N\to \infty}\lim_{T\to\infty} \frac{1}{T}\sum_{l=1}^3 \sum_{t=0}^{T-1} P_N(0,t;l;\alpha,\beta,\gamma) = (5-2\sqrt{6})(1+|\alpha+\beta|^2+|\beta+\gamma|^2-2|\beta|^2)=:c_{\alpha,\beta,\gamma}
\]
where $P_N(0,t;l;\alpha,\beta,\gamma)$ is the probability of finding the particle with chirality $l$ at position $0$ and time $t$ on the cyclic lattice with $N$ sites. In other words, $\frac{1}{T}\sum_{l=1}^3 \sum_{t=0}^{T-1} P_N(0,t;l;\alpha,\beta,\gamma) = \mu_{T,\psi}^N(0)$. Take the observable $\phi = \delta_0$ on $\LL_N$. Then $\langle \phi\rangle_{T,\psi}=\mu_{T,\psi}^N(0)$ and $\langle \phi\rangle=\frac{1}{N^d}$. Thus, $\lim_{N\to\infty}|\lim_{T\to\infty} \langle \phi\rangle_{T,\psi}-\langle \phi\rangle|=c_{\alpha,\beta,\gamma}>0$ if $(\alpha,\beta,\gamma)\neq (\frac{1}{\sqrt{6}},\frac{-2}{\sqrt{6}},\frac{1}{\sqrt{6}})$ or a constant multiple of it.
\end{proof}

\section{Fine analysis in dimension one}\label{sec:proof1d}

In this section we prove particularly strong ergodicity results which hold for one-dimensional walks, and were collectively stated in Theorem~\ref{thm:main1d}. We start with (1)--(3). We will often use the fact that in dimension one, the Floquet eigenvalue functions are analytic in the full parameter space $\theta\in\R$, see \cite[Proposition 5.3]{Rainer2013}.

Our first statement is that in dimension one, \eqref{e:flo} implies \eqref{eqn:ergodic_general} not only for regular observables, but for bounded observables in general. 

\begin{thm}\label{thm:1D_l_infty}
Let $U$ be a discrete-time quantum walk on $\Z$ as in \eqref{e:ugen}-\eqref{eqn:homogen_unitary}. If $U$ satisfies \eqref{e:flo}, then it satisfies \eqref{eqn:ergodic_general} for any observable $a^{(N)}$ such that $\|a^{(N)}\|_\infty \le 1$. 
 \end{thm}

Our second statement is that in the absence of flat bands, there is always a subsequence satisfying \eqref{e:flo}, and over which we may apply the conclusion of Theorem~\ref{thm:1D_l_infty}.

\begin{thm}\label{thm:1D_subseq_general}
Consider a family of normal $\nu \times \nu$ matrices $\{\widehat{U}(\theta)\}_{\theta \in [0,1)}$ with analytic entries. Suppose $\widehat{U}(\theta)$ does not have a flat band. Then \eqref{e:flo} holds on a subsequence. That is, there exists a strictly increasing sequence $(N_n)$ such that 
    \begin{equation}\label{e:flo_subseq}
\sup_{m\neq 0} \frac{\#\{r\in \LL_{N_n}:E_s(\frac{r+m}{N_n})-E_w(\frac{r}{N_n})=0\}}{N_n}\to 0
\end{equation}
for all $s,w$. As a result, 
\eqref{eqn:ergodic_general} holds along $\LL_{N_n}$ for all $a^{(N_n)}$ such that $\|a^{(N_n)}\|_\infty \le 1$.
\end{thm}

In Proposition~\ref{prp:wasrem}, we will show that  Theorem~\ref{thm:1D_subseq_general} is in general not true for higher-dimensional quantum walks.

To prove the above theorems, we state two lemmas. The proofs of these lemmas are pushed to Section \ref{subsec:tech_lemmas}.

We say a function $E: \Omega \to \C$ is \emph{periodic} if $\exists\, \alpha>0$ such that $\{\theta:E(\theta+\alpha)\neq E(\theta)\}$ is a set of measure zero. For an analytic function on a domain $\Omega$, this is equivalent to saying that $E(\theta+\alpha)= E(\theta)$ for all $\theta \in \Omega$. The smallest such $\alpha$ is called the period of $E$.

\begin{lem}\label{lem:non-ergodic_condition}
 For a non-constant analytic function $E:\R \to \C$ with a rational period $\alpha$, define $F(\theta):=E(\theta+\varphi)$ for all $\theta$ and some fixed $0<\varphi\leq \alpha$.

Then $F(\theta)= E(\theta+u \alpha+\varphi)$ for all $u \in \Z$. Moreover, there exist sequences $N_n \to \infty$, $m_n \in \LL_{N_n}$ such that 
\begin{equation}\label{e:nonergodic_contra}
   \#\Big\{r\in \LL_{N_n}:E\Big(\frac{r+m_n}{N_n}\Big)-F\Big(\frac{r}{N_n}\Big)=0\Big\} \to \infty 
\end{equation} 
if and only if for all but finitely many $n$, $\frac{m_n}{N_n}$ is of the form $\varphi+\alpha u_n$ for some non-negative integer  $u_n < \frac{1}{\alpha}$. In particular,  if $\alpha \geq 1$, then $\frac{m_n}{N_n}=\varphi$ for all but finitely many $n$.
\end{lem}

\begin{lem}\label{lem:1D_subseq_2conditions}
Consider a family of $\nu \times \nu$ normal matrices $\{\widehat{U}(\theta)\}_{\theta \in \R}$ with analytic entries and eigenvalues $E_1(\theta), \ldots ,E_{\nu}(\theta)$. Fix $1 \leq s,w \leq \nu$. Then the following are equivalent
\begin{enumerate}[\rm(i)]
    \item \eqref{e:flo} fails for $s$ and $w$:    \begin{equation}\label{e:cond_nonergodic}
\sup_{\substack{m\neq 0 \\ m \in \LL_N}} \frac{\#\{r\in \LL_{N}:E_s(\frac{r+m}{N})-E_w(\frac{r}{N})=0\}}{N}\nrightarrow 0.
\end{equation}
\item 
\[
\limsup_{N\to\infty}\sup_{\substack{m\neq 0 \\ m \in \LL_{N}}} \#\Big\{r\in \LL_{N}:E_s\Big(\frac{r+m}{N}\Big)-E_w\Big(\frac{r}{N}\Big)=0\Big\} =\infty \,.
\]
\item There exists a rational $0<\varphi<1$ such that
\[
E_s(\theta+\varphi)=E_w(\theta) \text{ for all } \theta \in \R.
\]
\end{enumerate}

Additionally, if there exists a subsequence $(m_{n_k})$ such that the fraction in \eqref{e:cond_nonergodic} does not converge to zero and $\frac{m_{n_k}}{N_{n_k}} \rightarrow 0$, then $E_s(\theta),E_w(\theta)$ are flat bands and $E_s(\theta)\equiv E_w(\theta)\equiv c$ for some $c\in\C$.
\end{lem}

Now we prove Theorems~\ref{thm:1D_l_infty} and \ref{thm:1D_subseq_general}, assuming Lemmas \ref{lem:non-ergodic_condition} and \ref{lem:1D_subseq_2conditions}.
\begin{proof}[Proof of Theorem \ref{thm:1D_l_infty}]
As we mentioned in Remark~\ref{rem:step3},
   the only place in the proof of Theorem~\ref{thm:pergra} where we needed the observable assumption is in Step 4 (Section \ref{subsec:step4}). 
   Therefore, to prove the theorem, it is sufficient to show that for any $\|a^{(N)}\|_\infty \le 1$,
\begin{equation}\label{eqn:errorterm_1D}
 \frac{1}{N} \sum_{m\neq 0}\sum_{j=1}^\nu \widehat{a}_j(m)e_{m}^{(N)}(v) \sum_{r\in\LL_N}\sum_{s,w=1}^{
\nu} \mathbf{1}_{A_{m}}(r,s,w)P_s\Big(\frac{r+m}{N}\Big)(i,j) P_w\Big(\frac{r}{N}\Big)(j,\ell)\ee^{2\pi\ii r  \frac{v-k'}{N}}\,   
\end{equation}
goes to zero
for all $v,i,\ell$ and $k'$, with $A_m := \{(r,s,w):E_s(\frac{r+m}{N})-E_w(\frac{r}{N})=0\}$.

The assumption of the theorem states that (i) of Lemma \ref{lem:1D_subseq_2conditions} is violated for all $s,w$ and therefore by (ii), there exists a constant $C>0$ such that 
\[
\sup_{\substack{m\neq 0 \\ m \in \LL_{N}}} \#\Big\{r\in \LL_{N}:E_s\Big(\frac{r+m}{N}\Big)-E_w\Big(\frac{r}{N}\Big)=0\Big\} \leq C
\]
for all $s,w$ and $N$. Therefore, we have  $\#A_m \leq C \nu^2$ $\forall m\neq 0$. Since every entry of $P_s$ and $P_w$ is bounded by 1, the absolute value of \eqref{eqn:errorterm_1D} is thus bounded above by $\frac{1}{N} \sum_{j=1}^\nu\sum_{m \neq 0}|\widehat{a}_j(m)e_m^{(N)}(v)|C \nu^2=\frac{ C \nu^2}{N^{3/2}} \sum_{j=1}^\nu\sum_{m \neq 0}|\widehat{a}_j(m)|$.

Applying the Cauchy-Schwarz inequality to $\widehat{a}_j \in \ell^2(\LL_{N})$, we have
\[
\sum_{m \neq 0}|\widehat{a}_j(m)| \leq \sqrt{N}\bigg(\sum_{m \neq 0}|\widehat{a}_j(m)|^2 \bigg)^{1/2}\leq \sqrt{N}\|a_j\|_{2} \leq N\|a_j\|_\infty \leq N\,.
\]
Here, we used that $\|\widehat{a}_j \|_2=\|a_j\|_2$ and $\|a\|_\infty\le 1$.
As a result of the last inequality, we get $\frac{C \nu^2}{N^{3/2}} \sum_{j=1}^\nu\sum_{m \neq 0}|\widehat{a}_j(m)| \to 0$, completing the proof.
\end{proof}

Now, we prove Theorem \ref{thm:1D_subseq_general}.
\begin{proof}[Proof of Theorem \ref{thm:1D_subseq_general}]
   If \eqref{e:flo} holds on the full sequence, then there is nothing to prove.  Therefore, assume on the contrary that there exists $1 \leq s_1,w_1 \leq \nu$ such that (i) of Lemma \ref{lem:1D_subseq_2conditions} is satisfied. Then by (iii) of Lemma \ref{lem:1D_subseq_2conditions},  there exists a rational number $\varphi_1=\frac{p_1}{q_1} \in (0,1)$, where $\gcd(p_1,q_1)=1$, such that $E_{s_1}(\theta+\varphi_1)=E_{w_1}(\theta)$ for all $\theta$. Define
   \begin{equation}\label{e:defnM}
       M=\operatorname{lcm} \Big\{q': E_s\Big(\theta+\frac{p'}{q'}\Big)\equiv E_w(\theta) \text{ for some } s,w \text{ and } p'<q', \gcd(p',q')=1\Big\},
   \end{equation}
   and consider the sequence $N_n=nM+1$. We claim that $N_n$ satisfies \eqref{e:flo_subseq} for all $s,w$. Suppose not, then it follows from Lemma \ref{lem:1D_subseq_2conditions} that there exists $s_2,w_2$ and a rational number $\varphi_2=\frac{p_2}{q_2}$ such that $E_{s_2}(\theta+\varphi_2)\equiv E_{w_2}(\theta)$. 
   
   Furthermore, since $E_{s_2}(\theta)$ is an eigenvalue of a matrix-valued function on the torus, it is well-known that there exists some integer $k$ such that $E_{s_2}(\theta+k) \equiv E_{s_2}(\theta)$ (see for instance, \cite[Chapter 2]{Kato}), and therefore if $\alpha$ is the period of $E_{s_2}$, then $k=r\alpha$ for some $r\in \N$, implying that $\alpha$ is rational.
   Take $\alpha=\frac{p_\alpha}{q_\alpha}$ with $\gcd(p_\alpha,q_\alpha)=1$.
   We divide the proof into the following two possible scenarios.
   
   \noindent\textbf{Case 1: $\alpha \geq  1$.} Since $E_{s_2}(\theta+\varphi_2)\equiv E_{w_2}(\theta)$, by Lemma \ref{lem:non-ergodic_condition}, we get that $\frac{m_n}{N_n}=\varphi_2=\frac{p_2}{q_2}$ for all but finitely many $n$. This implies $q_2m_n=p_2nM+p_2$, which is impossible since only the last term is not divisible by $q_2$. Hence, $(N_n)$ satisfies \eqref{e:flo_subseq} for all $s,w$. 
   
   \noindent\textbf{Case 2: $\alpha < 1$.} Here, using Lemma~\ref{lem:non-ergodic_condition} again, we get $\frac{m_n}{N_n}=\varphi_2+u \alpha$ for some non-negative integer $u<\frac{1}{\alpha}$ for all but finitely many $n$.  Note that here $\alpha=\frac{p_\alpha}{q_\alpha}<1$ and $E_{s_2}(\theta+\alpha) \equiv E_{s_2}(\theta)$. Therefore, $q_\alpha$ belongs to the collection in \eqref{e:defnM}, implying $q_\alpha\mid M$. Similarly, $q_2\mid M$. So we get that $\frac{m_n}{N_n}=\frac{t}{M}$ for some $t<M$, implying $Mm_n=tM+t$, which is impossible. This contradiction shows that $N_n$ satisfies \eqref{e:flo_subseq} for all $s,w$.

   The last assertion is a direct consequence of Theorem \ref{thm:1D_l_infty}. 
\end{proof}

\subsection{Proof of technical lemmas}\label{subsec:tech_lemmas}

In this section, we prove Lemmas~\ref{lem:non-ergodic_condition} and \ref{lem:1D_subseq_2conditions}.

\begin{proof}[Proof of Lemma \ref{lem:non-ergodic_condition}]
Clearly, $E(\theta+u\alpha+\varphi)=E(\theta+\varphi)=F(\theta)$ for all $\theta \in \R$ and $u \in \Z$.
Therefore, when $\frac{m_n}{N_n}$ is of the form $\varphi+u\alpha$ for some $u \in \Z$,
\[
   \#\Big\{r\in \LL_{N_n}:E\Big(\frac{r+m_n}{N_n}\Big)-F\Big(\frac{r}{N_n}\Big)=0\Big\}=N_n\,.
\]

We now prove the converse. Suppose \eqref{e:nonergodic_contra} holds for the pair of sequences $(N_n)$ and $(m_n \in \LL_{N_n})$. We divide the proof into two cases, depending on the value of $\alpha$. 

\begin{case}
    $\alpha \geq 1$.
\end{case}

We make the following claim:
\begin{claim}
    $\frac{m_n}{N_n} \to \varphi$.
\end{claim}
\begin{proof}[Proof of claim.]
Suppose the claim is not true. Then there exists a subsequence of $(m_n)$, say $(m_{k_n})$, such that $\frac{m_{k_n}}{N_{k_n}} \to \widetilde{\varphi}$ for some $\widetilde{\varphi} \neq \varphi$. 

 Since $E$ and $F$ are analytic on the compact $[0,1]$, there is an open neighborhood $\Omega\supset [0,1]$,  $\Omega \subset \C$ such that $E,F$ have unique analytic continuations to $\overline{\Omega}$.
 
 Consider the analytic functions $g_n,g:\overline{\Omega}\to \C$ defined by 
    \begin{align*}
            g_n(\theta) &= E\Big(\theta+ \frac{m_{k_n}}{N_{k_n}}\Big)-F(\theta) \quad \forall n, \text{ and}\\
    g(\theta) &= E(\theta+\widetilde{\varphi})-F(\theta).
    \end{align*}

   Since $E$ and $F$ are analytic on the compact set $\overline{\Omega}$, the sequence $\{g_n\}$ has a uniform upper bound and is uniformly equicontinuous. Therefore, by the Arzelà–Ascoli theorem, $\{g_n\}$ has a subsequence $\{g_{n_k}\}$ that converges to $g$ uniformly in $\overline{\Omega}$. To keep the notation simple, we take this subsequence as $\{g_n\}$.
   
We claim that $g \equiv 0$ on $\Omega$. Suppose $g \not\equiv 0$, then by the identity theorem, $Z(g)= \{x \in \Omega:g(x)=0\}$ is a finite set. By taking a smaller set if necessary, assume that $g$ has no zeros on $\partial\Omega$. By Hurwitz's theorem \cite[Thm. 6.4.1(c)]{Simon2015}, it follows that for sufficiently large $n$, the number of zeros of $g_n$ in $\Omega$ is the number of zeroes of $g$ (counting multiplicity), which is finite and independent of $n$. This contradicts \eqref{e:nonergodic_contra}.

Hence, $g \equiv 0$, implying $E(\theta+\widetilde{\varphi})= F(\theta)= E(\theta+\varphi)$ for all $\theta \in [0,1]$. This implies $E(\theta) = E(\theta+\widetilde{\varphi}-\varphi)$ for all $\theta \in [\varphi,1+\varphi]$, and by the analyticity of $E$ it follows that $E(\theta) = E(\theta+\widetilde{\varphi}-\varphi)$ for all $\theta \in \R$. Since $\widetilde{\varphi}-\varphi \in (-\alpha,1)$ and the period of $E$ is at least one, we deduce that $\widetilde{\varphi}=\varphi$, a contradiction to our assumption that $\widetilde{\varphi} \neq \varphi$. This shows that for all subsequences $(m_{n_k})$ of $(m_{n})$, the fraction $\frac{m_{n_k}}{N_{n_k}} \to \varphi$, proving our claim.
\end{proof}

We now return to the proof of the lemma. 
Suppose there exists infinitely many $n$ such that $\frac{m_n}{N_n} \neq \varphi$. Moving to a subsequence, we assume without loss of generality that $\frac{m_n}{N_n} \neq \varphi$ for all $n$.
Note that if $E\big(\frac{r+m_n}{N_n}\big)-F\big(\frac{r}{N_n}\big)=0$ for some $r \in \LL_{N_n}$, then
\[
E\Big(\frac{r+m_n}{N_n}\Big)=F\Big(\frac{r}{N_n}\Big)=E\Big(\frac{r}{N_n}+\varphi\Big).
\]

Denoting $y_n=\frac{r+m_n}{N_n}$ and $h_n=\varphi-\frac{m_n}{N_n}$, we have $\frac{r}{N_n}+\varphi= y_n+h_n$ and therefore $E(y_n)=E(y_n+h_n)$. 
Noting that $y_n \in [0,2]$, this implies 
\begin{multline}\label{e:newboun}
 \#\Big\{r\in \LL_{N_n}:E\Big(\frac{r+m_n}{N_n}\Big)-F\Big(\frac{r}{N_n}\Big)=0\Big\} \leq \\   \#\left\{\theta\in [0,2]:E(\theta)-E(\theta+h_n)=0\right\}. 
\end{multline}

We claim that the upper bound of the above inequality is a constant independent of $h_n$. To see this, consider a domain $\Omega \supset [0,2]$ such that $E$ has an analytic continuation to $\overline{\Omega}$ and $E'$ has no zeroes on $\partial\Omega$ (recall that $E$ is non-constant analytic on $\R$). Define the functions $f_n:\Omega \to \C$ by
\[  f_n(\theta)=\frac{E(\theta)-E(\theta+h_n)}{h_n} \,.
\] 
Note that $f_n$ is well-defined as $h_n \neq 0$. Since $\frac{m_n}{N_n} \to \varphi$, it follows that $h_n \to 0$, and therefore $f_n$ converges to $E'$ uniformly on $\overline{\Omega}$. Since $E$ is non-constant analytic, $E'$ has only finitely many zeros in $\Omega$, say $n_0$. 
Therefore, by Hurtwitz's theorem, for $n$ sufficiently large, the number of zeroes of $f_n$ equals the number of zeros $E'$ counting multiplicity. This shows that \eqref{e:newboun} is bounded from above by $n_0$, contradicting \eqref{e:nonergodic_contra}. This completes the proof of this case.

\vskip2pt
\begin{case}
$\alpha \in (0,1)$.
\end{case}

As $\alpha$ is rational, let $\alpha=\frac{p}{q}$ with $\gcd(p,q)=1$. Consider the functions $\widetilde{E},\widetilde{F}:\R \to \C$ defined by
\begin{align*}
   \widetilde{E}(\theta) &=E\left(\frac{p}{q}\theta\right) \qquad  \text{ and } \qquad 
   \widetilde{F}(\theta) =F\left(\frac{p}{q}\theta\right). 
\end{align*}
Note that $\widetilde{E}$ and $\widetilde{F}$ have period 1, and that 
\begin{equation}\label{e:newdp}
\widetilde{E}\Big(\theta+\frac{q}{p} \varphi\Big)=E\Big(\frac{p}{q}\theta+\varphi\Big)=F\Big(\frac{p}{q}\theta\Big)=\widetilde{F}(\theta)
\end{equation} 
for all $\theta \in \R$. Also, observe that $\frac{q}{p}\varphi \leq 1$ since $\varphi \leq \frac{p}{q}$.

Suppose \eqref{e:nonergodic_contra} holds for some sequence $(N_n)$ and $(m_n \in \LL_{N_n} \setminus \{0\})$. Let $\widetilde{N}_n=pN_n$, $r'=qr \mod \widetilde{N}_n$ and $m_n^\prime=qm_n \mod \widetilde{N}_n$. Then $E\big( \frac{r+m_n}{N_n}\big)=\widetilde{E}\big( \frac{q}{p}\big(\frac{r+m_n}{N_n}\big)\big)=\widetilde{E}\big( \frac{r'+m_n^\prime}{\widetilde{N}_n}\big)$ and $F\big(\frac{r}{N_n} \big)=\widetilde{F}\big(\frac{q}{p} \frac{r}{N_n}\big)=\widetilde{F}\big(\frac{r'}{\widetilde{N}_n}\big)$. Let $\mathcal{I}=\{0,1,\dots,\lfloor \frac{q}{p}\rfloor\}$ and consider the map 
\begin{align*}
\LL_{N_n}&\to \mathcal{I}\times \LL_{\widetilde{N}_n} \\
r &\mapsto (u,r')
\end{align*}
where $u$ is the unique integer such that $ \frac{up}{q}\leq \frac{r}{N_n} 
< \frac{(u+1)p}{q}$. Let $r_1,r_2 \in \LL_{N_n}$ such that $ \frac{up}{q}\leq \frac{r_1}{N_n},\frac{r_2}{N_n}  
< \frac{(u+1)p}{q}$ for some integer $u$. Then the condition on $r_1,r_2$ is equivalent to $u \widetilde{N}_n \leq r_1',r_2' < (u+1) \widetilde{N}_n$ and therefore $r_1' \not\equiv r_2' \mod \widetilde{N}_n$ if $r_1 \neq r_2$. This shows that the map $r \mapsto (u,r')$ is one-one and therefore we obtain
    \begin{equation*}
   \sum_{u \in \mathcal{I}}\#\Big\{r'\in \LL_{\widetilde{N}_n} \cap [u\widetilde{N}_n,(u+1)\widetilde{N}_n):\widetilde{E}\Big(\frac{r'+m_n^\prime}{\widetilde{N}_n}\Big)-\widetilde{F}\Big(\frac{r'}{\widetilde{N}_n}\Big)=0\Big\} \geq cN_n=\frac{c}{p}\widetilde{N}_n \,\forall n.
\end{equation*} 
As the sets $\{\LL_{\widetilde{N}_n} \cap [u\widetilde{N}_n,(u+1)\widetilde{N}_n): u \in \mathcal{I}\}$ are disjoint, we get that
    \begin{equation}\label{eqn:non-ergodic_case2}
   \#\Big\{r'\in \LL_{\widetilde{N}_n} :\widetilde{E}\Big(\frac{r'+m_n^\prime}{\widetilde{N}_n}\Big)-\widetilde{F}\Big(\frac{r'}{\widetilde{N}_n}\Big)=0\Big\} \geq \frac{c}{p}\widetilde{N}_n \,\forall n.
\end{equation} 
 Therefore, by Case 1 and \eqref{e:newdp}, \eqref{eqn:non-ergodic_case2} holds if and only if $\frac{m_n^\prime}{\widetilde{N}_n}=\frac{q}{p}\varphi$ for all but finitely many $n$, i.e, $\frac{qm_n \mod \widetilde{N}_n}{\widetilde{N}_n}=\frac{q}{p}\varphi$ which is equivalent to $\frac{m_n}{N_n}=\varphi+u\frac{p}{q}$ for some $u$, for all but finitely many $n$. This completes the proof. 
\end{proof}

\begin{proof}[Proof of Lemma \ref{lem:1D_subseq_2conditions}]

Clearly, (i) implies (ii).  Next, to see that (iii) implies (i), note that if 
$E_s(\theta+\varphi)=E_w(\theta) \text{ for all } \theta \in [0,1)$ for some rational $\varphi=\frac{p}{q}$, then for $m=pN$ and $N_n=qN$, we get $E_s(\frac{r+m}{N_n}) = E_s(\frac{r}{N_n}+\frac{p}{q})=E_w(\frac{r}{N_n})$ for all $r$ and $N_n$, thus \eqref{e:cond_nonergodic} holds.

We now prove that (ii) implies (iii). As $\frac{m_n}{N_n}$ is bounded, we may assume there exist sequences $(m_n), (N_n) $  such that $\frac{m_n}{N_n} \rightarrow \varphi_0 \in [0,1)$ and 
\begin{equation}\label{e:infnu}
   \#\Big\{r\in \LL_{N_n}:E_s\Big(\frac{r+m_n}{N_n}\Big)-E_w\Big(\frac{r}{N_n}\Big)=0\Big\} \to \infty
\end{equation}

We know from \cite[Proposition 5.3]{Rainer2013} that the eigenvalues of an analytic family of normal matrices indexed by $\theta\in\R$, are analytic on $\R$.

 Consider then the analytic functions $g_n,g:\R \to \C$ given by
\begin{align}\label{e:defn_gn}
    g_n(\theta) &= E_s\Big(\theta+ \frac{m_n}{N_n}\Big)-E_w(\theta) \quad \forall n, \\
    g(\theta) &= E_s(\theta+\varphi_0)-E_w(\theta) \nonumber.
\end{align}
Following the argument of Claim 1 in the proof of Lemma~\ref{lem:non-ergodic_condition}, we get that $g \equiv0$ and thus $E_s(\theta+\varphi_0)\equiv E_w(\theta)$.
Furthermore, since $E_s$ and $E_w$ are eigenvalues of a matrix function on the torus, there exists an integer $k$ such that $E_s(\theta+k) \equiv E_s(\theta)$ \cite[Chapter 2]{Kato}. This implies that $E_s$ is either constant, or has a rational period $\alpha$ as before. If $E_s$ is constant, then (iii) follows immediately (take e.g. $\varphi=\frac{1}{2}$). So assume $E_s$ is not flat, so that $\alpha>0$.

Suppose that $\varphi_0=0$. Then $E_s(\theta)\equiv E_w(\theta)$ and $E_s(\theta+\alpha)=E_s(\theta)$. Hence, $E_w(\theta)\equiv E_s(\theta+\alpha)$. If $E_s$ is not flat, then applying Lemma~\ref{lem:non-ergodic_condition} with $\varphi:=\alpha$, we get that $\frac{m_n}{N_n}=\alpha+\alpha u_n\ge \alpha$ for all large $n$. This contradicts that $\frac{m_n}{N_n}\to 0$. Hence, $\frac{m_n}{N_n}\to 0$ implies that $E_s$ is flat.

As $E_s$ is assumed to be non-flat, we have $\varphi_0\neq 0$. We may have $\varphi_0>\alpha$, so we define $\varphi:=\varphi_0-u\alpha$, where the integer $u\ge 0$ is chosen so that $0< \varphi\le \alpha$. Then $E_s(\theta+\varphi)=E_s(\theta+\varphi_0)=E_w(\theta)$. It follows by Lemma~\ref{lem:non-ergodic_condition} that $\frac{m_n}{N_n}=\alpha u_n +\varphi$ for all large $n$. So $\frac{m_n}{N_n}=\alpha (u_n-u) +\varphi_0$ for all large $n$. But $\frac{m_n}{N_n}\to\varphi_0$. Hence, $u_n=u$ and $\frac{m_n}{N_n}=\varphi_0$ for all large $n$. In particular, $\varphi_0$ is rational. This completes the proof of (iii), and we proved the last assertion before.
\end{proof}

We have now proved all the results used in the proof of Theorems~\ref{thm:1D_l_infty}-\ref{thm:1D_subseq_general}.

\subsection{Equidistribution on subsets}
In this section, we assume that $U$ has no flat bands but violates \eqref{e:flo}. We are thus in scenario \eqref{e:cond_nonergodic}, and may define $M$ by \eqref{e:defnM}.  
We will prove that for every $1 \leq k \leq M$, over the sequence $N_n=Mn+k$, the quantum walk gets equidistributed on the subsets $B_i=\Z M+i$, $i=0,1,\ldots, M-1$. In other words,
\[
B_i = \begin{cases} \{i,M+i,\dots,nM+i\}& \text{if } 0\le i\le k-1,\\ \{i,M+i,\dots,(n-1)M+i\}&\text{if } k\le i \le M-1. \end{cases}
\]

The theorem is as follows (it corresponds to Theorem~\ref{thm:main1d} (4)).

\begin{thm}\label{thm:equidistribution}
    Let $U$ be a quantum walk \eqref{e:ugen}-\eqref{eqn:homogen_unitary} on $\Z$ with no flat bands, violating \eqref{e:flo}. Choose $N_n=nM+k$ for fixed $1 \leq k \leq M$. Let $\psi$ be an initial state of compact support with $\|\psi\|_2=1$. Then there exist constants $c_{u,\psi}^{(k)}, u=0,1,2,\ldots, M-1$, independent of $N_n$, such that for all observables $\phi^{(N_n)}$ with $\|\phi^{(N_n)}\|_\infty\le 1$, we have
    \begin{equation}\label{e:eqsub}
\lim_{N_n\to\infty}\Big|\lim_{T\to\infty}\langle \phi \rangle_{T,\psi} \\
- \sum_{u=0}^{M-1}c_{u,\psi}^{(k)}\langle \phi\rangle_{B_u}\Big|=0\,,
\end{equation}
where $\langle \phi\rangle_{B_i} = \frac{1}{|B_i|}\sum_{r\le |B_i|} \phi(rM+i)$ denotes the uniform average on $B_i$.

If $\gcd(M,k)=1$ (e.g. $k=1$), then the subsequence equidistributes,
\begin{equation}\label{e:equidissub}
\lim_{N\to\infty}\Big|\lim_{T\to\infty}\langle \phi \rangle_{T,\psi} \\
- \langle\phi\rangle\Big|=0\,.
\end{equation}
\end{thm}
The constants $c_{u,\psi}^{(k)}$ are in fact explicit, see \eqref{e:cupsi}.
\begin{proof}
     The proof proceeds along the same lines as the proof of Theorem \ref{thm:pergra} and we use the same notations. Recall that there, the set of nodes was divided into two classes, $m=0$ and $m \neq 0$, and it was proved in Section~\ref{subsec:step4} that the contribution of the second class is asymptotically negligible.
Here too, we divide  the choices of $m \in \LL_{N_n}$ into two sets
\begin{align*}
    A_1(n)&= \Big\{ 0, \frac{N_n}{\gcd(k,M)},\frac{2N_n}{\gcd(k,M)},\dots, \frac{(\gcd(k,M)-1)N_n}{\gcd(k,M)}\Big\} \quad \text{and}\\
    A_2(n)&= \LL_{N_n} \setminus A_1.
\end{align*}
Note that $A_1(n)$ is always a finite set with cardinality $\gcd(k,M)$, independent of $n$. Furthermore, if $\gcd(k,M)=1$, then $A_1(n)=\{0\}$ for all $n$. As in \eqref{e:limpieces}, 
\begin{equation}\label{e:newdeta}
\lim_{T\to\infty} \langle \phi\rangle_{T,\psi}=\langle \psi,\opn(b_1)\psi\rangle+\langle\psi,\opn(b')\psi\rangle\,,
\end{equation}
where $b_1$ and $b'$ now contain the contributions of $m\in A_1(n)$ and $m\in A_2(n)$, respectively.

We make the following claim.

\begin{claim*}
There exists a constant $C<\infty$ such that
\begin{equation}\label{e:a2good}
\sup_n\sup_{m\in A_2(n)} \#\Big\{r \in \LL_{N_n}:E_s\Big(\frac{r+m}{N_n}\Big)-E_w\Big(\frac{r}{N_n}\Big)=0\Big\} \leq C \qquad \forall \,s,w.
\end{equation} 
\end{claim*}
\begin{proof}
Suppose on the contrary that the claim is false. Then  there exists a sequence $m_n \in  A_2(n)$ such that 
\[
\limsup_n   \#\Big\{r\in \LL_{N_n}:E_s\Big(\frac{r+m_n}{N_n}\Big)-E_w\Big(\frac{r}{N_n}\Big)=0\Big\} =\infty,
\]
 Proceeding as in the proof of  Lemma~\ref{lem:1D_subseq_2conditions}, we can find a further subsequence, which we also call $(m_n)$, such that $\frac{m_n}{N_n} \to \varphi$ and $E_s(\theta+\varphi)=E_w(\theta)$ for all $\theta$. Moreover, since $\frac{m_n}{N_n} \to \varphi$ it follows by Lemma~\ref{lem:non-ergodic_condition} that $\frac{m_n}{N_n}=\varphi$ for all but finitely many $n$.
 By the definition of $M$ in \eqref{e:defnM}, $\varphi$ is of the form $\varphi=\frac{t}{M}$ for some positive integer $t<M$. Thus, $\frac{m_n}{N_n}=\frac{t}{M}$, which implies 
$m_nM=nMt+kt$. This happens only if $M$ divides $kt$, i.e. $t\cdot  \frac{k}{\gcd(M,k)}=\tilde{\rho}\cdot \frac{M}{\gcd(M,k)}$ for some integer $\tilde{\rho}$. Since $\gcd \left(\frac{k}{\gcd(M,k)},\frac{M}{\gcd(M,k)} \right)=1$ and $t \in \Z_+$, it follows that $\frac{k}{\gcd(M,k)}$ divides $\tilde{\rho}$. Therefore $t$ is of the form $t=\rho \frac{M}{\gcd(M,k)}$ for a positive integer $\rho$. Substituting this back to $\frac{m_n}{N_n}=\frac{t}{M}$, we get that $m_n=\rho \frac{N_n}{\gcd(M,k)}$, a contradiction, since $m_n \in A_2(n)$. This proves the claim.
\end{proof}

Now by repeating the calculations in the proof of Theorem~\ref{thm:1D_l_infty}, with $a_j=\phi$ $\forall j$, we deduce that the contribution of the terms $m \in A_2(n)$ vanish asymptotically. More precisely 
\begin{equation}\label{e:rema2}
  \lim_{N_n\to\infty} \langle \psi, \opn(b')\psi\rangle = 0 \,.
\end{equation}

Now, we look at the contribution of the terms in $A_1(n)$.

For a fixed $n$, and $m \in A_1(n)$ such that $\frac{m}{N_n}=\frac{\rho}{\gcd(M,k)}$, $\rho\in\N$, we have
\begin{align*} 
   \widehat{\phi}(m) &=\frac{1}{\sqrt{N_n}}\sum_{r \in \LL_{N_n}}\ee^{-2\pi \ii r (m/N_n)}\phi(r)=\frac{1}{\sqrt{N_n}}\sum_{r \in \LL_{N_n}}\ee^{-2\pi \ii r \frac{\rho}{\gcd(M,k)}}\phi(r)\\
   &=\frac{1}{\sqrt{N_n}}\sum_{u=0}^{M-1}\sum_{r\in B_u} \ee^{-2\pi \ii r \frac{\rho}{\gcd(M,k)}}\phi(r) =\frac{1}{\sqrt{N_n}}\sum_{u=0}^{M-1}\ee^{-2\pi \ii u\frac{\rho}{\gcd(M,k)}}\sum_{r \in B_u}\phi(r)\\
   &=\frac{1}{\sqrt{N_n}}\left(\sum_{u=0}^{k-1}\ee^{-2\pi \ii  u\frac{\rho}{\gcd(M,k)}}(n+1)\langle \phi \rangle_{B_u}+\sum_{u=k}^{M-1}\ee^{-2\pi \ii u\frac{\rho}{\gcd(M,k)}}n\langle \phi \rangle_{B_u}\right)\\
   &=\frac{1}{\sqrt{N_n}}\left(\sum_{u=0}^{M-1}\ee^{-2\pi \ii  u\frac{m}{N_n}}n\langle \phi \rangle_{B_u}+\sum_{u=0}^{k-1}\ee^{-2\pi \ii  u\frac{m}{N_n}}\langle \phi \rangle_{B_u}\right),
\end{align*}
where we used that $\ee^{-2\pi \ii r \frac{\rho}{\gcd(M,k)}} = \ee^{-2\pi \ii u \frac{\rho}{\gcd(M,k)}}$ for $r\in B_u$, and the second-to-last equality follows from noting that $B_u$ has $(n+1)$ elements if $0\leq u \leq k-1$ and has $n$ elements otherwise. Here, the convention is that an empty sum is defined as zero. 
Therefore, for $m\in A_1(n)$ and $v \in \LL_{N_n}$, we have
\[
   \widehat{\phi}(m) e_{m}^{(N_n)}(v) e_r^{(N_n)}(v)
   =\frac{1}{N_n}\sum_{u=0}^{M-1}\ee^{\frac{-2\pi \ii  u m}{N_n}}n\langle \phi \rangle_{B_u}e_{m+r}^{(N_n)}(v)+\frac{1}{N_n}\sum_{u=0}^{k-1}\ee^{\frac{-2\pi \ii  u m}{N_n}}\langle \phi \rangle_{B_u}e_{m+r}^{(N_n)}(v),
\]
where we used that $e_{m}^{(N_n)}(v)  e_r^{(N_n)}(v)=\frac{1}{\sqrt{N_n}} e_{m+r}^{(N_n)}(v)$ for all $v \in \LL_{N_n}$. It follows that
\begin{align*}
	&\sum_{v \in \LL_{N_n}} \overline{\psi_i(v)}\widehat{\phi}(m) e_{m}^{(N_n)}(v) e_r^{(N_n)}(v) \\
    &=\frac{n}{N_n}\sum_{v \in \LL_{N_n}}\sum_{u=0}^{M-1}\overline{\psi_i(v)}\ee^{\frac{-2\pi \ii  u m}{N_n}}\langle \phi \rangle_{B_u}e_{m+r}^{(N_n)}(v)+\frac{1}{N_n}\sum_{v \in \LL_{N_n}}\sum_{u=0}^{k-1}\overline{\psi_i(v)}\ee^{\frac{-2\pi \ii  u m}{N_n}}\langle \phi \rangle_{B_u}e_{m+r}^{(N_n)}(v)\\
    &=\frac{n}{N_n}\sum_{u=0}^{M-1}\overline{\widehat{\psi_i}(r+m)}\ee^{\frac{-2\pi \ii u m}{N_n}}\langle \phi \rangle_{B_u}+\frac{1}{N_n}\sum_{u=0}^{k-1}\overline{\widehat{\psi_i}(r+m)}\ee^{\frac{-2\pi \ii  um}{N_n}}\langle \phi \rangle_{B_u},
\end{align*}
where the last equality follows since $\widehat{\psi}(r+m)=\sum_{v \in \LL_N} \overline{e_{r+m}^{(N)}(v)}\psi(v)$.

Now by \cite[Proposition 5.3]{Rainer2013}, the eigenvalues of an analytic family of normal matrices indexed by $\theta\in\R$, are analytic on $\R$. This implies that there exists a simply connected domain $D$ containing $[0,1]$ such that $E_s,E_w$ have unique analytic extension to $D$. It is well-known (see \cite[Chapter II.4]{Kato}) that the multiplicity of an eigenvalue $E_s(\theta)$ is independent of $\theta$ in simply connected domain, and as a consequence, the number of distinct eigenvalues of $\Utheta$ is independent of $\theta$. Using this observation, \eqref{e:opf} and Lemma~\ref{lem:limpse}, we obtain that the contribution of $m \in A_1(n)$ in $\langle \psi,\opn(b_1)\psi\rangle$  (with $a_j=\phi$ $\forall \, j$) is 
\begin{multline*}
    	\sum_{r,v \in \LL_{N_n}} \sum_{i,j,\ell=1}^{\nu} \sum_{s,w=1}^{\nu^\prime} \mathbf{1}_{S_r}(m,s,w) P_{E_s}\Big(\frac{r+m}{N_n}\Big)(i,j)\overline{\psi_i(v)}\widehat{\phi}(m) \\ P_{E_w}\Big(\frac{r}{N_n}\Big)(j,\ell)e_{m}^{(N)}(v)\widehat{\psi_{\ell}}(r) e_r^{(N_n)}(v),
\end{multline*}
where $\nu^\prime$ is the number of distinct eigenvalues of $\Utheta$.

 Substituting the sum over $v$ that we just computed previously and simplifying, we get that the contribution from $m \in A_1(n)$ in $\langle \psi,\opn(b_1)\psi\rangle$ is 
\begin{multline}
    \frac{n}{N_n}\sum_{r \in \LL_{N_n}} \sum_{j=1}^\nu\sum_{s,w=1}^{\nu^\prime}  \mathbf{1}_{S_r}(m,s,w) \Bigg[\overline{P_{E_s}\Big(\frac{r+m}{N_n}\Big)\widehat{\psi}(r+m)}\Bigg]_{j} \Bigg[P_{E_w}\Big(\frac{r}{N_n}\Big)\widehat{\psi}(r)\Bigg]_{j} \sum_{u=0}^{M-1}\ee^{\frac{-2\pi \ii u m}{N_n}} \langle \phi \rangle_{B_u}\\
       +\frac{1}{N_n} \sum_{r \in \LL_{N_n}} \sum_{j=1}^\nu\sum_{s,w=1}^{\nu^\prime}  \mathbf{1}_{S_r}(m,s,w) \Bigg[\overline{P_{E_s}\Big(\frac{r+m}{N_n}\Big)\widehat{\psi}(r+m)}\Bigg]_{j} \Bigg[P_{E_w}\Big(\frac{r}{N_n}\Big)\widehat{\psi}(r)\Bigg]_{j}\sum_{u=0}^{k-1}\ee^{\frac{-2\pi \ii u m}{N_n}} \langle \phi \rangle_{B_u}\label{eqn:contri_u_twoterms}
\end{multline}

 Summarizing, by \eqref{e:newdeta}, \eqref{e:rema2} and \eqref{eqn:contri_u_twoterms}, we have shown so far that 
\begin{equation}\label{e:preresult}
\lim_{N_n\to\infty}\Big|\lim_{T\to\infty} \langle \phi\rangle_{T,\psi} - \sum_{u=0}^{M-1} c_{u,\psi}^{(N_n,k)}\langle \phi\rangle_{B_u}\Big|= 0
\end{equation}
for some complicated coefficients $c_{u,\psi}^{(N_n,k)}$ which we now need to simplify.

We first show that the second term of \eqref{eqn:contri_u_twoterms}  goes to zero, as $n \to \infty$. For this, note that for any initial state $\psi=\sum_{p \in \Lambda} \psi(p)\delta_p$ with compact $\Lambda \subseteq \Z$, we have $|\widehat{\psi}(r)|=|\frac{1}{\sqrt{N_n}}\sum_{k \in \LL_{N_n}}\ee^{-2\pi \ii r\frac{k}{N_n}} \psi(k)| \leq \frac{|\Lambda|}{\sqrt{N_n}}$. Therefore,  $ \Big|\Big[P_{E_w}\big(\frac{r}{N_n}\big)\widehat{\psi}(r)\Big]_{j}\Big|\leq \nu  \|\widehat{\psi}\|_\infty \leq \frac{\nu|\Lambda|}{\sqrt{N_n}}$, and the same upper bound holds for $\Big|\Big[\overline{P_{E_s}\big(\frac{r+m}{N_n}\big)\hat{\psi}(r+m)}\Big]_{j}\Big|$. Furthermore, all the terms inside the summation in \eqref{eqn:contri_u_twoterms} are bounded above. It follows that the second term in \eqref{eqn:contri_u_twoterms} is bounded above by $\frac{1}{N_n} \sum_{r \in \LL_{N_n}} \frac{c}{N_n} \leq \frac{c}{N_n}$ for some constant $c$. As a result, the second term in \eqref{eqn:contri_u_twoterms} goes to zero as $n \to \infty$. 

Now, we look at the first term in \eqref{eqn:contri_u_twoterms}. Note that 
\begin{align}
&\sum_{j=1}^\nu\Bigg[\overline{P_{E_s}\Big(\frac{r+m}{N_n}\Big)\widehat{\psi}(r+m)}\Bigg]_{j} \Bigg[P_{E_w}\Big(\frac{r}{N_n}\Big)\widehat{\psi}(r)\Bigg]_{j} \nonumber \\
    &= \sum_{j=1}^\nu \sum_{i_1,i_2=1}^\nu \overline{P_{E_s}\Big(\frac{r+m}{N_n}\Big)(j,i_1)}\overline{\widehat{\psi}_{i_1}(r+m)} P_{E_w}\Big(\frac{r}{N_n}\Big)(j,i_2)\widehat{\psi}_{i_2}(r) \nonumber\\
    &= \sum_{i_1,i_2=1}^\nu P_{E_s}\Big(\frac{r+m}{N_n}\Big)P_{E_w}\Big(\frac{r}{N_n}\Big)(i_1,i_2)\overline{\widehat{\psi}_{i_1}(r+m)}\widehat{\psi}_{i_2}(r). \label{eqn:productprojection}
\end{align}

Also,
\begin{align}
    \overline{\widehat{\psi}_{i_1}(r+m)}\widehat{\psi}_{i_2}(r)&=\frac{1}{N_n} \sum_{k_1,k_2 \in \LL_{N_n}} \ee^{2\pi \ii k_1(\frac{r+m}{N_n})}\overline{\psi_{i_1}(k_1)}\ee^{-2\pi \ii k_2\frac{r}{N_n}}\psi_{i_2}(k_2) \nonumber\\
    &=\frac{1}{N_n} \sum_{k_1,k_2 \in \Lambda} \ee^{2\pi \ii (k_1-k_2)(\frac{r}{N_n})}\overline{\psi_{i_1}(k_1)}\ee^{2\pi \ii k_1\frac{m}{N_n}}\psi_{i_2}(k_2). \label{eqn:product_psi}
\end{align}

The coefficient of $\langle \phi \rangle_{B_u}$ in the collective contribution of $A_1(n)$ is obtained by summing over all $m \in A_1(n)$. Substituting back \eqref{eqn:productprojection} and \eqref{eqn:product_psi} in \eqref{eqn:contri_u_twoterms}, we have the coefficient of $\langle \phi \rangle_{B_u}$ as 
\begin{align*}
&\frac{n}{N_n}\sum_{m \in A_1(n)}\sum_{r \in \LL_{N_n}} \sum_{j=1}^\nu\sum_{\substack{s,w=1 \\ E_s(\frac{r+m}{N_n})\equiv E_w(\frac{r}{N_n})}}^{\nu^\prime}  \Bigg[\overline{P_{E_s}\Big(\frac{r+m}{N_n}\Big)\widehat{\psi}(r+m) }\Bigg]_{j}\Bigg[P_{E_w}\Big(\frac{r}{N_n}\Big)\widehat{\psi}(r)\Bigg]_{j} \ee^{\frac{-2\pi \ii u m}{N_n}} \\
&=\frac{n}{N_n^2}\sum_{m\in A_1(n)}\sum_{r \in \LL_{N_n}} \sum_{i_1,i_2=1}^\nu\sum_{k_1,k_2 \in \Lambda}\sum_{\substack{s,w=1\\ E_s(\frac{r+m}{N_n})\equiv E_w(\frac{r}{N_n})}}^{\nu^\prime}  P_{E_s}\Big(\frac{r+m}{N_n}\Big)P_{E_w}\Big(\frac{r}{N_n}\Big)(i_1,i_2) \\
&\hspace{30mm}\times \ee^{2\pi \ii (k_1-k_2)(\frac{r}{N_n})}\overline{\psi_{i_1}(k_1)}\psi_{i_2}(k_2) \ee^{-2\pi \ii (u-k_1)m /N_n} \,.
\end{align*}
Now for each $m \in A_1(n)$, we have $\{r: E_s(\frac{r+m}{N_n})\equiv E_w(\frac{r}{N_n})\} = \{r: E_s(\frac{r}{N_n}+\varphi)\equiv E_w(\frac{r}{N_n})\}$, where $\varphi:=\frac{\rho}{\gcd(M,k)}$ is some rational number independent of $N_n$. For fixed $s,w$, if the cardinality of this set is bounded independently of $N_n$, then by the triangle inequality,
\begin{multline*}
\frac{n}{N_n^2}\sum_{r \in \LL_{N_n}} \sum_{i_1,i_2=1}^\nu\sum_{k_1,k_2 \in \Lambda}  \mathbf{1}_{S_r}(m,s,w) P_{E_s}\Big(\frac{r+m}{N_n}\Big)P_{E_w}\Big(\frac{r}{N_n}\Big)(i_1,i_2) \ee^{2\pi \ii (k_1-k_2)(\frac{r}{N_n})}\\
\times \overline{\psi_{i_1}(k_1)}\psi_{i_2}(k_2) \ee^{-2\pi \ii (u-k_1)m /N_n}\to 0.
\end{multline*}
So we focus on the set of $m\in A_1(n)$ such that $\limsup_n\#\{r: E_s(\frac{r}{N_n}+\varphi)\equiv E_w(\frac{r}{N_n})\}= \infty$. As in Lemma~\ref{lem:1D_subseq_2conditions}, this holds if and only if $E_s(x+\varphi) \equiv E_w(x)$. Therefore, we may replace the condition in the summation of $s,w$, $ E_s(\frac{r+m}{N_n})\equiv E_w(\frac{r}{N_n})$ with $ E_s(x+\varphi)\equiv E_w(x)$, which does not depend on $r$. So the only $r$-dependent term is
\[
\frac{1}{N_n}\sum_{r \in \LL_{N_n}}   P_{E_s}\Big(\frac{r}{N_n}+\varphi\Big)P_{E_w}\Big(\frac{r}{N_n}\Big)(i_1,i_2) \ee^{2\pi \ii (k_1-k_2)(\frac{r}{N_n})} \ee^{-2\pi \ii (u-k_1)\varphi},
\] 
which, as a Riemann sum, converges to
\[
\int_{0}^1P_{E_s}(\theta+\varphi)P_{E_w}(\theta)(i_1,i_2) e^{2\pi \ii (k_1-k_2)\theta} \ee^{-2\pi \ii (u-k_1)\varphi} \dd \theta,
\]
for all fixed $\varphi,i_1,i_2,k_1,k_2,s$ and $w$ (for this convergence, one can either use the argument in the later proof of \eqref{e:limterm}, or the fact that eigenprojections in dimension one can be chosen to be globally analytic, in particular continuous, see \cite[Prp. 5.3]{Rainer2013}). Now, noting that as $n \to \infty$, $\frac{n}{N_n} \to \frac{1}{M}$, we get that the coefficient of $\langle \phi \rangle_{B_u}$ converges to
\begin{multline*}
    c_{u,\psi}^{(k)}=  \frac{1}{M} \sum_{k_1,k_2 \in \Lambda}\sum_{i_1,i_2=1}^\nu \overline{\psi_{i_1}(k_1)}\psi_{i_2}(k_2) \sum_{\substack{0 \leq \varphi<1\\ \varphi\gcd(M,k) \in \Z}} \sum_{\substack{s,w=1 \\ E_s(x+\varphi)\equiv E_w(x)}}^{\nu^\prime} \ee^{2\pi \ii (u-k_1)\varphi}\\
    \cdot  \int_{0}^1 P_{E_s}\left(\theta+\varphi\right) P_{E_w}(\theta)(i_1,i_2)  \ee^{2\pi \ii (k_1-k_2)\theta} \dd\theta.\label{eqn:cm_psi}
\end{multline*}
where we recall that $\varphi=\frac{\rho}{\gcd(M,k)}$, so the summation over $0\le \rho\le \gcd(M,k)-1$ transforms to a summation over $\varphi$ with $\varphi \gcd(M,k) \in \Z$ and $0\le \varphi<1$.

Note that when $\varphi=0$, the condition over $s,w$ becomes $E_s(x)\equiv E_w(x)$, implying $P_{E_s}(\theta)P_{E_w}(\theta)=P_{E_s}(\theta)$. Since, $\sum_s P_{E_s}(\theta)=\mathrm{Id}_\nu$, we have $\sum_s P_{E_s}(\theta)(i_1,i_2)=1$ if $i_1=i_2$ and zero otherwise. Also note that $\int_0^1\ee^{2\pi \ii (k_1-k_2)\theta}=0 $ if $k_1\neq k_2$ and the term $e^{2\pi \ii (k_1-u)\varphi}=1$ as $\varphi=0$. As a result, the term corresponding to $\varphi=0$ becomes
\[
\frac{1}{M} \sum_{k_1 \in \Lambda}\sum_{i_1=1}^\nu  \overline{\psi_{i_1}(k_1)}\psi_{i_1}(k_1)  =\frac{1}{M}
\]
since $\|\psi\|_2=1$, and $c_{u,\psi}^{(k)}$ can therefore also be written as
\begin{multline}\label{e:cupsi}
    c_{u,\psi}^{(k)}= \frac{1}{M}+ \frac{1}{M} \sum_{k_1,k_2 \in \Lambda}\sum_{i_1,i_2=1}^\nu  \overline{\psi_{i_1}(k_1)}\psi_{i_2}(k_2) \sum_{\substack{0 < \varphi<1\\ \varphi\gcd(M,k) \in \Z}} \\
     \sum_{\substack{s,w=1 \\ E_s(x+\varphi)\equiv E_w(x)}}^{\nu^\prime}\ee^{2\pi \ii (k_1-u)\varphi} \int_{0}^1 P_{E_s}\left(\theta+\varphi\right) P_{E_w}(\theta)(i_1,i_2)  \ee^{2\pi \ii (k_1-k_2)\theta} \dd\theta.
\end{multline}
Using  \eqref{e:preresult} and  $|\sum_{u=0}^{M-1} c_{u,\psi}^{(N_n,k)}\langle \phi\rangle_{B_u} - \sum_{u=0}^{M-1} c_{u,\psi}^{(k)}\langle \phi\rangle_{B_u}|\to 0$, we obtain \eqref{e:eqsub}. Statement \eqref{e:equidissub} also follows, see Remark~\ref{rem:cupsi}.
\end{proof}

\begin{rem}\label{rem:cupsi}
\begin{enumerate}[(i)]
    \item  For all $k$ and $\psi$, $c_{u, \psi}^{(k)}=c_{u+\gcd(M,k),\psi}^{(k)}$. This is because the only term in \eqref{e:cupsi} that is a function of $u$ is $\ee^{2\pi \ii (k_1-u)\varphi}$ and $\ee^{2\pi \ii (k_1-u)\varphi}=\ee^{2\pi \ii (k_1-u-\gcd(M,k))\varphi}$ since $\varphi \gcd(M,k)\in\Z$.
    \item An immediate consequence of the above observation is that if $\gcd(k,M)=1$, then $c_{u,\psi}$ is the same for all values of $u$ and equal to $\frac{1}{M}$, which means that the walk gets equidistributed for the sequence $N_n=Mn+k$.
    \item For a fixed initial state $\psi$, the semiclassical limits of the walk are completely determined by the coefficients $c_{u,\psi}^{(k)}$. Note that the only term in \eqref{e:cupsi} that contains $k$ is $\gcd(M,k)$.
    As a result, the number of semiclassical limits of the walk is at most the number of possible choices of $\gcd(M,k)$, which is exactly the number of factors of $M$. If the prime factorization of $M$ is $M=p_1^{\alpha_1}p_2^{\alpha_2}\cdots p_\ell^{\alpha_\ell}$, then the number of distinct semiclassical limits of $\mu_{T,\psi}^{N}$ is thus at most $(\alpha_1+1)(\alpha_2+1)\cdots (\alpha_\ell+1)$.
\end{enumerate}   
\end{rem}

\begin{exa}\label{exa:39}
  Consider the quantum walk on $\Z$ given by the unitary matrix:
\[
U=\begin{pmatrix}
    S_1 & 0 \\
    0 & S_{-2}
\end{pmatrix}
\begin{pmatrix}
    1 & 0 \\
    0 & 1
\end{pmatrix}=
\begin{pmatrix}
    S_1 & 0 \\
    0 & S_{-2}
\end{pmatrix}.
\]
This walk violates \eqref{e:mainlim} for bounded observables, since for $N_n=2n$ and $\psi=\delta_0 \otimes \begin{pmatrix}
    0 \\
    1
\end{pmatrix}$
the measure $\mu_{T,\psi}^N$ is supported on even integers for all $T \geq 1$. For example, taking $\phi=\mathbf{1}_{\text{even}}$ yields $\langle\phi\rangle_{T,\psi}=1$ and $\langle\phi\rangle=\frac{1}{2}$. In contrast, if $\widetilde{\psi}=\delta_0 \otimes \begin{pmatrix}
    1 \\
    0
\end{pmatrix}$, then the measure $\mu_{T,\widetilde{\psi}}^N$ is uniform on $\LL_N$. Let us compare these direct observations with what we get from Theorem~\ref{thm:equidistribution}. We compute the values of the coefficients $c_{u,\psi}^{(k)}$ for this quantum walk for a general initial state $\psi$ of compact support. We have
\[
\Utheta=\begin{pmatrix}
    \ee^{-2\pi \ii \theta} &0 \\
    0  &\ee^{4\pi \ii \theta}
\end{pmatrix}
\]
so the eigenvalues of $\Utheta$ are $E_1(\theta)=\ee^{-2\pi \ii \theta}$ and $E_2=\ee^{4\pi \ii \theta}$ with eigenprojections $P_{E_1}(\theta)=\begin{pmatrix}
    1 & 0\\
    0 & 0
\end{pmatrix}$ and $P_{E_2}(\theta)=\begin{pmatrix}
    0 & 0\\
    0 & 1
\end{pmatrix}$, respectively. Note that $E_2\left(\theta+\frac{1}{2}\right)=E_2(\theta)$ for all $\theta$ and there are no further phase shift relations among the eigenvalues of $\Utheta$; this gives $M=2$. We know $c_{u,\psi}^{(1)} = \frac{1}{M}=\frac{1}{2}$. Furthermore, $P_{E_2}(i_1,i_2)$ is non-zero only if $i_1=i_2=2$. Hence, in this case $c_{u, \psi}^{(2)}$ simplifies to (note that $\varphi=\frac{1}{2}$ here)
\begin{align*}
c_{u,\psi}^{(2)}&= \frac{1}{M}+ \frac{1}{M} \sum_{k_1,k_2 \in \Lambda} \overline{\psi_{2}(k_1)}\psi_{2}(k_2)\ee^{2\pi \ii (k_1-u)\left(\frac{1}{2}\right)}\int_{0}^1   \ee^{2\pi \ii (k_1-k_2)\theta} \dd\theta\\
&= \frac{1}{2}+ \frac{1}{2} \sum_{k \in \Lambda} |\psi_{2}(k)|^2\ee^{\pi \ii (k-u)}  
\end{align*}
for $u=0,1$.  In particular, $\widetilde{\psi}=\delta_0\mathop\otimes\begin{pmatrix}1\\0\end{pmatrix}$ equidistributes ($c_{0,\widetilde{\psi}}^{(2)}=c_{1,\widetilde{\psi}}^{(2)}=\frac{1}{2})$ as previously observed, while for $\psi = \delta_0\mathop\otimes\begin{pmatrix} 0\\1\end{pmatrix}$, $\langle\phi\rangle_{T,\psi}$ approaches $1\langle \phi\rangle_{B_0}+0\langle \phi\rangle_{B_1}$, i.e. the semiclassical measure on $\LL_{2n}$ is $\frac{1}{n}\mathbf{1}_{\text{even}}$.

This example also illustrates the difference between subsequence convergence in quantum and classical random walks. The classical version of this walk with initial state $\delta_0\mathop\otimes \binom{f_1}{f_2}$ is a random walk taking either one step to the right with probability $p=|f_1|^2 \in (0,1)$ or two steps to the left. The corresponding classical walk is aperiodic and its stationary distribution is the uniform distribution on $\LL_{N}$ for all values of $N$. In particular, this shows that for every subsequence $N_n$, the total variation distance between the stationary measure of the classical random walk on $\LL_{N_n}$ and the uniform distribution on $\LL_{N_n}$ is zero, unlike the case of quantum walks. 
\end{exa}

\subsection{PQE for regular observables in 1d}
In this section, we prove point (5) in Theorem~\ref{thm:main1d}.

\begin{thm}\label{thm:1D_regular}
 Let $U$ be a quantum walk on $\Z$ as in \eqref{e:ugen}-\eqref{eqn:homogen_unitary} with no flat bands. Let $\psi = \sum_{k\in \Lambda} \psi(k)\delta_k$ be an initial state of compact support.
\begin{enumerate}[\rm(i)]
\item If $\phi^{(N)}(k) = f(k/N)$ for some function $f\in H^s(\T)$, $s>1/2$, then 
\begin{equation*}
\lim_{N\to\infty}\lim_{T\to\infty} \langle \phi \rangle_{T,\psi}=\int_0^1 f(x)\dd x \,. 
\end{equation*}
\item If $\phi^{(N)}$ is the restriction to $\LL_N$ of a function $f\in\ell^1(\Z)$, then
\[
\lim_{N\to\infty}\lim_{T\to\infty} \langle \phi \rangle_{T,\psi}=0 \,.
\]
\end{enumerate}

\end{thm}
As in Definition~\ref{def:rego}, in case (i) we choose the version of $\phi$ which is continuous. 
\begin{proof}[Proof of Theorem \ref{thm:1D_regular}]
If \eqref{e:flo} holds, then the theorem directly follows from Theorem~\ref{thm:cri0}, since $\lim_{T\to\infty}\langle\phi\rangle_{T,\psi}=(\lim_{T\to\infty}\langle\phi\rangle_{T,\psi}-\langle\phi\rangle)+\langle\phi\rangle\to 0+\int f$ in case (i) and $\to 0+0$ in case (ii). So we assume \eqref{e:flo} does not hold. We consider the sequence $N_n^{(k)}=nM+k$, where $M$ is as defined in \eqref{e:defnM} and $1 \leq k \leq M$. We shall prove that for all values of $k$, the limit on the subsequence is as stated in (i) and (ii), 
and therefore the theorem holds.

We first consider the observables $\phi$ from the class $\ell^1(\Z)$ and recall the constant $c_{u,\psi}^{(k)}$ from Theorem \ref{thm:equidistribution}. Note that for $\phi \in \ell^1(\Z)$, $\langle \phi \rangle=\frac{1}{N}\sum_{j \in \LL_N} \phi^{(N)}(j)$ converges to zero as $n \to \infty$, and $\langle \phi \rangle_{B_u}= \frac{1}{n}\sum_{j=0}^{n-1} \phi(jM+u)$ or $\frac{1}{n+1}\sum_{j=0}^n\phi(jM+u)$, both of which converge to $0$ for all values of $u$. Therefore, the expression $\sum_{u=0}^{M-1} c_{u,\psi}^{(k)}\langle \phi \rangle_{B_u}$ converges to zero, for all values of $k$. Therefore we have
\begin{align*}
   \Big|\lim_{T\to\infty} \langle \phi \rangle_{T,\psi}\Big| \leq  \Big|\lim_{T\to\infty}\langle \phi \rangle_{T,\psi} 
- \sum_{u=0}^{M-1}c_{u,\psi}^{(k)}\langle \phi\rangle_{B_u}\Big| + \Big|\sum_{u=0}^{M-1}c_{m,\psi}^{(k)}\langle \phi\rangle_{B_u}\Big| \to 0
\end{align*}
as $n\to\infty$, using Theorem~\ref{thm:equidistribution}. Thus, for any $k$,
\[
\lim_{n\to\infty}\lim_{T\to\infty} \langle \phi \rangle_{T,\psi}=0
\]

Next, suppose $\phi$ has the form $\phi(r)=f(r/N)$ for some continuous $f \in H^s(\mathbb{T})$. Then $\langle \phi \rangle_{B_u}=\frac{1}{n}\sum_{j=0}^{n-1} f\big(\frac{jM+u}{N} \big)$ if $k\le u\le M-1$ or $\frac{1}{n+1}\sum_{j=0}^n f\big(\frac{jM+u}{N} \big)$ if $0\le u \leq k-1$. Either way, we consider the subdivision $\{x_i\}_{i=0}^n$ of $[0,1]$ given by $x_i = \frac{iM}{N}$ for $i=0,\dots,n-1$ and $x_n:=1$. Then both sums are Riemann sums up to a multiplicative sequence converging to one, hence $\langle \phi\rangle_{B_u}\to \int f(x)\,\dd x$ for any $u$.
Furthermore, we have $\sum_{u=0}^{M-1}c_{u,\psi}^{(k)}=1$ for all $k$. Hence, it follows that 
\[
\lim_{n\to\infty}\lim_{T\to\infty} \langle \phi \rangle_{T,\psi}=\int_0^1 f(x)\dd x \,,
\]
for all $k$, and this completes the proof.
\end{proof}

\begin{cor}\label{cor:acpqe}
    Let $U$ be a quantum on $\Z$ as in \eqref{e:ugen}-\eqref{eqn:homogen_unitary}. If $U$ has no flat bands, then $U$ satisfies \eqref{e:mainlim} for regular observables.
\end{cor}
\begin{proof}
Recall that the limit $\lim_{T\to\infty} \langle \phi\rangle_{T,\psi} = \langle \phi\rangle_\psi$ exists, with $\langle \phi\rangle_\psi:=\sum_{k\in \LL_N} \phi(k)\mu_\psi^N(k)$, see \eqref{e:limtonly}. We thus need to show that $|\langle \phi\rangle_\psi-\langle \phi\rangle|\to 0$ as $N\to\infty$. But Theorem~\ref{thm:1D_regular} implies that $\langle \phi\rangle_\psi$ and $\langle \phi\rangle$ both have the same limit (either $\int f$ or $0$) if $\phi$ is a regular observable. The corollary follows.
\end{proof}

The assumption that $\phi$ is regular cannot be dropped: we will see in Proposition~\ref{prp:coidia} that there exist $1d$ walks violating \eqref{e:mainlim} if $\phi$ only satisfies $\|\phi\|_\infty\le 1$. We also saw this in Example~\ref{exa:39}.

This proves point (5) in Theorem~\ref{thm:main1d}. In fact, the converse follows from Proposition~\ref{prp:flanoqe}.

\subsection{Applications to periodic Schr\"odingers}\label{sec:schro}

Even though the main objective of this paper is to study the ergodicity of discrete-time quantum walks, the scope of the previous subsections is much broader. Let us discuss some applications to continuous-time quantum walks (CTQW) and eigenvectors ergodicity. In that framework there is no spin space. We work instead with a $\Z$-periodic graph, and the fundamental domain takes up the role of the spin space. See \cite{BKS} for some figures.

Theorem~\ref{thm:1D_subseq_general} applies directly to that framework and yields the following corollary. Here, the approximation of the $\Z$-periodic graph $\Gamma$ is done by the subgraphs $\Gamma_{N}=\cup_{\br\in\LL_N^d}(V_f+\br)$, where $V_f$ is the fundamental domain of $\Gamma$, which is assumed to be finite, with $\nu$ vertices.

\begin{cor}
Consider a $\Z$-periodic graph with a finite fundamental domain. Suppose the underlying periodic Schr\"odinger operator $H_\Gamma$ has no flat bands. Then there is a subsequence $N_n$ of $N$ such that over $\Gamma_{N_n}$,
\begin{enumerate}[\rm(1)]
    \item Any othonormal eigenbasis of $H_{\Gamma_{N_n}}$ is quantum ergodic in the sense of \cite[Thm. 1.2]{McKSa}.
    \item The CTQW $\ee^{-\ii tH_{\Gamma_{N_n}}}$ is ergodic in the sense of \cite[Thm. 3.1]{BS}
\end{enumerate}
\end{cor}

Indeed, for both frameworks, the same Floquet assumption \eqref{e:flo} arises but with $\widehat{U}(\theta)$ Hermitian. As Theorem~\ref{thm:1D_subseq_general} holds for normal matrices, it applies to (1)-(2) as well.

We also state the following quite unexpected result:

\begin{cor}\label{cor:qestrong}
    Consider a $\Z$-periodic graph with underlying Schr\"odinger operator $H_\Gamma$. If \eqref{e:flo} is satisfied, then for any orthonormal eigenbasis $(\psi_u^{(N)})$ of $H_{\Gamma_N}$, we have $\frac{1}{f(N)}\sum_{u\in\Gamma_N}|\langle \psi_u^{(N)},a\psi_u^{(N)}\rangle-\langle\psi_u^{(N)},\opn(\overline{a})\psi_u^{(N)}\rangle|^2\to 0$, for any $f(N)\to\infty$.
\end{cor}

This follows from Lemma~\ref{lem:1D_subseq_2conditions}. In fact, following Step 3 in \cite{McKSa}, using the fact that $|A_m|\le C$, we can divide by $f(N)$ instead of $N^d$  (here $d=1$). Corollary~\ref{cor:qestrong} is an especially strong statement of quantum ergodicity. One can rephrase this result by saying that the quantum variance decays like $N^{-1}$. Note that quantum \emph{unique} ergodicity fails however, see \cite[\S\,5.1]{McKSa}.

Graph products are a natural example of periodic graphs, and it was proved in \cite{McKSa} that for all finite graphs $G$, the cartesian product $G \mathop\square \Z^d$ and the strong product $G \mathop\boxtimes \Z^d$ obey the Floquet assumption \eqref{e:flo}. This is however not true for the tensor product $G \times \Z^d$ (for an example where \eqref{e:flo} is violated for $G \times \Z$, see Proposition~1.6 of \cite{McKSa}). In the following proposition, we show that even when \eqref{e:flo} is violated, the CTQW on $G \times \Z$ is \eqref{e:mainlim} for all $\ell^\infty$-observables of the form $a(v_q,k):=\phi(k)$, where $v_q \in G$ and $k \in \Z$.
\begin{cor}
    For all finite graphs $G$ with $0 \notin \sigma(G)$, the continuous time-quantum walk on the tensor product $G \times \Z$ satisfies \eqref{e:mainlim} for all $\phi^{(N)}$ with $\|\phi^{(N)}\|_\infty\le 1$ : 
\begin{equation*}
\lim_{N\to\infty}\Big|\lim_{T\to\infty} \langle \phi \rangle_{T,\psi} - \langle \phi\rangle \Big|=0 \, ,
\end{equation*}
where $\langle \phi\rangle_{T,\psi} = \sum_{r\in\LL_N} \phi(r)\mu_{T,\psi}^N(r)$ and $\mu_{T,\psi}^N(r)=\frac{1}{T}\int_0^T\sum_{v_q\in G}|(\ee^{-\ii tH_{\Gamma_N}}\psi)(v_q,r)|^2\,\dd t$.
\end{cor}
\begin{proof}
The proof of Theorem~\ref{thm:equidistribution} adapts to CTQW without difficulty. To maintain consistency with the calculation in Theorem 3.1 of \cite{BS}, we take the initial state as the qubit $\psi=\delta_{k_1}\mathop\otimes\delta_{v_p}$, where $v_p \in G$ and $k_1 \in \Z$. Here, along the subsequence $N_n=nM+k$, the coefficient of $\langle \phi \rangle_{B_u}$ is
\begin{multline*}
        c_{u,\psi}^{(k)}= \frac{1}{M}+ \frac{1}{M} \sum_{\substack{0 < \varphi<1\\ \varphi\gcd(M,k) \in \Z}} \sum_{\substack{s,w=1 \\ E_s(x+\varphi)\equiv E_w(x)}}^{\nu^\prime} \hspace{-3mm}\ee^{2\pi \ii (k_1-u)\varphi}
    \int_{0}^1 P_{E_s}\left(\theta+\varphi\right) P_{E_w}(\theta)(v_p,v_p)   \dd\theta.
\end{multline*}

For the graph $G \times \Z$, the eigenvalues of $\Utheta$ are of the form $E_s(\theta) = 2\mu_s \cos 2\pi \theta$, where $\mu_s$ is an eigenvalue of the adjacency matrix of $G$, and the eigenvectors and therefore the eigenprojections of $\Utheta$ are independent of $\theta$. Therefore for all $\varphi>0$,
    \[
    P_{E_s}(\theta+\varphi) P_{E_w}(\theta)=\begin{cases}
        P_{E_s} &\text{ if } s=w \\
        0 &\text{ if } s\neq w
    \end{cases}
    \]
    and the sum in $c_{u,\psi}^{(k)}$ reduces to $w=s$.
    
Note that $E_s(\theta+\varphi)= 2\mu_s \cos 2\pi (\theta+\varphi) \not \equiv E_s(\theta)$ for any $\varphi \in (0,1)$ as $\mu_s\neq 0$. Therefore for $\varphi \neq 0$, the sum in $c_{u,\psi}^{(k)}$ is empty, i.e. for all $k$ and $u$, $c_{u,\psi}^{(k)}=\frac{1}{M}$. Thus, $\sum_{u=0}^{M-1} c_{u,\psi}^{(k)} \langle \phi \rangle_{B_u}=\frac{1}{M}\sum_{u=0}^{M-1} \langle \phi \rangle_{B_u}=\langle \phi \rangle$. It follows that for all $0 \leq k \leq M-1$, along the subsequence $N_n=nM+k$, $\lim_{n\to\infty}\big|\lim_{T\to\infty} \langle \phi \rangle_{T,\psi} - \langle \phi\rangle \big|=0$. Since the same limit holds along every subsequence, we have the result.
\end{proof}

Finally, we see from the proofs that Theorem~\ref{thm:1D_regular} and Corollary~\ref{cor:acpqe} are quite direct consequences of the important Theorem~\ref{thm:equidistribution}. As we previously mentioned, the proof of Theorem~\ref{thm:equidistribution} adapts without difficulty to continuous-time quantum walks, so we get:

\begin{cor}
Theorem~\ref{thm:1D_regular} and Corollary~\ref{cor:acpqe} hold for CTQW on $\Z$-periodic graphs, with $\mu_{T,\psi}^N(r)=\frac{1}{T}\int_0^T\sum_{v_q\in G}|(\ee^{-\ii tH_{\Gamma_N}}\psi)(v_q,r)|^2\,\dd t$.
\end{cor}

\section{Focus on two-state quantum walks on the line}\label{sec:2state}

In this section, we consider different models of quantum walks on $\Z$ with two spin states. More specifically, we look at three different models: coined walks of the form $U=\cS(I \otimes \cC)$ defined by a $2\times 2$ coin; the split step quantum walk, where each step of the quantum walk involves two operations of the coin matrix before stepping in each direction; and finally, the arc-reversal quantum walk on $\Z$.  These walks are widely studied in the literature as they are completely solvable. Points (6)-(7) of Theorem~\ref{thm:main1d} will follow along the way, (see Proposition~\ref{prp:offdia}), except for the last point of (7) which we postpone to Proposition~\ref{prp:nrgstro}.

We introduce the following shorthand notations.

\begin{defa}
   We say a quantum walk $U$ on $\Z$ is 
   \begin{enumerate}[(i)]
       \item $\ell^\infty$-\eqref{eqn:ergodic_general} if \eqref{eqn:ergodic_general} holds for all $a^{(N)}\in\ell^2(\LL_N,\C^\nu)$ with $\|a^{(N)}\|_\infty\le 1$.
       \item $\ell^\infty$-\eqref{e:mainlim} if \eqref{e:mainlim} holds for all $\phi^{(N)} \in \ell^2(\LL_N)$ with $\|\phi^{(N)}\|_\infty\le 1$.
   \end{enumerate}
\end{defa}

Note that $\ell^\infty$-\eqref{eqn:ergodic_general} implies $\ell^\infty$-\eqref{e:mainlim}. We recall that in the present setting of one-dimensional walks, Theorem~\ref{thm:1D_l_infty} tells us that \eqref{e:flo} implies $\ell^\infty$-\eqref{eqn:ergodic_general} and the convergence of $\mu_\psi^N$ and $\mu_{\psi,j}^N$ to their respective limits in total variation distance. Also, absence of flat bands guarantees \eqref{e:mainlim} for regular observables (Corollary~\ref{cor:acpqe}) and $\ell^\infty$-\eqref{eqn:ergodic_general} on a subsequence (Theorem~\ref{thm:1D_subseq_general}).

\subsection{Two-state coined walk}
Coined walks on $\Z$ are historically the first models that gave birth to the field. Notwithstanding their simplicity, they nicely illustrate the quantum effects of the walk. In the case of two spins, limit theorems for the walk per unit time were established in \cite{Konno2002,Konno2005}, see also \cite{CJWW} for recent extensions. In \cite{Konno2002,Konno2005}, the step sizes were chosen as one, with nonzero coin entries. We consider a generalized version of this model wherein we allow arbitrary step sizes. We will see that this affects the ergodicity of the quantum walk.

Consider a walk $U$ on $\ell^2(\Z,\C^2)$ of the form 
\[
U=\begin{pmatrix}
	aS_{-\alpha} & bS_{-\alpha} \\
	cS_{\beta}& dS_{\beta}
\end{pmatrix}
\]
where $\begin{pmatrix}
a & b \\
c & d
\end{pmatrix}=:C$ is unitary. 

\subsubsection{Coins with nonzero entries}

We start with the case where the coin $C$ has nonzero entries, i.e. $abcd\neq 0$. This is a common assumption in the literature, see e.g. \cite{Konno2002,Konno2005}, but we will see later that the case $abcd=0$ is also very interesting when $\alpha\neq\beta$.

 \begin{thm}\label{thm:2state_gen}
 	Consider the quantum walk $U=\begin{pmatrix}
 		aS_{-\alpha} & bS_{-\alpha} \\
 		cS_{\beta}& dS_{\beta}
 	\end{pmatrix}$ on $\Z$ with $\alpha,\beta\ge 1$. If $abcd\neq 0$ and $\gcd(\alpha,\beta)=1$, then \eqref{e:flo} is satisfied. As a result, $U$ is $\ell^\infty$-\eqref{eqn:ergodic_general}.
    \end{thm}
 \begin{proof}
 	From \eqref{e:uhatheta}, we see that $\widehat{U}(\theta)$ here is given by 
 	\[
    \widehat{U}(\theta)=\begin{pmatrix}
 		a\ee^{2\pi \ii \alpha \theta} & b\ee^{2\pi \ii \alpha \theta} \\
 		c\ee^{-2\pi \ii \beta \theta} & d\ee^{-2\pi \ii \beta \theta} 
 	\end{pmatrix}.
    \]
 	The characteristic polynomial of $\Utheta$ is
 	\begin{equation}\label{e:charpoco}
    p(\lambda)=\lambda^2-(a\ee^{2\pi \ii \alpha \theta}+d\ee^{-2\pi \ii  \beta\theta} )\lambda+(ad-bc)\ee^{2\pi \ii (\alpha-\beta) \theta}
    \end{equation}
 	Denoting the eigenvalues of $\Utheta$ by $\lambda_1(\theta)$ and $\lambda_2(\theta)$, we get that $\lambda_1(\theta)\lambda_2(\theta)=(ad-bc)\ee^{2\pi \ii (\alpha-\beta) \theta}$, and $\lambda_1(\theta)+\lambda_2(\theta)=a\ee^{2\pi \ii \alpha \theta}+d\ee^{-2\pi \ii  \beta\theta}$. Let $\mu_1(\theta)$ and $\mu_2(\theta)$ denote the eigenvalues of $\widehat{U}(\theta+\varphi)$ for some fixed $\varphi>0$. If $\lambda_{1}(\theta)=\mu_{1}(\theta)$ for $\varphi \in (0,1)$, then 
 	\begin{equation}\label{eqn:lambda_j2}
 		\frac{\mu_2(\theta)}{\lambda_{2}(\theta)}=\frac{\ee^{2\pi \ii (\alpha-\beta) (\theta+\varphi)}}{\ee^{2\pi \ii (\alpha-\beta) \theta}}=\ee^{2\pi \ii (\alpha-\beta)\varphi}.
 	\end{equation}
 	Therefore, $\mu_2(\theta)=\lambda_{2}(\theta)\ee^{2\pi \ii (\alpha-\beta)\varphi}$. By considering the sum of eigenvalues, we have
 	\begin{align*}
 		\lambda_{1}(\theta)+\lambda_{2}(\theta) &=a\ee^{2\pi \ii \alpha \theta}+d\ee^{-2\pi \ii  \beta\theta}, \\
 		\mu_{1}(\theta)+\mu_2(\theta) &=a\ee^{2\pi \ii \alpha (\theta+\varphi)}+d\ee^{-2\pi \ii  \beta(\theta+\varphi)}.
 	\end{align*}
 	Subtracting the first equation from the second, using $\lambda_{1}(\theta)=\mu_{1}(\theta)$ and substituting the expression of $\mu_2(\theta)$ from \eqref{eqn:lambda_j2}, we get that
 	\begin{equation}\label{eqn:lambda2_condition}
 		\lambda_{2}(\theta)\big(\ee^{2\pi \ii (\alpha-\beta)\varphi}-1\big)=a\ee^{2\pi \ii \alpha \theta}\left(\ee^{2\pi \ii \alpha \varphi}-1\right)+d\ee^{-2\pi \ii \beta \theta}\big(\ee^{-2\pi \ii \beta \varphi}-1\big).
 	\end{equation}
 
 	Since $\Utheta$ is unitary, then $|\lambda_{2}(\theta)|=1$ and therefore for \eqref{eqn:lambda2_condition} to hold, we must have
 	\[
    \left|\ee^{2\pi \ii(\alpha-\beta)\varphi}-1\right|=\left|a\ee^{2\pi \ii \alpha \theta}\left(\ee^{2\pi \ii \alpha \varphi}-1\right)+d\ee^{-2\pi \ii \beta \theta}\big(\ee^{-2\pi \ii \beta \varphi}-1\big)\right|.
    \]
 	Note that once we fix $\varphi$, the only variable is $\theta$ and therefore the last equation has a solution at $\theta_0$ only if the equation 
 	 	\begin{equation}\label{eqn:c1c2c3}
 	 		c_3=\left|c_1\ee^{2\pi \ii \alpha \theta}+c_2\ee^{-2\pi \ii \beta \theta}\right|
 	 	\end{equation}	
has a solution at $\theta_0$, for 
 	 	\begin{align*}
 	 		c_1 &= a\left(\ee^{2\pi \ii \alpha \varphi}-1\right),\\
 	 		c_2 &= d\big(\ee^{-2\pi \ii \beta \varphi}-1\big) \text{ and} \\
 	 		c_3 &= \big|\ee^{2\pi \ii (\alpha-\beta)\varphi}-1\big|.
 	 	\end{align*}
 	 	Taking the square modulus in \eqref{eqn:c1c2c3}, we get a necessary condition,
 	 	\begin{equation}\label{eqn:c1c2c3_simplified}
 	 		|c_1|^2+|c_2|^2+2\Re(\overline{c_1}c_2\ee^{-2\pi \ii (\alpha+\beta)\theta})=|c_3|^2
 	 	\end{equation}
 	 Now, we look at the number of possible solutions of \eqref{eqn:c1c2c3_simplified} (in $\theta$) by dividing into different cases based on the values of $c_1,c_2$ and $c_3$. 
 	 	\vskip2pt
 	 	\noindent \textbf{Case I:} $c_1c_2 \neq 0$.\\
 	 	In this case, solving \eqref{eqn:c1c2c3_simplified} reduces to solving $2\Re(\overline{c_1}c_2e^{-2\pi \ii (\alpha+\beta)\theta})=$ constant. Therefore, \eqref{eqn:c1c2c3_simplified} holds for at most $2(\alpha+\beta)$ values of $\theta$. 
 	 	\vskip2pt
 	 	\noindent\textbf{Case II:} $c_1c_2=0$ and $c_3 \neq 0$.\\
 	 	We first consider the subcase $c_1=0$. Since $a \neq 0$, $\varphi=r/\alpha$ for some integer $r$ such that $1 \leq r \leq \alpha-1$. Here, \eqref{eqn:c1c2c3} reduces to $c_3=|c_2\ee^{-2\pi \ii \beta \theta}|=|c_2|$, that is,
 	 	\begin{align*}
 	 		\left|d\big(\ee^{-2\pi \ii \beta \frac{r}{\alpha}}-1\big)\right| &= \left|\ee^{2\pi \ii (\alpha-\beta)\frac{r}{\alpha}}-1\right|\\
 	 		&= \left|\ee^{-2\pi \ii \beta\frac{r}{\alpha}}-1\right|,
 	 	\end{align*}
 	 	which implies that $|d|=1$ since $\gcd(\alpha,\beta)=1$. Since $U$ is unitary, $|d|=1$ implies $|a|=1$ and $|b|=|c|=0$, a contradiction since $abcd \neq 0$. A similar calculation also works when $c_2=0$ and therefore \eqref{eqn:c1c2c3_simplified} has no solutions in this case.
 	 	\vskip2pt
 	 	\noindent\textbf{Case III:} $c_1c_2=0$ and $c_3=0$.\\
 	 	Note that once $c_3=0$ and $c_1c_2=0$, we get from \eqref{eqn:c1c2c3_simplified} that $c_1=c_2=0$. Now $c_1=0$ only if $\varphi=\frac{r}{\alpha}$ for some $1 \leq r \leq \alpha-1$, and $c_2=0$ only if $\varphi=\frac{r'}{\beta}$ for some $1 \leq r' \leq \beta-1$. Since $\gcd(\alpha,\beta)=1$, such $r,r'$ do not exist, so this case cannot occur. 

        We have shown that for all the possible cases, and for each fixed $\varphi \in (0,1)$, the number of $\theta$ such that $E_s(\theta+\varphi)=E_w(\theta)$ for some $s,w$ is bounded by a constant independent of $\varphi$. By Lemma~\ref{lem:1D_subseq_2conditions}, it follows that $U$ obeys \eqref{e:flo} and consequently $U$ is $\ell^\infty$-\eqref{eqn:ergodic_general}.
 \end{proof}
 
The following is an immediate corollary of Theorem \ref{thm:2state_gen}.

 \begin{cor}\label{cor:had}
		The Hadamard walk is $\ell^\infty$-\eqref{eqn:ergodic_general} and $\delta(\mu_\psi^N,\mu^N)\to 0$.
	\end{cor}

\begin{rem}\label{rem:hadcyc}
As we mentioned in the introduction, the Hadamard walk of Corollary~\ref{cor:had} is the only concrete example whose ergodicity has previously been established in the quantum walk literature \cite{Aha,Bednarska2003}, using explicit computations of $\mu_\psi^N(v) =\lim_{T\to\infty}\mu_{T,\psi}^N(v)$. The paper \cite{NayakVishwanath2000} asked the same question but had a different approach. There, the authors work directly on $\Z$ and observe that the distribution of $P(n,t):= \|(U^t \psi)(n)\|_{\C^2}^2$ tends to be close to uniform for vertices $n\in [\![\frac{-t}{\sqrt{2}},\frac{t}{\sqrt{2}}]\!]$ as time grows large (note that $P(\cdot,t)$ is a probability measure on $[\![-t,t]\!]$). The way they justify this is by computing the asymptotics of $P(n,t)$ and effectively looking at the pushforward measure $h_\ast P(\cdot,t)$ on $[-1,1]$ by the map $h:[\![-t,t]\!]\to [-1,1]$, $h(n)=n/t$ and computing some of its moments. Actually, the convergence of $h_\ast P(\cdot ,t)$ was rigorously established by Konno later \cite{Konno2005} among his limit theorems, the limiting measure is indeed supported on $[\frac{-1}{\sqrt{2}},\frac{1}{\sqrt{2}}]$, but is not really uniform; it has a density with respect to the Lebesgue measure, which is ``somehow'' flat-like near zero; the phenomenon which is observed in \cite{NayakVishwanath2000}.
\end{rem}

In Theorem~\ref{thm:2state_gen}, besides the assumption that $abcd\neq 0$, we assumed that $\gcd(\alpha,\beta)=1$. If $\gcd(\alpha,\beta)>1$, we know that the classical random walks on the cycle 
$\LL_{N_n}$ with $N_n=n \operatorname{lcm}(\alpha,\beta)$ and step sizes $\alpha$ to the left  and $\beta$ to the right, are periodic, and their stationary distribution is not the uniform distribution on $\LL_{N_n}$. In the following example, we show that this is true for discrete-time quantum walks as well.

 \begin{exa}
 	Consider the quantum walk defined by $U=\frac{1}{\sqrt{2}}\begin{bmatrix}
 		S_{-2} & S_{-2} \\
 		S_{2}& -S_{2}
 	\end{bmatrix}$. This a coined walk with a Hadamard coin, but with step size two. Here,
 	$\widehat{U}(\theta)=\frac{1}{\sqrt{2}}\begin{pmatrix}
 		\ee^{4\pi \ii \theta} & e^{4\pi \ii \theta} \\
 		\ee^{-4\pi \ii \theta} & -\ee^{-4\pi \ii \theta} 
 	\end{pmatrix},$ implying that $\Utheta=\widehat{U}(\theta+1/2)$ for all $\theta$. In particular, $E_s\big(\frac{r+N/2}{N}\big)=E_s\big(\frac{r}{N}\big)$ for all $r$ and even $N$, so \eqref{e:flo} is violated. $\ell^\infty$-\eqref{e:mainlim} is also violated on a subsequence. To see this, choose $N_n=n \operatorname{lcm}(2,2)=2n$ and consider the initial state $\psi=\delta_0 \otimes f$. As the step size is 2, the positions that $U^{k}$ takes are $P_k=\{2s\mod N_n: 0 \leq s \leq k\}$. Since $N_n$ is even, $P_k$ is a subset of even natural numbers.
    Therefore, by choosing 
    $\phi = \mathbf{1}_{\text{odd in }\LL_{N_n}}$, we get that \eqref{e:mainlim} is violated.

    Here, choosing $N_n$ as even is important, because if $N_n$ is odd, the quantum walk can go to the position $N-2$ at the first step, which is odd, and eventually reach every position in at most in $\frac{N-1}{2}$ steps. This is in harmony with Theorem~\ref{thm:main1d} which says that some subsequence in fact satisfies \eqref{eqn:ergodic_general}. As there are no flat bands, the same theorem tells us that \eqref{e:mainlim} is still satisfied for the full sequence for \emph{regular} observables.
 \end{exa}

 \subsubsection{Diagonal and anti-diagonal coin matrices}
We now study the quantum walks whose coin matrices have zero entries. Since the coin matrix $C=\begin{pmatrix}
 		a & b \\
 		c &d
 	\end{pmatrix}$ is unitary, then $a=0$ if and only if $d=0$, and similarly $b=0$ if and only if $c=0$. Therefore, the coin matrices with zero entries are either diagonal or anti-diagonal matrices. We treat these cases separately in Propositions~\ref{prp:coidia} and \ref{prp:offdia}.

  \begin{prp}\label{prp:coidia}
      For the coined walk on $\Z$ given by $U=\begin{pmatrix}
 		aS_{-\alpha} & 0 \\
 		0 &d S_{\beta}
 	\end{pmatrix}$, $|a|=|d|=1$ and $\alpha,\beta \in \Z_+$, if $\alpha=\beta=1$, then the quantum walk satisfies \eqref{e:flo} and is $\ell^\infty$-\eqref{eqn:ergodic_general}. If $\alpha\neq 1$ or $\beta\neq 1$, it is not $\ell^\infty$-\eqref{e:mainlim}.
  \end{prp}
\begin{proof}
Here $\Utheta=\begin{pmatrix}
 	a\ee^{2 \pi \ii \alpha \theta} & 0 \\
 	0 & d \ee^{-2 \pi \ii \beta \theta}
 \end{pmatrix}$, and its eigenvalues are $a\ee^{2 \pi \ii \alpha \theta}$ and $d \ee^{-2 \pi \ii \beta \theta}$, with periods $1/\alpha$ and $1/\beta$, respectively.

If $\alpha=\beta=1$, both eigenvalues have period 1 and \eqref{e:flo} is satisfied. Hence as a consequence of Theorem~\ref{thm:main1d}, $U$ is $\ell^\infty$-\eqref{eqn:ergodic_general}.

If either one of $\alpha, \beta>1$, \eqref{e:flo} is violated. Let us construct a sequence $(N_n) \to \infty$ along which \eqref{e:mainlim} is violated. Suppose $\alpha>1$ (the case $\beta>1$ is similar), and consider the sequence $N_n=\alpha n$ for $n \geq 1$ and the initial state $\psi=\delta_0 \otimes \begin{pmatrix}
     1 \\ 0
 \end{pmatrix}=\begin{pmatrix}
     \delta_0 \\ 0
 \end{pmatrix}$.
Then $U^k=\begin{pmatrix}
     a^k S_{-\alpha k} &0 \\ 
     0 &d^k S_{\beta k}
 \end{pmatrix}$ and consequently
 $U_{N_n}^k\psi=\delta_{(-\alpha k) \ \, \mathrm{mod} N_n} \otimes \begin{pmatrix}
     a^k \\ 0
 \end{pmatrix}$ for all $k$ and $n$.  Since $N_n=\alpha n$, it follows that the set of positions that $U_{N_n}$ reaches is $\{(-\alpha k)\mod N_n: k \in \N\}$, which is precisely the set $\cP_n=\{\alpha i: 0\leq i < n \}$. Therefore, the measure $\mu_{T,\psi}^{N_n}$ is supported on $\cP_n$ and in fact $\lim_{T \to \infty} \mu_{T,\psi}^{N_n}(r)=\frac{\mathbf{1}_{\cP_n}(r)}{n}$, where $\mathbf{1}_{\cP_n}$ is the indicator function of $\cP_n$. Hence, if $\phi^{(N_n)}=\mathbf{1}_{\cP_n}$, then $\lim_{T\to\infty}\langle \phi\rangle_{T,\psi}=1$, while $\langle \phi\rangle = \frac{1}{\alpha}$. This violates $\ell^\infty$-\eqref{e:mainlim}.
  \end{proof}

 \begin{rem}
Note that for any $\alpha,\beta\ge 1$, the diagonal walk of Proposition~\ref{prp:coidia} has no flat bands, and thus satisfies \eqref{e:mainlim} for regular observables.

For $\alpha=\beta=1$, the walk takes $\delta_0 \otimes \begin{pmatrix}
     f_1 \\
     f_2
 \end{pmatrix}\mapsto \delta_{N-1} \otimes \begin{pmatrix}
     f_1 \\
     0
 \end{pmatrix}+\delta_1 \otimes \begin{pmatrix}
     0\\
     f_2
 \end{pmatrix}\mapsto \delta_{N-2} \otimes \begin{pmatrix}
     f_1 \\
     0
 \end{pmatrix}+\delta_2 \otimes \begin{pmatrix}
     0\\
     f_2
 \end{pmatrix}\mapsto \dots$. After $N$ steps, $U_N^N\psi = \psi$. 
Here periodic boundary conditions mean that we are effectively working on a cycle.

Also note that $U_N^n\psi$ itself does not converge, a time averaging is necessary.

It is also clear that there is no equidistribution in spin space, e.g. if $f_2=0$, the spin remains ``up'' at all times.
 \end{rem}
 The dynamics of the case $a=d=0$ is much more interesting. 
 
  \begin{prp}\label{prp:offdia}
      For a two-state coined walk on $\Z$ given by $U=\begin{pmatrix}
 		0 & bS_{-\alpha} \\
 		cS_{\beta} &0
 	\end{pmatrix}$, $|b|=|c|=1$ and $\alpha,\beta \in \Z_+$, the walk $U$
    \begin{enumerate}[\rm (i)]
    \item has a flat band and is not $\ell^\infty$-\eqref{e:mainlim} if $\alpha=\beta$.
    \item is $\ell^\infty$-\eqref{eqn:ergodic_general} if $|\alpha-\beta|=1$.
    \item is not $\ell^\infty$-\eqref{e:mainlim} if \emph{($|\alpha-\beta|>2$)} or \emph{(}$|\alpha-\beta|=2$ with $\alpha,\beta$ even\emph{)}.
    \item is $\ell^\infty$-\eqref{e:mainlim} but not $\ell^\infty$-\eqref{eqn:ergodic_general} if $|\alpha-\beta|=2$ with $\alpha,\beta$ odd.
    \end{enumerate}
    Furthermore, there are no flat bands except in case \emph{(i)}, and \eqref{e:flo} holds only in \emph{(ii)}.
  \end{prp}

  \begin{figure}[h!]
    \centering
    \includegraphics[width=140mm,height=70mm]{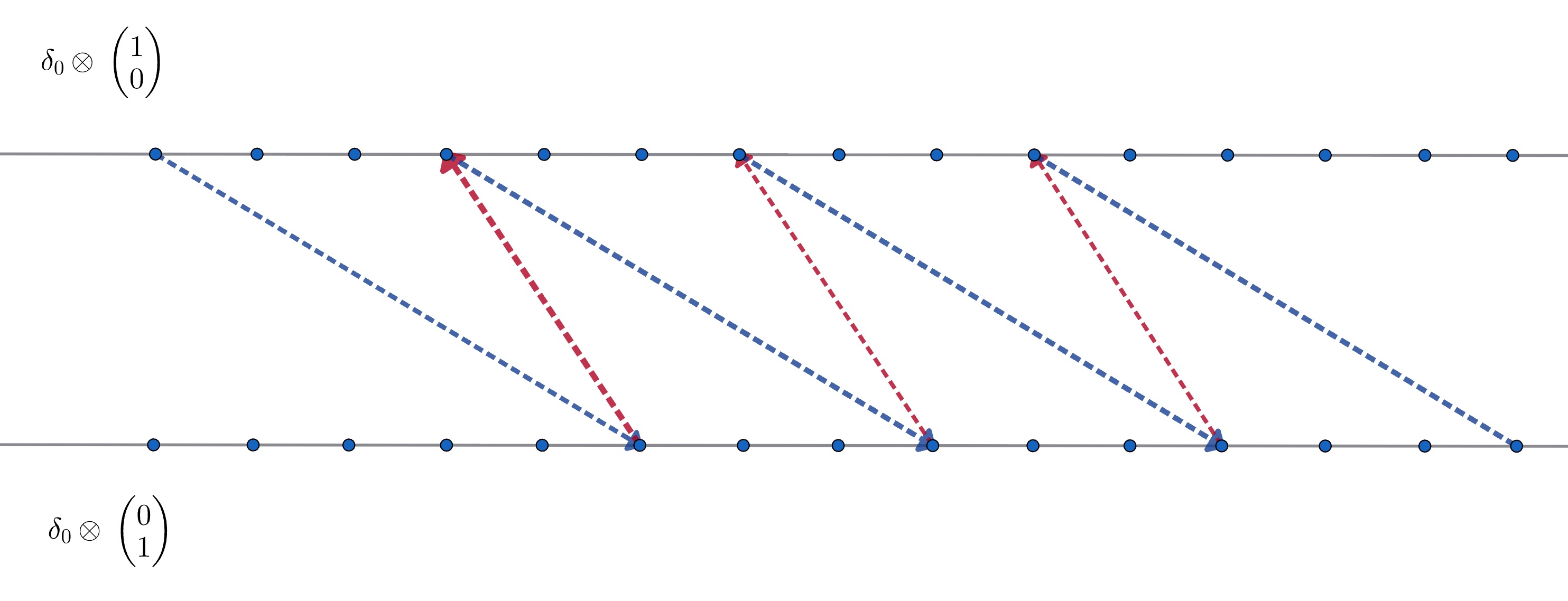}
    \caption{Time evolution of the $1d$ quantum walk with coin matrix $\begin{pmatrix}
        0& 1\\ 1  &0
    \end{pmatrix}$
    and step sizes $\alpha=2$ and $\beta=5$, and initial state $\psi=\delta_0 \otimes \begin{pmatrix}
        1 \\ 0
    \end{pmatrix}$.}
    \label{fig:2state_example}
\end{figure}

\begin{proof}
    We first look at the time evolution of the quantum walk for the initial state $\psi=\delta_0\otimes \begin{pmatrix}
     1 \\ 0
 \end{pmatrix}$. Notice that 
 $U^{2}=\begin{pmatrix}
     bcS_{\beta-\alpha} &0 \\ 
     0 & bcS_{\beta-\alpha}
 \end{pmatrix}$ and 
\begin{equation}\label{eqn:U_evol_anti-diagonal}
    U\left( \delta_0 \otimes \begin{pmatrix}
     1 \\0 
 \end{pmatrix}
 \right)= \delta_\beta \otimes \begin{pmatrix}
     0 \\ c
 \end{pmatrix} \qquad
 \text{ and } \qquad
  U\left( \delta_0 \otimes \begin{pmatrix}
     0 \\1 
 \end{pmatrix}
 \right)= \delta_{-\alpha} \otimes \begin{pmatrix}
     b \\ 0
 \end{pmatrix}. 
\end{equation}

Therefore,
\begin{align}\label{e:upooff}
U_{N_n}^{2k}\psi&=\begin{pmatrix}
     (bc)^k S_{\beta-\alpha}^k &0 \\ 
     0 & (bc)^k S_{\beta-\alpha}^k
 \end{pmatrix}\left(\delta_{0}\otimes \begin{pmatrix}
     1 \\ 0
 \end{pmatrix}\right)=\delta_{k \beta-k\alpha \ \, \mathrm{mod} N_n}\otimes\begin{pmatrix}
     (bc)^k \\ 0
 \end{pmatrix}, \nonumber \\
U_{N_n}^{2k+1}\psi&=\begin{pmatrix}
     0 & bS_{-\alpha} \\ 
     cS_{\beta} & 0
 \end{pmatrix}U_{N_n}^{2k}\psi= \delta_{(k+1)\beta-k\alpha \, \ \mathrm{mod} N_n} \begin{pmatrix}
       0 \\ c(bc)^k
 \end{pmatrix}.
\end{align}

Hence, the set of points in the position space $\LL_{N_n}$ that the quantum walk reaches starting from the initial state $\psi=\delta_0 \otimes \begin{pmatrix}
    1 \\ 0
\end{pmatrix}$ is
\[
\mathcal{P}_n=\{(k \beta-k \alpha)\mod N_n,\ (k\beta-k\alpha+\beta)\mod N_n: k \in \N\}.
\]

Now, we look at the eigenvalues of $\Utheta$. We have
    $\Utheta=\begin{pmatrix}
 	0 &be^{-2 \pi \ii \alpha \theta}\\
 	ce^{2 \pi \ii \beta \theta} &0
 \end{pmatrix}$
 and therefore the eigenvalues of $\Utheta$ are 
 \begin{equation}\label{e:eigenoff}
E_1(\theta)=\sqrt{bc}\,\ee^{\pi \ii (\beta-\alpha)\theta} \qquad \text{ and } \qquad E_2(\theta)=-\sqrt{bc}\,\ee^{\pi \ii (\beta-\alpha)\theta}.
 \end{equation}
 
 Now, we prove each of the cases separately.
 
 (i) For $\alpha=\beta$, the eigenvalues of $\Utheta$ are $E_1(\theta)=\sqrt{bc}$ and $E_2(\theta)=-\sqrt{bc}$.   Also,  
\[
U^{2k}\psi=\delta_{0}\otimes\begin{pmatrix}
     (bc)^k \\ 0
 \end{pmatrix} \quad \text{ and } \quad 
U^{2k+1}\psi= \delta_\beta \otimes \begin{pmatrix}
       0 \\ (bc)^k
 \end{pmatrix} \quad  \text{ for all } k \in \N.
\]
Therefore, the walk concentrates on two points $0$ and $\beta$, in clear violation of \eqref{e:mainlim}. In fact, $\mu_{T,\psi}^N=\frac{\delta_0+\delta_\beta}{2}$ if $T$ is even and $\mu_{T,\psi}^N=(\frac{1}{2}+\frac{1}{2T})\delta_0+(\frac{1}{2}-\frac{1}{2T})\delta_\beta$ if $T$ is odd, so taking e.g. $\phi=\delta_0+\delta_\beta$, $\phi$ is regular, satisfies $\langle \phi\rangle_{T,\psi}=1$ but $\langle\phi\rangle=\frac{2}{N}$. This proves (i).

(ii) Back to \eqref{e:eigenoff}, if $|\alpha-\beta|=1$, then  \eqref{e:flo} is satisfied, so $U$ is $\ell^\infty$-\eqref{eqn:ergodic_general}.
 
(iii) For $|\alpha - \beta|>1$,
 \[
 E_2(\theta)=\ee^{\ii \pi} E_1(\theta)=\sqrt{bc}\,\ee^{\ii \pi (\beta-\alpha)\left(\theta+\frac{1}{\beta-\alpha}\right)}=E_1\Big(\theta+\frac{1}{\beta-\alpha}\Big).
 \]
 So \eqref{e:flo} is violated on the subsequence of $N_n=n|\beta-\alpha|$ at $m=n$.

 To show that $U$ is not $\ell^\infty$-\eqref{e:mainlim} along $N_n$, we first consider $\alpha,\beta$ such that $|\alpha-\beta|>2$. Recall \eqref{e:upooff}.
Note that the quantum walk on $\LL_{N_n}$ reaches back the initial position after $2n$ steps, and then retraces its steps. Since $N_n=n|\beta-\alpha|$, we have $N_n>2n$ and therefore the set of positions that the quantum walk reaches, $\cP_n$ is a strict subset of $\LL_{N_n}$. Consider the observable $\phi=\mathbf{1}_{\cP_n^c} \in \ell^\infty(\Z)$. 
Since $N_n=n|\alpha-\beta|$, we get $\supp(\phi) \cap \cP_n = \emptyset$ for all $n$, therefore $\langle \phi\rangle_{T,\psi}=0$ for all $n$. On the other hand, 
\[
\langle \phi \rangle= \frac{|\cP_n^c|}{N_n}=\frac{N_n-2n}{N_n}=
1-\frac{2}{|\beta-\alpha|} \nto 0.
\]

Similarly, when $|\alpha-\beta|=2 $ with $\alpha,\beta$ even, the set $\cP_n$ consists only of even integers. Hence, choosing $\phi$ as the indicator function of odd integers, we get a violation of \eqref{e:mainlim}. 

The time evolution of the quantum walk in this case is depicted in Figure \ref{fig:2state_example}.

(iv) Consider $\alpha,\beta$ odd, with $|\alpha-\beta|=2$. We shall first show that $U$ is $\ell^\infty$-\eqref{e:mainlim}. For this, we show that the measure $\mu_{T,\psi}^{N}$ converges in total variation distance to the uniform distribution for every normalized initial state $\psi$, and it then follows by Proposition~\ref{prp:tvd} that $U$ is $\ell^\infty$-\eqref{e:mainlim}. Note from \eqref{eqn:U_evol_anti-diagonal} and linearity of $U$ that for $\psi=\sum_{p \in \Lambda} \psi(p)\delta_p$, with $\psi(p)=\binom{\psi_1(p)}{\psi_2(p)}$, we have for $k\in\N$,
\[
U_{N_n}^{2k}\psi=\begin{pmatrix}
     (bc)^k S_{k(\beta-\alpha)} &0 \\ 
     0 & (bc)^k S_{k(\beta-\alpha)}
 \end{pmatrix} \psi=\sum_{p \in \Lambda}(bc)^k
      \psi(p)\delta_{p+k \beta-k\alpha \ \,\mathrm{mod} N_n} 
 \]
 \[
U_{N_n}^{2k+1}\psi=\begin{pmatrix}
     0 & bS_{-\alpha} \\ 
     cS_{\beta} & 0
 \end{pmatrix}U_{N_n}^{2k}\psi=\sum_{p \in \Lambda} (bc)^k \binom{b\psi_2(p)\delta_{p+k\beta-(k+1)\alpha \ \,\mathrm{mod} N_n}}{c\psi_1(p)\delta_{p+(k+1)\beta-k\alpha \ \,\mathrm{mod} N_n}} .
\]

Assume $\beta-\alpha=2$. Then $(U_{N_n}^{2k}\psi)(r)=(bc)^k\psi(r-2k)$, so recalling $|b|=|c|=1$,
\[
\sum_{k=0}^{n-1} \|(U_{N_n}^{2k}\psi)(r)\|^2 = \begin{cases} \sum_{m\text{ even in }\LL_N}\|\psi(m)\|^2 &\text{if } r \text{ is  even},\\\sum_{m\text{ odd in }\LL_N}\|\psi(m)\|^2 &\text{if } r \text{ is  odd}. \end{cases}
\]
Next $(U_{N_n}^{2k+1}\psi)(r)=(bc)^k\binom{b \psi_2(r-2k+\alpha)}{c\psi_1(r-2k-\beta)}$, so $\|(U_{N_n}^{2k+1}\psi)(r)\|^2 = |\psi_2(r-2k+\alpha)|^2+|\psi_1(r-2k-\alpha-2)|^2$, as $\beta=\alpha+2$. Since $\alpha$ is odd, it follows that
\[
\sum_{k=0}^{n-1} \|(U_{N_n}^{2k+1}\psi)(r)\|^2 = \begin{cases} \sum_{m\text{ odd in }\LL_N}\|\psi(m)\|^2 &\text{if } r \text{ is  even},\\\sum_{m\text{ even in }\LL_N}\|\psi(m)\|^2 &\text{if } r \text{ is  odd}, \end{cases}
\]
Combining both equations, we get that
\[
\mu_{2n,\psi}^{N_n}(r)=\frac{1}{2n}\sum_{k=0}^{2n-1} \|(U_{N_n}^{k}\psi)(r)\|^2= \frac{1}{2n}\sum_{m\in\LL_N}\|\psi(m)\|^2=\frac{1}{2n}=\mu^{N_n}(r),
\]
for all $r$, where $\mu^{N_n}$ is the uniform measure on $\LL_{N_n}$.  More generally, $\mu_{2nq,\psi}^{N_n}(r)=\mu^{N_n}(r)$ for any integer $q$, since the values of $U_{N_n}^{2k}\psi$ and $U_{N_n}^{2k+1}\psi$ are periodic of period $2n$. Finally, letting $T=2nq+t$ with $0\le t<2n$, we have $\mu_{T,\psi}^{N_n}(r)=\frac{1}{2nq+t}(q+c(\psi,t))$, with $c(\psi,t)\le \|\psi\|^2\le 1$. Thus, $\mu_\psi^{N_n}(r) = \lim_{q\to\infty} \mu_{2nq+t,\psi}^{N_n}(r)=\frac{1}{2n}=\mu^{N_n}(r)$ for any $t$. In particular, $\delta(\mu_{\psi}^{N_n},\mu^{N_n})=0$, so $U$ satisfies $\ell^\infty$-\eqref{e:mainlim} by Proposition~\ref{prp:tvd}.

But this walk is not $\ell^\infty$-\eqref{eqn:ergodic_general}. To see this, consider $N_n=|\beta-\alpha| n=2n$ and note from \eqref{eqn:U_evol_anti-diagonal} that for the initial state $\psi=\delta_0 \otimes \begin{pmatrix}
     1 \\
   0
 \end{pmatrix}$, the quantum walk only reaches positions of the form $\delta_{2k} \otimes \begin{pmatrix} c' \\0 \end{pmatrix}$ or $\delta_{2k+1} \otimes \begin{pmatrix} 0 \\c' \end{pmatrix}$, where $|c'|=1$. Consider the observable $a=\begin{pmatrix}
     a_1\\a_2
 \end{pmatrix} \in \ell^\infty(\Z,\C^2)$, where $a_1$ is the indicator function of odd integers and $a_2$ is the indicator function of even integers.

 Then $a_1(2r)(U_N^{2k}\psi)_1(2r)=a_2(2r+1)(U_N^{2k+1}\psi)_2(2r+1)=0$ because $a$ vanishes there, and $a_2(2r)(U_N^{2k}\psi)_2(2r)=a_1(2r+1)(U_N^{2k+1}\psi)_1(2r+1)=0$ because $U^n\psi$ vanishes there.
Therefore, $\langle U_N^k \psi, a U_N^k \psi \rangle=0$. 

Now, we calculate $ \langle a \rangle_{\psi}$. Recall from \eqref{e:eigenoff} that the eigenvalues of $\Utheta$ are $E_1(\theta)=\sqrt{bc}\,\ee^{2\pi \ii \theta }$ and $E_2(\theta)=-\sqrt{bc}\,\ee^{2\pi \ii \theta }$, with unit eigenvectors $u_1(\theta)=\left(\frac{\sqrt{b}}{\sqrt{|b|+|c|}},\frac{\sqrt{c}}{\sqrt{|b|+|c|}}\ee^{\pi \ii (\alpha+\beta) \theta}\right)$ and $u_2(\theta)=\left(\frac{\sqrt{b}}{\sqrt{|b|+|c|}},-\frac{\sqrt{c}}{\sqrt{|b|+|c|}}\ee^{\pi \ii(\alpha+\beta) \theta}\right)$, respectively.
 Therefore, we get the projection matrices as
 \begin{align*}
P_{E_1}(\theta) &= \frac{1}{|b|+|c|}
\begin{pmatrix}
b & \sqrt{bc}\, \ee^{-\pi i (\alpha+\beta)\theta} \\
\sqrt{bc}\, \ee^{\pi i (\alpha+\beta)\theta} & c
\end{pmatrix} \text{ and } \\
P_{E_2}(\theta) &= \frac{1}{|b|+|c|}
\begin{pmatrix}
b & -\sqrt{bc}\, \ee^{-\pi i (\alpha+\beta)\theta} \\
-\sqrt{bc}\, \ee^{\pi i (\alpha+\beta)\theta} & c
\end{pmatrix} 
 \end{align*}

For $\psi=\delta_0 \otimes \begin{pmatrix} 1 \\ 0\end{pmatrix}$, we have $\widehat{\psi}(r)=\frac{1}{\sqrt{N}}\begin{pmatrix}1\\0\end{pmatrix}$ for all $r$. Since
\[
P_{E_1}(\theta)\begin{pmatrix}1\\0\end{pmatrix} = \frac{1}{|b|+|c|}\begin{pmatrix}b\\\sqrt{bc}\,\ee^{\pi\ii(\alpha+\beta)\theta}\end{pmatrix} \quad \text{and}\quad P_{E_2}(\theta)\begin{pmatrix}1\\0\end{pmatrix}=\frac{1}{|b|+|c|}\begin{pmatrix}b\\-\sqrt{bc}\,\ee^{\pi\ii(\alpha+\beta)\theta}\end{pmatrix}
\]
then \eqref{e:avp} gives, noting that $\langle a_1\rangle=\langle a_2\rangle=\frac{1}{2}$,
\begin{align*}
\langle a\rangle_\psi &= \frac{\langle a_1\rangle}{N} \sum_{r\in \LL_N} \frac{2|b|^2}{(|b|+|c|)^2} + \frac{\langle a_2\rangle}{N}\sum_{r\in \LL_N} \frac{2|bc|}{(|b|+|c|)^2} 
   \\
   &=\frac{|b|^2+|bc|}{(|b|+|c|)^2} =\frac{1}{2}\nto 0 .
\end{align*}
 Therefore, the quantum walk violates \eqref{eqn:ergodic_general} for this $a=a^{(N)}$.
\end{proof}

\subsection{Arc-reversal quantum walks on cycles}
Homogeneous arc-reversal quantum walks on cycles fall into our framework and are given by \eqref{eqn:unitary_arcreversal}. 
\begin{thm}
    Let $U$ be a homogeneous arc-reversal quantum walk with coin matrix $C=\begin{pmatrix}
        a &b \\
        c &d
    \end{pmatrix}$ such that $abcd \neq 0$. Then $U$ obeys \eqref{e:flo}.
\end{thm}
For an arc-reversal quantum walk $U$, its unitary matrix is given by $\begin{pmatrix}
    cS_1 & dS_1 \\
    aS_{-1} &bS_{-1}
\end{pmatrix}$ and the characteristic polynomial of $\Utheta$ is 
\[
p(\lambda)=\lambda^2-(c\,\ee^{-2\pi \ii \theta}+b\,\ee^{2\pi \ii  \theta} )\lambda+(bc-ad) \,.
\]
Comparing with \eqref{e:charpoco}, the proof clearly follows the same lines as the proof of Theorem~\ref{thm:2state_gen}.

\subsection{Split-step quantum walks of general step size}
In this section, we look at a generalization of split-step quantum walk in which the jump sizes $\alpha,\beta$ are fixed arbitrary natural numbers. See Section~\ref{sec:spli} for background. For $\alpha=\beta=1$, this model has appeared in several works \cite{BHJ,HJS,CRWAGW,AbSt,ACSW} with some variations. We generalize the version in \cite{AbSt}, cf. eq. (13) there. The unitary matrix in our case is given by
\begin{equation}\label{e:split-step_general}
  U=\begin{pmatrix}
    S_\alpha &0 \\
    0 & 1
\end{pmatrix}
\begin{pmatrix}
    -t &r \\
    r & t
\end{pmatrix}\begin{pmatrix}
    1 &0 \\
    0 & S_{-\beta}
\end{pmatrix}\begin{pmatrix}
    -t&r \\
    r & t
\end{pmatrix}=\begin{pmatrix}
    r^2 S_{\alpha-\beta}+t^2S_\alpha &rtS_{\alpha-\beta}-rtS_\alpha \\
    rtS_{-\beta}-rt &r^2+t^2S_{-\beta}
\end{pmatrix},  
\end{equation}
with $r^2+t^2=1$. For $\alpha=\beta=1$, the characteristic polynomial of $\widehat{U}(\theta)$ reduces to $ \lambda^2-2\lambda(r^2+t^2\cos 2\pi\theta) +1$. The eigenvalues can be explicitly computed as
\begin{equation}\label{e:rtsimple}
\lambda_{\pm}(\theta)=r^2+t^2\cos 2\pi\theta\pm \ii t\sqrt{t^2\sin^2 2\pi\theta + 2r^2(1-\cos 2\pi\theta)}
\end{equation}
and one sees that they satisfy \eqref{e:flo}. More generally:
 \begin{thm}\label{thm:2state_gen_ss}
 Consider the split-step quantum walk $U$ on $\Z$ given in \eqref{e:split-step_general}, with $r,t \in \R\setminus\{0\}$, $r^2+t^2=1$ and $\alpha,\beta\ge 1$. If $\gcd(\alpha,\beta)=1$, then $U$ satisfies \eqref{e:flo}.
 \end{thm}
\begin{proof}
For this quantum walk $U$, we have
\[ \Utheta=\begin{pmatrix}
    r^2 \ee^{2\pi \ii (\beta-\alpha)\theta}+ t^2 \ee^{-2\pi \ii \alpha\theta} & rt \ee^{2\pi \ii (\beta-\alpha)\theta} - rt \ee^{-2\pi \ii \alpha\theta}\\
     rt \ee^{2\pi \ii\beta\theta}-rt & r^2+ t^2\ee^{2\pi \ii \beta\theta}
\end{pmatrix}.\]
The characteristic polynomial of $\Utheta$ has the first-order coefficient $r^2\ee^{2\pi \ii (\beta-\alpha)\theta}+r^2+t^2\ee^{-2\pi \ii \alpha\theta}+t^2\ee^{2\pi \ii \beta\theta}$ and constant coefficient $(r^2+t^2)^2\ee^{2\pi \ii (\beta-\alpha)\theta}=\ee^{2\pi \ii (\beta-\alpha)\theta}$. Denote the eigenvalues of $\Utheta$ as $\lambda_{1}(\theta)$ and $\lambda_{2}(\theta)$ and that of $\widehat{U}(\theta+\varphi)$ as $\mu_{1}(\theta)$ and $\mu_{2}(\theta)$. Noting that the constant coefficient equals the product of eigenvalues, we have $\lambda_1(\theta)\lambda_2(\theta)=\ee^{2\pi\ii(\beta-\alpha)\theta}$ and $\mu_1(\theta)\mu_2(\theta)=\ee^{2\pi\ii(\beta-\alpha)(\theta+\varphi)}$. If $\lambda_{1}(\theta)=\mu_{1}(\theta)$, we get 
\[
\mu_{2}(\theta)=\ee^{2\pi \ii (\beta-\alpha)\varphi}\lambda_{2}(\theta).
\]
Looking at the first-degree coefficients, we have
\begin{align*}
   \lambda_{1}(\theta)+\lambda_{2}(\theta) &=r^2\ee^{2\pi \ii (\beta-\alpha)\theta}+r^2+t^2\ee^{-2\pi \ii \alpha\theta}+t^2\ee^{2\pi \ii \beta\theta} \\
\mu_{1}(\theta)+\mu_{2}(\theta) &=r^2\ee^{2\pi \ii (\beta-\alpha)(\theta+\varphi)}+r^2+t^2\ee^{-2\pi \ii \alpha(\theta+\varphi)}+t^2\ee^{2\pi \ii \beta(\theta+\varphi)} 
\end{align*}
Subtracting the first equation from the second, using $\lambda_1(\theta)=\mu_1(\theta)$, and substituting the expressions of $\mu_{1}(\theta)$ and $\mu_{2}(\theta)$, we get
\begin{equation}\label{e:rt-case_c1-4}
    c_4\lambda_{2}(\theta)=c_1\ee^{2\pi \ii (\beta-\alpha)\theta}+c_2\ee^{-2\pi \ii \alpha\theta}+c_3\ee^{2\pi \ii\beta \theta}
\end{equation}
where 
\begin{align*}
    c_1 &= r^2(\ee^{2\pi \ii (\beta-\alpha)\varphi}-1) \\
    c_2 &= t^2 (\ee^{-2\pi \ii \alpha \varphi}-1) \\
    c_3 &=t^2(\ee^{2\pi \ii \beta \varphi}-1) \text{ and }\\
    c_4 &= \ee^{2\pi \ii (\beta-\alpha)\varphi}-1.
\end{align*}
 Taking the square modulus in \eqref{e:rt-case_c1-4}, we get 
\begin{equation}\label{e:rt-case_c1-4_simplified}
    |c_4|^2=|c_1|^2+|c_2|^2+|c_3|^2+2 \Re(\overline{c_1}c_3\ee^{2 \pi \ii \alpha \theta}+\overline{c_2}c_3\ee^{2 \pi \ii (\alpha+\beta) \theta}+\overline{c_2}c_1\ee^{2 \pi \ii \beta \theta}) \,.
\end{equation}

Now we consider different subcases based on the values of $c_1,c_2,c_3$ and $c_4$. 

\noindent\textbf{Case I:} Exactly two entries from $\{c_1,c_2,c_3\}$ are zero.\\
In this case, the last term in the right-hand side of \eqref{e:rt-case_c1-4_simplified} is equal to zero. Therefore \eqref{e:rt-case_c1-4_simplified} becomes $|c_4|^2=|c_1|^2+|c_2|^2+|c_3|^2$. If $c_1=0$, then $c_4=0$, and the equation becomes $0=|c_2|^2+|c_3|^2$, implying that both $c_2$ and $c_3$ are zero, a contradiction. If $c_1 \neq 0$, then $c_2=c_3=0$ by assumption, so we get $|c_4|=|c_1|$. As $|c_1|=r^2|c_4|$, this gives $r=1$, so $t=0$, a contradiction. Therefore, \eqref{e:rt-case_c1-4_simplified} has no solution in this case.
\vskip2pt
\noindent \textbf{Case II:}  Exactly one entry from $\{c_1,c_2,c_3\}$ is zero. \\
Suppose $c_1$ is the entry that is zero, then \eqref{e:rt-case_c1-4_simplified} becomes 
\begin{align*}
   0 &=|c_2|^2+|c_3|^2+2\Re (\overline{c_2}c_3\ee^{2 \pi \ii (\alpha+\beta) \theta}) .   
\end{align*}
Note that once $\varphi$ is fixed, $c_2,$ and $c_3$ are constants and as a result, the above equation has at most $2(\alpha+\beta)$ solutions. Therefore, the theorem holds in this case. 

Now, suppose $c_1 \neq 0$ and $c_2=0$, then \eqref{e:rt-case_c1-4_simplified} simplifies to $|c_4|^2=|c_1|^2+|c_3|^2+2 \Re (\overline{c_1}c_3\ee^{2 \pi \ii \alpha \theta})$. Here too, this equation has only finitely many solutions, independent of $\phi$, and therefore the theorem follows.

The case $c_1\neq 0$ and $c_3=0$ is similar, and therefore \eqref{e:rt-case_c1-4} has only finitely many solutions in this case.
\vskip2pt
\noindent \textbf{Case III:} $c_1,c_2,c_3 \neq 0$. \\
In this case, \eqref{e:rt-case_c1-4_simplified} becomes
\begin{equation*}
    \Re(\overline{c_1}c_3\ee^{2 \pi \ii \alpha \theta}+\overline{c_2}c_3\ee^{2 \pi \ii (\alpha+\beta) \theta}+\overline{c_2}c_1\ee^{2 \pi \ii \beta \theta})=\text{constant}.
\end{equation*}
To prove that the above equation has only finitely many solutions, we consider the complex polynomial $p(z)=\overline{c_1}c_3z^\alpha+\overline{c_2}c_3z^{\alpha+\beta}+\overline{c_2}c_1z^\beta$. Then $f(z)=2\Re(p(z))$ is a Laurent polynomial $\overline{c_3}c_2z^{-(\alpha+\beta)}+\overline{c_1}c_2z^{-\beta}+\overline{c_3}c_1z^{-\alpha}+\overline{c_1}c_3z^\alpha+\overline{c_2}c_3z^{\alpha+\beta}+\overline{c_2}c_1z^\beta$. Multiplying by $z^{\alpha+\beta}$, $f(z)$ becomes a polynomial, and this implies that $f(z)$ has at most $2(\alpha+\beta)$ zeros. Hence, the theorem is true in this case as well.
Now, we are left with the final case.
\vskip2pt
\noindent \textbf{Case IV:} $c_1,c_2,c_3 = 0$. \\
Note that $c_2=0$ only if $m=\frac{r}{\alpha}$ for some $0<r<\alpha-1$ and $c_3=0$ only if $m=\frac{r'}{\beta}$ for some $0<r'<\beta-1$. Such a choice of $r,r'$ is not possible since $\gcd(\alpha,\beta)=1$, and therefore this case does not occur.

Therefore, we have shown that for all $\varphi>0$ and $s,w$, the number of $\theta$ for which $E_s(\theta+\varphi)=E_w(\theta)$ is bounded by a constant independent of $\varphi$. The theorem follows.
\end{proof}

\section{Higher dimensional walks}\label{sec:higher_dim_QW}
In this section, we look at several models of higher-dimensional quantum walks. We start with a Fourier coined walk on $\Z^2$ in \S\,\ref{sec:fou}. This is an example of a so-called PUTO walk, see Section~\ref{sec:puto} for background. The PUTO model exhibits entanglement, and the spin space is of low dimension. We prove that the Fourier coined walk satisfies \eqref{e:flo}. We next proceed in \S\,\ref{subsec:QW_non-entangle} with non-entangled walks; these are well-defined on $\Z^d$ and have a spin space of high dimension. We give classes of ergodic and non-ergodic walks. We conclude with a discussion of other models in \S\,\ref{sec:other}, through which we prove Proposition~\ref{prp:failhi}.

\subsection{Fourier Coin in 2d}\label{sec:fou}

We start our investigation of higher dimensional walks with a generalization of the Hadamard walk to dimension two. Namely, we consider a coined walk on $\Z^2$ defined through the Fourier coin $F_4=\frac{1}{2}\left(\ii^{(k-1)(\ell-1)}\right)_{k,\ell=1}^4$. The corresponding walk was studied in some detail in \cite{KomatsuTate2019}, who notably proved the absence of flat bands.

For the Fourier coined walk on $\Z^2$, the Floquet matrix takes the following form, for $\theta=(\theta_1,\theta_2)$,
\[
\Utheta=\frac{1}{2}\begin{pmatrix}
	\ee^{-2 \pi \ii \theta_1} & \ee^{-2\pi \ii \theta_1} & 	\ee^{-2 \pi \ii \theta_1} & \ee^{-2\pi \ii \theta_1} \\
	\ee^{2 \pi \ii \theta_1} & \ii \ee^{2\pi \ii \theta_1} & -	\ee^{2 \pi \ii \theta_1} & -\ii \ee^{2\pi \ii \theta_1} \\
	\ee^{-2 \pi \ii \theta_2} & -\ee^{-2\pi \ii \theta_2} & 	\ee^{-2 \pi \ii \theta_2} & -\ee^{-2\pi \ii \theta_2} \\
	\ee^{2 \pi \ii\theta_2} & -\ii \ee^{2\pi \ii \theta_2} & 	-\ee^{2 \pi \ii \theta_2} & \ii \ee^{2\pi \ii \theta_2} \\
\end{pmatrix},
\]
see \eqref{e:putosta}. Our aim is to show that $\widehat{U}(\theta)$ satisfies \eqref{e:flo} and so $U$ satisfies \eqref{eqn:ergodic_general}. Computing the eigenvalues is  difficult, so we will get around this using the theory of irreducibility of Bloch varieties, which may be applicable to broader settings.

Let $z_i := \ee^{2\pi \ii\theta_i}$. The characteristic polynomial $\det(2\widehat{U}(z) - x I)$ takes the form
\[
-16\ii +4x(z_1^{-1}+\ii z_1+z_2^{-1}+\ii z_2) - x^2(1-\ii)(z_1z_2^{-1}+z_1^{-1}z_2+2) -x^3(z_1^{-1}+\ii z_1+z_2^{-1}+\ii z_2)+x^4\\
\]

\begin{prp}
The characteristic polynomial $p(z,\lambda)=\det(\widehat{U}(z)-\lambda I)$ is irreducible as a Laurent polynomial in both $z=(z_1,z_2)$ and $\lambda$.
\end{prp}
The proposition means that the only way to factor $p(z,\lambda)=p_1(z,\lambda)p_2(z,\lambda)$ with $p_i$ a nontrivial polynomial in $\lambda$ and a Laurent polynomial in $z$ is to take $p_1$ or $p_2$ to be a monomial in $z$. Such a trivial factorization always exists, e.g. $p(z,\lambda)=z_1(z_1^{-1}p(z,\lambda))$.
\begin{proof}
    Note that $p(z,\lambda) = \frac{1}{16} \det( 2\widehat{U}(z)-2 \lambda I)$. Taking $x=2\lambda$ above yields
\begin{multline*}
(-16\ii-4\lambda^2(2-2\ii)+16\lambda^4)+z_1^{-1}(8\lambda-8\lambda^3) + z_1(8\ii\lambda-8\ii\lambda^3)\\
+z_2^{-1}(8\lambda-8\lambda^3) +z_2(8\ii\lambda-8\ii\lambda^3)-4z_1z_2^{-1}(1-\ii)\lambda^2 - 4z_1^{-1}z_2(1-\ii)\lambda^2\,.
\end{multline*}
Let us multiply this Laurent polynomial by $z_1z_2$. This gives
\begin{multline*}
P(z,\lambda)=(-16\ii-8\lambda^2(1-\ii)+16\lambda^4)z_1z_2+z_2(8\lambda-8\lambda^3) + z_1^2z_2(8\ii\lambda-8\ii\lambda^3)\\
+z_1(8\lambda-8\lambda^3) +z_1z_2^2(8\ii\lambda-8\ii\lambda^3)-4z_1^2(1-\ii)\lambda^2 - 4z_2^2(1-\ii)\lambda^2\,.
\end{multline*}

We will show that $P(z,\lambda)$ is irreducible as a traditional polynomial in $(z,\lambda)$, i.e. the only way to factorize $P(z,\lambda) = R(z,\lambda) Q(z,\lambda)$ as a product of two polynomials $R,Q$ in $z$ and $\lambda$ is to have $R(z,\lambda)=A$ or $Q(z,\lambda)=E$ for some constants $A,E\in \C$. 

For this, suppose towards a contradiction that we have a nontrivial factorization $P(z,\lambda) = R(z,\lambda) Q(z,\lambda)$. By looking at the degree of $P(z,\lambda)$, we see that $R(z,\lambda) = \sum_{0\le n,m\le 2} a_{nm} z_1^nz_2^m$ and $Q(z,\lambda)=\sum_{0\le n,m\le 2} b_{nm} z_1^n z_2^m$, and $a_{22}=b_{22}=0$, since a factor of $P$ cannot be of degree higher than $P$. Similarly, if $a_{12}\neq 0$, then since $3=\deg P = \deg R + \deg Q$, then we must have $\deg Q=0$, which leads to $Q(z,\lambda)=b_{00}$. Either $b_{00}\in \C$ is a constant, which would contradict our hypothesis that the factorization is nontrivial, or $Q(z,\lambda)=b_{00}(\lambda)$ is a nontrivial polynomial $\lambda$, with roots $\lambda_i$ independent of $z$, which would give rise to flat bands for $U$, a contradiction since $U$ has no eigenvalues, see \cite{KomatsuTate2019}. So we must have $a_{12}=0$. Similarly  $a_{21}=b_{12}=b_{21}=0$. So any nontrivial factorization of $P$ must have the form
\begin{equation}\label{e:fac}
(A+Bz_1+Cz_2+Dz_1z_2)(E+Fz_1+Gz_2+Hz_1z_2)\,.
\end{equation}
Since $P$ has no $z_1^2z_2^2$ term, we must have $DH=0$. Without loss, we assume $H=0$. 

Comparing with the coefficients of $z_1^nz_2^m$, we arrive at the following system of equations:
\[
\begin{cases} AE=0& (1),\\BE+AF=8\lambda-8\lambda^3 = AG+CE & (z_1,z_2),\\ BF=4(\ii-1)\lambda^2 = CG& (z_1^2,z_2^2),\\ DE + BG+CF=-16\ii-8\lambda^2(1-\ii)+16\lambda^4& (z_1z_2),\\ DF=8\ii \lambda - 8\ii\lambda^3 = DG& (z_1^2z_2,z_1z_2^2).\end{cases}
\]

Recall that $A,B,\dots,G$ need to be polynomials in $\lambda$.

The last equation gives $D(F-G)=0$. Since $DF=8\ii\lambda-8\ii\lambda^3$, then $D$ cannot be the zero polynomial in $\lambda$, so we must have $F=G$.

From the $3$rd equation we deduce $F(B-C)=0$, implying as before that $B=C$.

Consequently the 4th equation reads $DE+2BF=-16\ii+8\lambda^2(\ii-1)+16\lambda^4$. Using the 3rd equation, this gives $DE = -16\ii+16\lambda^4$.

Next, $AE=0$. Since $DE=-16\ii+16\lambda^4$, then $E$ cannot be the zero polynomial, so we must have $A=0$.

We thus get the equations 
\[
\begin{cases}
BE=8\lambda-8\lambda^3,\\ BF=4(\ii-1)\lambda^2,\\ DE=-16\ii+16\lambda^4,\\ DF=8\ii\lambda-8\ii \lambda^3.
\end{cases}
\]
Now $BE=8\lambda-8\lambda^3 = 8\lambda(1-\lambda)(1+\lambda)$. If $B$ contained a factor $(1\pm \lambda)$, it would not divide $4(\ii-1)\lambda^2$, contradicting the equation of $BF$. So $B=c$ or $B=c\lambda$, for some $c\in\C$. 

If $B=c$, then $F=\frac{4(\ii-1)}{c}\lambda^2$, contradicting the equation of $DF$, since $\lambda^2$ does not divide $8\ii\lambda-8\ii\lambda^3$. Thus, $B=c\lambda$ and $F=\frac{4(\ii-1)\lambda}{c}=:d\lambda$. We deduce from the first and last equations that $E=\frac{8}{c}(1-\lambda^2)$ and $D=\frac{8\ii}{d}(1-\lambda^2)$. Hence, $DE=\frac{64\ii}{cd}(1-\lambda^2)^2$. This contradicts the above equation for $DE$.

We proved that the system of equations is inconsistent, so $P(z,\lambda)$ is irreducible. 

Finally, $p(z,\lambda) =\frac{1}{16}z_1^{-1}z_2^{-1}P(z,\lambda)$. If $p(z,\lambda)=r(z,\lambda)q(z,\lambda)$ is a factorization as Laurent polynomials in $z,\lambda$, 
then for some $k_i,j_i$ large enough, $\tilde{r}(z,\lambda)=z_1^{k_1}z_2^{k_2}r(z,\lambda)$ and $\tilde{q}(z,\lambda)=z_1^{j_1}z_2^{j_2}q(z,\lambda)$ are polynomials in $(z,\lambda)$. Thus, $\frac{1}{16}z_1^{k_1+j_1-1}z_2^{k_2+j_2-1}P(z,\lambda)=\tilde{r}(z,\lambda)\tilde{q}(z,\lambda)$ is a polynomial factorization. But $P(z,\lambda)$ is irreducible, so the only factors of the LHS are $z_1^w$, $z_2^v$ and $P(z,\lambda)$. Hence, either $\tilde{r}$ or $\tilde{q}$ contains $P(z,\lambda)$, say $\tilde{q}$. Then $\tilde{r}$ must be a monomial $cz_1^nz_2^m$. So $r$ is a Laurent monomial $cz_1^pz_2^q$, $p,q\in \Z$.
\end{proof}

\begin{cor}\label{cor:founrg}
  The Fourier walk satisfies \eqref{e:flo} and is thus \eqref{eqn:ergodic_general}.  
\end{cor}
\begin{proof}
    We apply the criterion in \cite[Cor. 1.4]{Wen1}. Namely, fix $\zeta=(\zeta_1,\zeta_2)\in \C^2\setminus\{(1,1)\}$ with $|\zeta_j|=1$ and compute $p(\zeta z, \lambda)$, where $\zeta z:= (\zeta_1 z_1, \zeta_2 z_2)$. This gives
\begin{multline*}
16 p(\zeta z,\lambda) = (-16\ii-4\lambda^2(2-2\ii)+16\lambda^4)+z_1^{-1}\zeta_1^{-1}(8\lambda-8\lambda^3) + z_1\zeta_1(8\ii\lambda-8\ii\lambda^3)\\
+z_2^{-1}\zeta_2^{-1}(8\lambda-8\lambda^3) +z_2\zeta_2(8\ii\lambda-8\ii\lambda^3)-4z_1z_2^{-1}\zeta_1\zeta_2^{-1}(1-\ii)\lambda^2 - 4z_1^{-1}z_2\zeta_1^{-1}\zeta_2(1-\ii)\lambda^2\,.
\end{multline*}
If $p(\zeta z,\lambda)\equiv p(z,\lambda)$ as polynomials, they are equal at any $\lambda$. Considering them 
as Laurent polynomials in $z$ and comparing the coefficients of $z_1z_2^{-1}$, we get $\zeta_1\zeta_2^{-1}=1$, hence $\zeta_1=\zeta_2$. 
Similarly, comparing the coefficients of $z_1^{-1}$, we get $\zeta_1^{-1}=1$. Thus, $\zeta_1=\zeta_2=1$, a contradiction. This shows that $p(\zeta z,\lambda)\not \equiv p(z,\lambda)$. By \cite[Cor. 1.4]{Wen1}, it follows that \eqref{e:flo} is satisfied.
\end{proof}

To better understand this criterion, suppose we had a characteristic polynomial of the form $p(z,\lambda)=\lambda^2+(z_1^2+z_2^2)\lambda+1$ for example. In that case $p(\zeta z,\lambda)\equiv p(z,\lambda)$ with $\zeta_1=\zeta_2=-1$, so the criterion would have failed.

\begin{rem}\label{rem:uniher}
In Corollary~\ref{cor:founrg} we used \cite[Cor. 1.4]{Wen1}. Although that paper considers Schr\"odinger operators and so the Floquet matrix $\mathcal{A}(z)$ there is Hermitian, the proof goes on almost without change to our unitary setting, as the matrix elements are similar (trigonometric polynomials, see \eqref{e:uhatheta}). The only worthwhile modification is in \cite[erratum]{Wen1}, where 
Rolle's theorem is used for the real-valued eigenvalue function. Let $g(t)=E_s(\frac{r}{N}+t\frac{m_N}{N})$ and $\frac{m_N}{N}\to 0$. If $g(1)-g(0)=0$, then noting that $|g(t)|=1$ in our case, we have two scenarios: either over the segment $[0,1]$, the function $g(t)$ has winded one or more times over the unit circle, or the function moved in a specific direction from $0$ to $t_0$ (say counterclockwise), then moved back in the opposite direction. 
Because $\frac{m}{N}$ is very small and $E_s$ is continuous, it must be that the second alternative occurs.

We claim that the derivative of $g$ at $t_0$ is zero. For this, simply study $\frac{g(t)-g(t_0)}{t-t_0}$. Consider the real part of this quotient. We see it has a fixed sign if $t<t_0$, and has the opposite sign if $t>t_0$. Hence, $\Re g'(t_0)\le 0$ and $\Re g'(t_0)\ge 0$ implying $\Re g'(t_0)=0$. Similarly, $\Im g'(t_0)=0$. It follows that $g'(t_0)=0$. 
The rest of the proof in \cite{Wen1} carries over. We also refer to \cite[\S\,1.2]{Chirka} for background on complex analytic sets.
\end{rem}

\subsection{Quantum walk without entanglement}\label{subsec:QW_non-entangle}
This model, introduced by Mackay et al. \cite{Mackay}, is one of the first models of quantum walks for dimensions two or higher. It was presented as a generalization of the 1-D Hadamard walk; however, the same idea applies to any collection of $d$ 1-D quantum walks. Namely, if $\{U^{(i)}=S^{(i)}(I \otimes C_i): 1\leq i \leq d\}$ is a collection of 1-D coined quantum walks, where $S^{(i)}$ is the shift matrix in $U^{(i)}$ and $C_i$ is an $m \times m$ coin matrix, then we define the coined walk $U$ on $\Z^d$ by 
\begin{equation}\label{e:higherdimQW_model1}
  U= \otimes_{i=1}^d U^{(i)} : \ell^2(\Z^d) \otimes \C^{m^d} \to \ell^2(\Z^d) \otimes \C^{m^d}.
\end{equation}
The coin matrix of $U$ is $C_d=\otimes_{j=1}^d C$ and the shift matrix is $S=S^{(1)} \otimes S^{(2)} \otimes \cdots \otimes S^{(d)}$. 
We used the commutative property of the tensor product of vector spaces, which implies $(\ell^2(\Z) \otimes \C^{m})^{\otimes d} \cong \ell^2(\Z^d) \otimes \C^{m^d}$. For an initial state $\psi=\psi_1 \otimes \psi_2 \otimes \cdots \otimes \psi_d$, where $\psi_i \in \ell^2(\Z) \otimes \C^m$ for $1 \leq i \leq d$, we have $U\psi=\otimes_{i=1}^d U^{(i)} \psi_i$. This implies that the walk along the $i$-th component is completely determined by the qubit $\psi_i$ independently of other $\psi_j$. In other words, the qubits are non-interacting, which is why this quantum walk is also known as the separable quantum walk on $\Z^d$.

In Proposition~\ref{prp:nom_ergodic_nonentangle}, we consider walks such that each $U^{(i)}$ is of the form \eqref{eqn:homogen_unitary} with $U^{(i)}_{kl}=\sum_{p=\pm 1} U_{kl}(p)S_p$, i.e., the set of possible jumps is $\{+1,-1\}$. Since each entry in the matrix form of $U$ is a product of entries of $U^{(i)}$, the shift operators appearing in this matrix belong to the set 
\begin{equation}\label{eqn:shift_highdim}
  \{S_{x_1e_1+x_2e_2+\cdots +x_de_d}=S_{\mathbf{x}\cdot \mathbf{e}}: x_i \in \{+1,-1\} \quad \forall \ i \}\,,
\end{equation}
where $\{e_i\}$ is the standard basis of $\Z^d$. With this, we may prove the following proposition.

\begin{prp}\label{prp:nom_ergodic_nonentangle}
    Let $d \geq 2$ and $U=\mathop\otimes_{i=1}^d U^{(i)}$ on $\Z^d$, where $U_{kl}^{(i)}=\sum_{p=\pm 1} U_{kl}(p)S_p$. Then
    \begin{enumerate}[\rm(i)]
        \item The eigenvalues of $\widehat{U}(\btheta)$ do not obey the Floquet assumption \eqref{e:flo}.
        \item $U$ is not $\ell^\infty$-\eqref{e:mainlim}.
    \end{enumerate}
\end{prp}
Depending on the coefficients $U_{kl}(p)$, this walk may or may not have flat bands. For example, it has no flat bands if $U^{(i)}$ are all standard Hadamard walks.
\begin{proof}
 As observed, the shift operators appearing in $U$ are from the set \eqref{eqn:shift_highdim}. Consequently, all entries in $\widehat{U}(\btheta)$ containing $\btheta$ come from the set $\{\ee^{2\pi \ii \mathbf{x}\cdot \btheta}: x_i \in \{+1,-1\} \,\forall \ i \}$. For $\ee^{2\pi\ii \bx\cdot \btheta}$ in this set, if $\boldsymbol{\varphi}=(\frac{1}{2},\frac{1}{2},0,\ldots,0)$, we have $\ee^{2\pi \ii \mathbf{x}\cdot \btheta}=\ee^{2\pi \ii \mathbf{x}\cdot (\btheta+\boldsymbol{\varphi})}$ since $\bx\cdot \boldsymbol{\varphi}=\frac{x_1+x_2}{2} \in \{\pm 1,0\}$. Hence, $\widehat{U}(\btheta+\boldsymbol{\varphi})=\widehat{U}(\btheta)$ for all $\btheta$ and \eqref{e:flo} is violated for $N_n$ even and $\bm = (\frac{N_n}{2},\frac{N_n}{2},0,\dots,0)$.

    Now, we prove (ii). For that, we consider the sequence $N_n=2n$ and show that the quantum walk concentrates on a strict subset $\cP \subset \Z^d$ with the property $\# \cP \cap \LL_{N_n}^d < cN_n^d$ for all $n$, for some constant $c<1$. Let $\psi=\delta_{\mathbf{0}} \otimes f$, for some $\|f\|_2=1$. Since the possible shift operators are from the set \eqref{eqn:shift_highdim}, the set of positions in $\LL_{N_n}^d$ that $U_{N_n}$ reaches in $k$ steps is 
    \[
    \cP_k^{N_n}=\{((r_1-s_1)\mod N_n,\ldots, (r_d-s_d)\mod N_n): r_i,s_i \in \Z_{\geq 0}, r_i+s_i=k \, \forall i\}.
    \]
    Therefore, the set of all possible positions that the quantum walk $U_{N_n}$ reaches is
    \begin{multline*}
        \mathcal{P}_{N_n}=\cup_k \cP_k^{N_n}=\{\left((r_1-s_1)\mod N_n, \ldots, (r_d-s_d)\mod N_n\right):\\ r_i,s_i \in \Z_{\geq0}, r_i+s_i=r_j+s_j \, \forall 1\leq  i,j \leq d\}.
    \end{multline*}
     In the above expression, $r_i+s_i$ is the number of steps the quantum walk has taken to reach $(r_i-s_i) \mod N_n$, with $r_i$ being the number of steps in the $+e_i$ direction and $s_i$ the number of steps in the $-e_i$ direction. We make the following claim 

\begin{figure}[h!]
    \centering
    \includegraphics[width=0.5\linewidth]{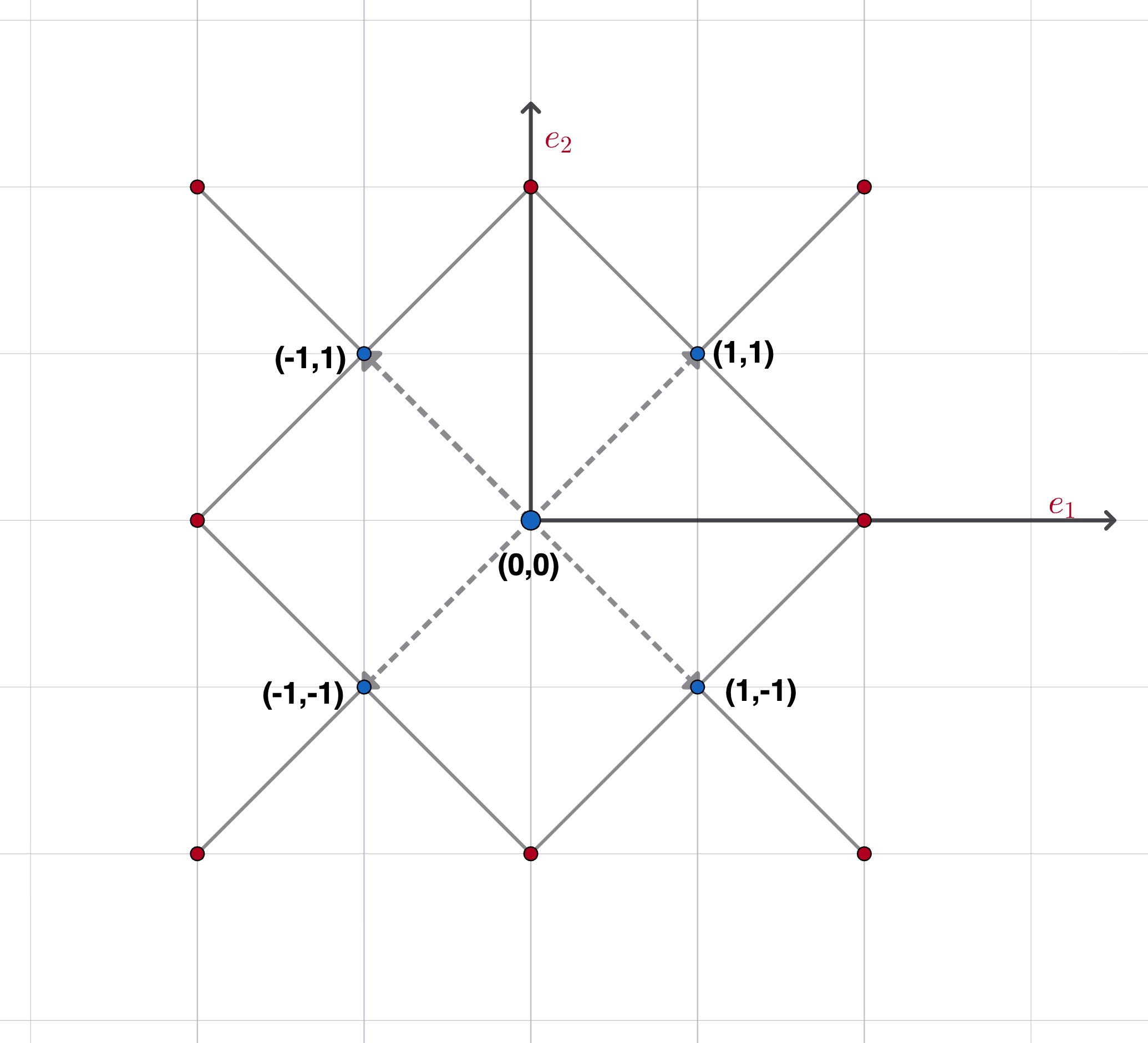}
    \caption{Illustration of the possible positions that a quantum walk without entanglement on a 2-dimensional integer lattice reaches, with dotted arrows depicting the four vectors $\pm e_1 \pm e_2$. The blue vertices are the ones that can be reached from $(0,0)$ in one step and the red vertices are the ones that can be reached in two steps.}
    \label{fig:2dqw}
\end{figure}

\begin{claim*}
 $\mathcal{P}_{N_n}=\{(w_1,w_2,\ldots,w_d) \in \LL_{N_n}^d: (-1)^{w_i}=(-1)^{w_j}, \, \forall \, 1 \leq i,j \leq d\}$.
    \end{claim*}
    First, consider an element of $\mathcal{P}_{N_n}$ for which $r_i+s_i=k$ is even.
    Then for each $i$, $r_i$ and $s_i$ must have the same parity. Furthermore, since $N_n$ is even, then $(r_i-s_i)\mod N_n$ is even for all $ 1\leq i \leq d$. Similarly, if $r_i+s_i$ is odd for all $i$, then $r_i-s_i \mod N_n$ is also odd for all $i$. This proves the inclusion of $\mathcal{P}_N$ in the other set. To prove the reverse inclusion, consider $(w_1,w_2,\ldots,w_d) \in \LL_N^d$ such that $(-1)^{w_i}=(-1)^{w_j}$ for all $ 1 \leq i,j \leq d$. Without loss of generality, suppose that $w_1 \leq w_2 \leq \cdots \leq w_d$. Since $w_d$ has the same parity as all other $w_j$'s, $\frac{w_d-w_i}{2}$ is a non-negative integer for all $i$. For each $i$, we choose $r_i=\frac{w_d+w_i}{2}$ and $s_i=\frac{w_d-w_i}{2}$. Note that here $r_i+s_i=w_d$ for all $i$, and that $r_i-s_i=w_i$. This implies that $(w_1,\ldots , w_d) \in \cP_{N_n}$. This completes the proof of the claim. 

    Finally, define $\cP=\cup_n \cP_{N_n}=\{(w_1,\ldots,w_d) \in \Z_{\geq 0}^d: (-1)^{w_i}=(-1)^{w_j}, \, \forall \, 1 \leq i,j \leq d\}$. Note that for each $n$, $\cP^c \cap \LL_{N_n}^d$ contains $N_n^d-2\big(\frac{N_n}{2}\big)^d=(2n)^d-2n^d$ elements. Therefore, if $\phi=\mathbf{1}_{\cP^c} \in \ell^\infty(\Z^d)$, then the average of $\phi$ restricted to $\LL_{N_n}^d$, $\langle \phi \rangle_{\LL_{N_n}^d}=\frac{2^d-2}{2^d}$. And since the quantum walk never reaches $\cP^c$, $\langle \phi\rangle_{T,\psi}=0$. This shows that $U$ is not $\ell^\infty$-\eqref{e:mainlim}.   
\end{proof}

We now introduce a strengthening of \eqref{e:flo} on constituent $1d$ quantum walks, which ensures that their tensor product satisfies \eqref{e:flo}. 

\begin{prp}\label{prp:higherQW_reduction}
Let $V$ be a $1d$ discrete-time quantum walk with $\{E_i(\theta):1 \leq i \leq \dim(V)\}$ as set of eigenvalues of $\widehat{V}(\theta)$. Suppose $E_i(\theta+\varphi) \not\equiv \ee^{2\pi\ii\xi}E_j(\theta)$ for all $i,j$ and $0<\varphi,\xi<1$. Then $U= V^{\otimes d}$ satisfies \eqref{e:flo}, and \eqref{eqn:ergodic_general} holds for all regular observables.
\end{prp}
From Proposition~\ref{prp:nom_ergodic_nonentangle}, we know that Hadamard constituent walks should violate this assumption. Indeed, the eigenvalues $E_{\pm}(\theta)=\frac{-\ii\sin2\pi \theta\pm \sqrt{\cos^22\pi\theta+1}}{\sqrt{2}}$ of this walk satisfy $E_+(\theta+\frac{1}{2})\equiv \ee^{\pi \ii} E_-(\theta)$. Also excluded are walks $V$ with an eigenvalue $E_j(\theta)=\ee^{2\pi\ii \alpha_j\theta}$.
\begin{proof}
    Suppose \eqref{e:flo} does not hold for $U$. We prove that $E_i(\theta+\varphi) \equiv \ee^{2\pi \ii \xi} E_j(\theta)$ for some $i,j$ and some $0<\varphi,\xi<1$.  
    
    By assumption, there exist $N_n \to \infty$, $( \mathbf{m}_n \in \LL_{N_n}^d \setminus \{\mathbf{0}\})$ and some $c>0$ such that 
 \begin{equation}\label{e:flo_higherdim}
  \# \Big\{\br \in \LL_{N_n}^d:  E_s\Big( \frac{\br+\bm_n}{N_n}\Big)-E_w\Big(\frac{\br}{N_n}\Big)=0\Big\} \geq cN_n^d.   
 \end{equation}

We represent the vectors $\br$ and $\bm_n$ by $\br=(r^{(1)},r^{(2)},\ldots, r^{(d)})$ and $\bm_n=(m_n^{(1)},m_n^{(2)},\ldots, m_n^{(d)})$, respectively. Since $\bm_n \neq 0$ for all $n$, there is a subsequence $\bm_{n_k}$ and $1 \leq \ell \leq d$ such that $m_{n_k}^{(\ell)} \neq 0$ for all $n$ and furthermore $\frac{m_{n_k}^{(\ell)}}{N_n} \to \varphi$ for some constant $\varphi \in [0,1]$. Without loss of generality, assume $\ell=1$. As a consequence of \eqref{e:flo_higherdim}, there exists a sequence of vectors $\br_n^\prime=(r_n^{(2)},r_n^{(3)},\ldots,r_n^{(d)}) \in \LL_{N_n}^{d-1}$ such that 
\begin{equation}\label{e:non-ergodic_reduction1D}
  \# \Big\{r \in \LL_{N_n}:  E_s\Big( \frac{(r,\br_n^\prime)+\bm_n}{N_n}\Big)-E_w\Big(\frac{(r,\br_n^\prime)}{N_n}\Big)=0\Big\} \geq cN_n.    
\end{equation}

Now the eigenvalues of $\widehat{U}(\btheta)$ have the form $E_{i_1}(\theta_1)E_{i_2}(\theta_2)\cdots E_{i_d}(\theta_d)$, where $E_{i_r}(\theta_r)$ are the eigenvalues of $\widehat{V}(\theta_r)$. Consequently, for $\br \in \LL_N^d$, $E_s\Big( \frac{\br+\bm_n}{N_n}\Big)=E_w\Big(\frac{\br}{N_n}\Big)$ if and only if 
$E_{i_1}\big(\frac{r_n^{(1)}+m_n^{(1)}}{N_n}\big)E_{i_2}\big( \frac{r_n^{(2)}+m_n^{(2)}}{N_n}\big)\cdots E_{i_d}\big( \frac{r_n^{(d)}+m_n^{(d)}}{N_n}\big)=E_{j_1}\big(\frac{r_n^{(1)}}{N_n}\big)E_{j_2}\big(\frac{r_n^{(2)}}{N_n}\big)\cdots E_{j_d}\big(\frac{r_n^{(d)}}{N_n}\big)$. Choosing $\br_n^\prime$ satisfying \eqref{e:non-ergodic_reduction1D} and noting that $\widehat{V}(\theta)$ is unitary, we have 
\[
E_{i_1}\Big(\frac{r_n^{(1)}+m_n^{(1)}}{N_n}\Big)= \prod_{u=2}^d \frac{E_{j_u}\big(\frac{r_n^{(u)}}{N_n}\big)}{E_{i_u}\big( \frac{r_n^{(u)}+m_n^{(u)}}{N_n}\big)}E_{j_1}\Big(\frac{r_n^{(1)}}{N_n}\Big)=:\ee^{2\pi\ii \xi_n}E_{j_1}\Big(\frac{r_n^{(1)}}{N_n}\Big) \,.
\]

If required, by taking a subsequence, assume that $\xi_n \to \xi \in [0,1]$. 
Define the functions $h_n,h:[0,1] \to \C$ by
\begin{align*}
    h_n(\theta) &=E_{i_1}(\theta)-\ee^{2 \pi\ii \xi_n}E_{j_1}(\theta), \\
    h(\theta) &=E_{i_1}(\theta)-\ee^{2 \pi\ii \xi}E_{j_1}(\theta).
\end{align*}

We claim that $h \equiv 0$. Suppose not; then proceeding as in the proof of Lemma~\ref{lem:non-ergodic_condition}, by application of Hurwitz's theorem, it follows that the number of zeros of $h_n$ is bounded uniformly by a constant, in contradiction to \eqref{e:non-ergodic_reduction1D}.
 
Therefore, we conclude that $h \equiv 0$, i.e., $E_{i_1}(\theta)=\ee^{2 \pi\ii \xi}E_{j_1}(\theta)$ for all $\theta \in [0,1]$, proving the first assertion of the theorem. The second assertion now follows from Theorem \ref{thm:pergra}. 
\end{proof}

Using the previous proposition, we provide an example of a higher-dimensional quantum walk in any dimension $d\ge 2$ which is \eqref{eqn:ergodic_general} for regular observables.
\begin{thm}
    For every $d \geq 2$, consider the $d$-dimensional walk $U=\otimes_{i=1}^d V(e_i)$, where $V=\frac{1}{2}\begin{pmatrix} I +  S_1 & -S_1 +I\\ S_{-1} -I& I+S_{-1}\end{pmatrix}$. Then $U$  is \eqref{eqn:ergodic_general} for regular observables.
\end{thm}
Note that $V$ is the $1d$ split-step quantum walk defined in \eqref{e:split-step_general} for $t=r=\frac{1}{\sqrt{2}}$ and $\alpha=\beta=1$.
One reason why $U$ is ergodic, while the $U$ of Proposition~\ref{prp:nom_ergodic_nonentangle} is not, is that here in each constituent walk $V$, at each step, part of the mass stays at rest and part hops at distance one.

\begin{proof}
    By Proposition~\ref{prp:higherQW_reduction}, it suffices to prove that the eigenvalues of $\widehat{V}(\theta)$ do not satisfy $E_{i}(\theta+\varphi) \equiv \ee^{2\pi \ii \xi}E_{j}(\theta)$ for any $1 \leq i,j \leq 2$ and $0<\varphi,\xi<1$. Since $V$ is a split-step quantum walk, the eigenvalues of $\widehat{V}(\theta)$ are given by 
    \begin{equation}\label{e:spli2}
    \lambda_{\pm}(\theta)=\frac{1}{2}+\frac{1}{2}\cos 2\pi\theta \pm  \frac{\ii}{\sqrt{2}} \sqrt{\frac{\sin^2(2\pi\theta)}{2}  +1-\cos 2\pi\theta}\,,
    \end{equation}
    see \eqref{e:rtsimple}. Let $E_i$ be either $\lambda_+$ or $\lambda_-$. If $E_i(\theta+\varphi)\equiv \ee^{2\pi\ii \xi}E_i(\theta)$ for some $\varphi,\xi$, then $E_i(\varphi)=\ee^{2\pi\ii\xi}E_i(0)=\ee^{2\pi\ii\xi}$ and $E_i(-\varphi)=\frac{1}{\ee^{2\pi\ii\xi}}E_i(0)=\ee^{-2\pi\ii\xi}$. This yields $E_i(\varphi)E_i(-\varphi)=1$, which is a contradiction since by \eqref{e:spli2}, we have $E_i(\varphi)E_i(-\varphi)=E_i(\varphi)^2\neq 1$ for $0<\varphi<1$.

    Thus, if $E_i(\theta+\varphi)\equiv \ee^{2\pi\ii\xi}E_j(\theta)$, then we must have $i\neq j$. Take $E_i=\lambda_+$ and $E_j=\lambda_-$. Then $E_i(\theta+\varphi)\equiv \ee^{2\pi\ii\xi}E_j(\theta)$ implies $\lambda_+(\varphi)=\ee^{2\pi\ii\xi}\lambda_-(0)=\ee^{2\pi\ii\xi}$ and $\lambda_+(\varphi+\frac{1}{2})=\ee^{2\pi\ii\xi}\lambda_-(\frac{1}{2})=-\ii \ee^{2\pi\ii\xi}=-\ii\lambda_+(\varphi)$. Considering the imaginary part of this equation, we get from \eqref{e:spli2},
     \begin{align*}
         -\left(\frac{1}{2}+\frac{1}{2}\cos 2\pi\varphi\right) &=  \frac{1}{\sqrt{2}} \sqrt{\frac{\sin^2(2\pi\varphi)}{2}  +1+\cos 2\pi\varphi}
     \end{align*}
     The RHS of this equation is non-negative for all $\varphi$ and the LHS is non-negative only if $\varphi=\frac{1}{2}$. Thus, $\varphi=\frac{1}{2}$. So $\lambda_+(\theta+\frac{1}{2}) \equiv -\ii \lambda_+(\theta)$. Comparing the imaginary parts of this equation, we get
     \begin{align*}
         -\left(\frac{1}{2}+\frac{1}{2}\cos 2\pi\theta\right) &=  \frac{1}{\sqrt{2}} \sqrt{\frac{\sin^2(2\pi\theta)}{2}  +1+\cos 2\pi\theta} \quad \text{ for all } 0 \leq \theta<1.
     \end{align*}
     As this is clearly not true, we must have $\lambda_+(\theta+\varphi) \not\equiv \ee^{2\pi\ii\xi}\lambda_-(\theta)$ for all $\varphi,\xi$. Similarly $\lambda_-(\theta+\varphi) \not\equiv \ee^{2\pi\ii\xi}\lambda_+(\theta)$. Therefore, by Proposition~\ref{prp:higherQW_reduction}, $U$ is \eqref{eqn:ergodic_general}.
\end{proof}

We finally consider walks with constituent coins of different sizes. Recall the Hadamard walk $U_H$ and Grover walk $U_G$ on $\Z$ given in \eqref{e:uhagro}.

\begin{prp}[Failure of Theorem~\ref{thm:1D_subseq_general} if $d>1$]\label{prp:wasrem}
Let $U$ be the $2d$ quantum walk $U=U_H\mathop\otimes U_G$. Then $U$ has purely absolutely continuous spectrum, but violates \eqref{e:mainlim} for bounded observables, and violates \eqref{e:flo} on all subsequences.
\end{prp}
\begin{proof}   
The Floquet matrix of $U$ is 
    \[
      \widehat{U}(\theta_1,\theta_2)=\widehat{U}_H(\theta_1) \otimes \widehat{U}_G(\theta_2).
    \]
    Therefore, the eigenvalues of $\widehat{U}(\btheta)$ are of the form $\lambda(\theta_1)\mu(\theta_2)$, where $\lambda(\theta_1)$ and $\mu(\theta_2)$ are the eigenvalues of $\widehat{U}_H(\theta_1)$ and $\widehat{U}_G(\theta_2)$, respectively. Since $U_H$ has no flat bands and $U_G$ does not have a zero flat band, then $U$ has no flat bands. However, the Grover walk has $1$ as an eigenvalue (flat band). So  $\widehat{U}(\btheta)$ has an eigenvalue of the form $E_s(\theta_1,\theta_2)=\lambda(\theta_1)$. It follows that
    \begin{equation}\label{e:nosub}
         \frac{\#\{\br\in \LL_{N}^2:E_s(\frac{\br+(0,1)}{N_n})-E_s(\frac{\mathbf{r}}{N_n})=0\}}{N_n^2}= 1. 
    \end{equation}
    This implies that \eqref{e:flo} is not obeyed along any subsequence $(N_n)$.
    
 Moreover, $U$ is not $\ell^\infty$-\eqref{e:mainlim}. To see this, consider an initial state $\psi=\psi_1 \otimes \psi_2$, with $\psi_1 \in \ell^2(\Z)^2$, $\psi_2\in \ell^2(\Z)^3$ a qubit, $\|\psi_i\|_2=1$ and $\psi_2 \neq \delta_0\otimes  \alpha\big(\frac{1}{\sqrt{6}},\frac{-2}{\sqrt{6}},\frac{1}{\sqrt{6}} \big)$ for any $\alpha \in \C$. To simplify the notations, we drop the subscript $N$ from the unitary operators. The marginal distribution of $\mu_{T,\psi}^N$ with respect to the second coordinate is  $\big(\mu_{T,\psi}^N\big)_{y}(r_2)=\sum\limits_{r_1} \mu_{T,\psi}^N(r_1,r_2)=\sum\limits_{r_1}\frac{1}{T}\sum\limits_{k=0}^{T-1}\mathbb{P}(U^k\psi=(r_1,r_2))=\sum\limits_{r_1}\frac{1}{T}\sum\limits_{k=0}^{T-1}\mathbb{P}(U_H^k\psi_1=r_1)\mathbb{P}(U_G^k\psi_2=r_2)=\frac{1}{T}\sum\limits_{k=0}^{T-1} \big(\sum\limits_{r_1}\mathbb{P}(U_H^k\psi_1=r_1)\big)\mathbb{P}(U_G^k\psi_2=r_2)=\frac{1}{T}\sum\limits_{k=0}^{T-1}\mathbb{P}(U_G^k\psi_2=r_2) =\mu_{T,\psi_2}^N(U_G)(r_2)$, where $\mu_{T,\psi_2}^N(U_G)$ is the time-averaged probability distribution of the Grover walk. Now, consider the observable $\tilde{\phi} \in \ell^2(\Z^2)$ defined by $\tilde{\phi}(x,y)= \phi(y)$, where $\phi=\delta_0$ is the observable from the proof of Lemma~\ref{lem:inui}. Then 
 \[
 \langle \tilde{\phi} \rangle_{T,\psi} = \sum_{r_1,r_2 \in \LL_N} \tilde{\phi}(r_1,r_2) \mu_{T,\psi}^N(r_1,r_2)= \sum_{r_2 \in \LL_N} \phi(r_2) \big(\mu_{T,\psi}^N\big)_{y}(r_2)= \langle \phi \rangle_{T,\psi_2},
 \]
 and
 \[
  \langle \tilde{\phi} \rangle= \frac{1}{N^2}\sum_{r_1,r_2 \in \LL_N} \tilde{\phi}(r_1,r_2)= \frac{1}{N}\sum_{r_2 \in \LL_N} \phi(r_2)=\langle \phi \rangle.
 \]
 Since we know from the proof of Lemma~\ref{lem:inui} that $\lim_{N\to\infty}|\lim_{T\to\infty} \langle \phi\rangle_{T,\psi_2}-\langle \phi\rangle|>0$, it  follows that $U_H \otimes U_G$ violates \eqref{e:mainlim} for bounded observables, along every subsequence.
\end{proof}

\subsection{Other models}\label{sec:other} We briefly conclude with two more models of higher dimension.
\subsubsection{Direct sums}
Let $U^{(i)}$ be 1D quantum walks acting on the coordinate directions $\pm e_i$. For definiteness, let $U^{(i)}$ be a $2\times 2$ coined walk. We take $U=\mathop\oplus_{i=1}^d U^{(i)}$. More precisely, $U$ is defined by
\[
U=\begin{pmatrix}
    U^{(1)} & 0 &0&\cdots &0\\
    0  &U^{(2)} &0 &\cdots &0\\
    \vdots &\vdots &\ddots & &\vdots \\
    0  &0 &0 &\cdots &U^{(d)}
\end{pmatrix} \qquad
\text{ and } \qquad
U\psi= \begin{pmatrix}   U^{(1)}\begin{pmatrix}
    \psi_1 \\
    \psi_2
\end{pmatrix} \\
    U^{(2)}\begin{pmatrix}
    \psi_3 \\
    \psi_4
\end{pmatrix} \\
 
   \vdots \\
   U^{(d)} \begin{pmatrix}
    \psi_{2d-1} \\
    \psi_{2d}
\end{pmatrix}
    \end{pmatrix}.
\]
In this model, the components $\psi_{2i-1}$ and $\psi_{2i}$ completely determine the dynamics of the walk along the direction $e_i$, and $\|U^{(i)}\binom{\psi_{2i-1}}{\psi_{2i}}\|=\|\binom{\psi_{2i-1}}{\psi_{2i}}\|$ for all $i$.

We will use this model to prove Proposition~\ref{prp:failhi}.

The eigenvalues of $\widehat{U}(\btheta)$ are the eigenvalues of the individual blocks of $\widehat{U}(\btheta)$, and therefore the set of eigenvalues is given by
\[
\sigma(\widehat{U}(\btheta))=\{E_s(\theta_i): E_s(\theta_i) \in \sigma(\widehat{U}_i(\theta_i)), 1\leq i \leq d, 1 \leq s \leq 2\}.
\]
This model violates \eqref{e:flo} on every subsequence as in \eqref{e:nosub}, by taking $E_s(\btheta):=\lambda_1(\theta_1)$. 
Furthermore, if we start with the initial state $\psi=\delta_{\mathbf{0}} \otimes \delta_1 \in \ell^2(\Z^d) \otimes \C^{2d}$, then the walk concentrates on the set $\{(x,0,0,\ldots,0): x \in \LL_N\}$ implying that the walk violates \eqref{e:mainlim} for regular observables. To see this, take for simplicity $d=2$ and consider $\phi(x,y)=b(y/N)$, where 
$b(\theta)=\frac{1}{2}\big(1+\cos(2\pi\theta)\big)$. Then $\langle \phi\rangle = \frac{1}{N^2}\sum_{(x,y)}\phi(x,y)=\frac{1}{N}\sum_{y=0}^{N-1}b(y/N)\to \int_0^1 b(y)\,\dd y =\frac{1}{2}$. On the other hand, $\langle \phi\rangle_{T,\psi}=\sum_{(x,y)}b(y/N)\mu_{T,\psi}^N(x,y)=b(0)\sum_{(x,0)}\mu_{T,\psi}^N(x,0)=b(0)=1$. This shows that \eqref{e:mainlim} is violated. Moreover, $\phi$ is regular since 
for $f(x,y)=b(y)$, we have $\widehat{f}(\bk)=\int_{\T^2}\ee^{-2\pi\ii \bk\cdot (x,y)}b(y)\,\dd x\dd y=\delta_0(k_1)\widehat{b}(k_2)=\delta_0(k_1)(\frac{\mathbf{1}_{\{0,1,-1\}}(k_2)}{2})$.

This model has no flat bands if each $U^{(i)}$ has no flat bands. In particular, choosing each $U^{(i)}$ to be a standard Hadamard walk proves Proposition~\ref{prp:failhi}.

\subsubsection{A model of Di Franco-McGettrick-Busch} This is a $2d$ model introduced in \cite{Franco_1qubit_higherdim}, as an attempt to reduce the resources necessary for feasible experimental realization of two-dimensional quantum walks. Each time-step of this walk consists of two one-dimensional shift operators, $S^{(x)}$ acting on the $x$-axis and $S^{(y)}$ acting on the $y$-axis, both with step size one in each direction. The quantum walk is then defined by
\[
U=S^{(y)} (C \otimes I) S^{(x)} (C \otimes I),
\]
where $C$ is a $2 \times 2$ unitary coin. An initial state $\delta_{(p_1,p_2)} \otimes f \in \ell^2(\Z^2) \otimes \C^2$ thus reaches $(p_1+1,p_2+1),(p_1+1,p_2-1),(p_1-1,p_2+1)$ and $(p_1-1,p_2-1)$ in a single step, by first moving to $(p_1\pm 1,p_2)$ according to its spin, then moving to $(p_1\pm 1,p_2\pm 1)$. 

Arguing similarly to Proposition \ref{prp:nom_ergodic_nonentangle}, one sees that the walk is not $\ell^\infty$-\eqref{e:mainlim}. Furthermore, since the composition of the shift operators of $U$ yield shifts of the form $S_{\pm e_1\pm e_2}$, then all the terms in $\widehat{U}(\btheta)$ that are functions of $\btheta$ are of the form $\ee^{2\pi \ii (\pm \theta_1 \pm \theta_2)}$. This implies that $\widehat{U}(\theta_1,\theta_2)=\widehat{U}(\theta_1+\frac{1}{2},\theta_2+\frac{1}{2})$ for all $\btheta$, so the walk violates \eqref{e:flo}.

\section{Conclusion: around condition (NRG)}\label{sec:conclu}

We have shown in this paper that \eqref{e:flo} implies strong dynamical ergodicity for quantum walks (Theorem~\ref{thm:pergra}). We investigated \eqref{e:flo} further in Section~\ref{sec:proof1d} in dimension one, and obtained in particular that \eqref{e:flo} is violated for $s$ and $w$ iff there is some rational $0<\varphi<1$ such that $E_s(\theta+\varphi)=E_w(\theta)$ for all $\theta\in\R$, see Lemma~\ref{lem:1D_subseq_2conditions}. We also illustrated the validity of \eqref{e:flo} at length in models of dimension one and higher.

In this section we address three questions:
\begin{itemize}
    \item Is there a similar spectral equivalence between \eqref{e:flo} and the Lebesgue measure of the set $\Omega_{\balpha}(s,w)=\{\btheta\in\T^d:E_s(\btheta+\balpha)=E_w(\btheta)\}$ in higher dimension ?
    \item Is \eqref{e:flo} necessary for \eqref{eqn:ergodic_general} ?
    \item If not, can we characterize \eqref{eqn:ergodic_general} through a condition weaker than \eqref{e:flo} ?
\end{itemize}

We give a negative answer to the second question
(Proposition~\ref{prp:nrgstro}) and partial answers to the first and third ones. For the latter, the condition we get is a bit technical, and it would be desirable to have a subset of walks for which \eqref{e:flo} characterizes \eqref{eqn:ergodic_general}.

\subsection{(NRG) and the Lebesgue measure of band functions' level sets}\label{rem:higher_dim_equiv}
In this section, we restrict ourselves to quantum walks, and therefore the eigenvalue functions $E_s,E_w$ either take values on the real line (for continuous-time random walks) or on the unit circle (for discrete-time quantum walks). This is unlike the case $d=1$, where we considered $E_s$ and $E_w$ to be eigenvalues of a family of normal matrices more generally.

Fix $s,w$ and consider the zero set
\[
\Omega_{\balpha}(s,w) = \{\btheta\in\T^d:E_s(\btheta+\balpha)-E_w(\btheta)=0\} \,.
\]

We introduce two conditions \textbf{(A)} and \textbf{(A$'$)} as follows:

\smallskip

\textbf{(A)} There is no non-zero $ \balpha \in [0,1)^d \cap \Q^d$ such that $\Omega_{\balpha}(s,w)$ has Lebesgue measure greater than zero.

\medskip

\textbf{(A$'$)} There is no non-zero $\balpha \in [0,1)^d$ such that $\Omega_{\balpha}(s,w)$ has Lebesgue measure greater than zero. 

\smallskip

For quantum walks on $\Z^d, d \geq 2$, we conjecture that \eqref{e:flo}$\iff$\textbf{(A)} $\forall \,s,w$. We will prove here that \textbf{(A$'$)}$\implies$\eqref{e:flo}$\implies$\textbf{(A)}.

Before we discuss the equivalence of \eqref{e:flo} and \textbf{(A)}, let us recall that there is a subvariety $S$ such that all eigenvalues $E_u$ are holomorphic on $\T^d\setminus S = \cup_{a=1}^{\mathfrak{f}} \cO_a$, which is a finite union of domains \cite{Wilcox,Tate2019}. The eigenvalue multiplicity is constant on each domain. In particular, $|S|=0$. If $\widetilde{S}=S\cup (S-\balpha)$, then $|\widetilde{S}|=0$ and $\T^d\setminus \widetilde{S} = \cup_{a=1}^{\mathfrak{f}'}\widetilde{\cO}_a$ is a finite union of domains on which both $E_u$ and $E_u(\cdot+\balpha)$ are holomorphic.

\begin{lem}
\eqref{e:flo}$\implies$\emph{\textbf{(A)}} $\forall\,s,w$.
\end{lem}
\begin{proof}
Suppose there exists an $\balpha \in [0,1)^d \cap \Q^d$ such that $\Omega_{\balpha}(s,w)$ has Lebesgue measure greater than zero.
Then by the identity theorem, it follows that $E_s(\btheta+\balpha) \equiv E_w(\btheta)$ over $\widetilde{\cO}_a$ for some $a$.

If $\balpha = (\frac{p_1}{q_1},\dots,\frac{p_d}{q_d})$, $N_n=n\prod_i q_i$ and $m_i = n p_i \prod_{j\neq i} q_i$, this yields $\balpha=\frac{\bm}{N_n}$ on this subsequence of $N$, and restricting ourselves to $\widetilde{\cO}_a$, we get that 
\[
\frac{\#\{\br\in \LL_{N_n}^d: \frac{\br}{N} \in \widetilde{\cO}_a,\, E_s(\frac{\br+\bm}{N_n})-E_w(\frac{\br}{N_n})=0\}}{N^d_n} = \frac{\#\{\br:\frac{\br}{N_n}\in \widetilde{\cO}_a\}}{N_n^d}\to |\widetilde{\cO}_a|> 0\,,
\]    
which clearly violates \eqref{e:flo}. Here we used that $\widetilde{\cO}_a$ is Jordan measurable since its boundary lies in $\widetilde{S}$.
\end{proof}

\begin{lem}\label{lem:wencaierr}
\emph{\textbf{(A$'$)}} $\forall \,s,w\implies$\eqref{e:flo}. 
\end{lem}
Lemma~\ref{lem:wencaierr} strengthens Corollary 1.3 in \cite[erratum]{Wen1}, as assumption \textbf{(A$'$)} is weaker than having no nontrivial periods. Our proof will follow a similar scheme.
\begin{proof} 
Suppose \eqref{e:flo} does not hold. Then for some $s,w$, there exist sequences $(N_n)$ and $(\bm_n \in \LL_{N_n}^d-\{\mathbf{0}\})$, $\frac{\bm_n}{N_n} \to \balpha \in [0,1]^d$ such that
\begin{equation}\label{e:rem35}
  \frac{\# \Big\{\br  \in \LL_{N_n}^d:   E_s\big( \frac{\br+\bm_n}{N_n}\big)-E_w\big( \frac{\br}{N_n}\big)=0\Big\}}{N_n^d}\nto 0 \,. 
\end{equation}
Recall the negligible set $\widetilde{S}$ defined above, with $\T^d\setminus \widetilde{S}=\cup_{a=1}^{\mathfrak{f}'} \widetilde{\cO}_a$. We let $\widetilde{\cO}_a^\varepsilon\subset \widetilde{\cO}_a$ be an $\varepsilon$ shrinking and let $\widetilde{S}^\varepsilon = \T^d\setminus \cup_{a=1}^{\mathfrak{f}'} \widetilde{\cO}_a^\varepsilon$. Then
\[
\frac{\#\{\br\in \LL_{N_n}^d:\frac{\br}{N_n}\in \widetilde{S}^\varepsilon,\,E_s(\frac{\br+\bm_n}{N_n})=E_w(\frac{\br}{N_n})\}}{N_n^d} \le \frac{\#\{\frac{\br}{N_n}\in \widetilde{S}^\varepsilon\}}{N_n}\xrightarrow{n\to\infty} |\widetilde{S}^\varepsilon| \xrightarrow{\varepsilon\downarrow 0} 0\,.
\]
Consequently, \eqref{e:rem35} implies (up to taking a subsequence) that there is some $a$ such that
\begin{equation}\label{e:rem351}
\lim_{\varepsilon\downarrow 0}\lim_{n\to\infty} \frac{\#\{\br\in \LL_{N_n}^d:\frac{\br}{N_n}\in \widetilde{\cO}_a^\varepsilon,\,E_s(\frac{\br+\bm_n}{N_n})=E_w(\frac{\br}{N_n})\}}{N_n^d} =\kappa>0\,.
\end{equation}

Clearly $|\widetilde{\cO}_a|>0$, otherwise the previous argument would have given a zero limit.

As in \cite[erratum]{Wen1}, we consider three cases.

\smallskip

\noindent \textbf{Case (i).} $\balpha\neq \mathbf{0}$. (This is Case 2 in \cite[erratum]{Wen1}).

Let $S_a = \{\btheta\in \widetilde{\cO}_a:E_s(\btheta+\balpha)=E_w(\btheta)\}$. If $|S_a|=0$, then $S_a\cap \overline{\widetilde{\cO_a^\varepsilon}}$ can be covered by finitely many small balls whose total Lebesgue measure tends to zero. If $E_s(\frac{\br+\bm}{N})=E_w(\frac{\br}{N})$, then because $\frac{\bm}{N}\to \balpha$ and $E_s$ is continuous, we get that $\frac{\br}{N}$ must be in these balls covering $S_a$, see \cite{Wen1}. Following (6) and (7) in \cite[erratum]{Wen1}, this implies that the LHS of \eqref{e:rem351} tends to zero. Consequently, \eqref{e:rem351} implies that $|S_a|>0$. Using the identity theorem, we get $E_s(\btheta+\balpha)\equiv E_w(\btheta)$ on $\widetilde{\cO}_a$, which has positive measure. So we showed that \textbf{(A$'$)} is violated, except that $\balpha\in [0,1]^d$ here, instead of $[0,1)^d$. But this can be immediately fixed: if $\alpha_1=1$ for example, then $E_s(\alpha_1-1,\alpha_2,\dots,\alpha_d)=E_s(\alpha_1,\alpha_2,\dots,\alpha_d)$ since $E_s$ is a function on the torus. This completes the proof in this case.

\smallskip

\noindent \textbf{Case (ii).} $\balpha=\mathbf{0}$ and $s=w$. (This is case $1_2$ in \cite[erratum]{Wen1}. Here $\cO_a=\widetilde{\cO}_a$).

It is shown in \cite[erratum]{Wen1} that there exists $0\neq T\in [0,1]^d$ such that if $\widetilde{S}_a = \{\btheta\in \cO_a : \nabla E_s(\btheta)\cdot T=0\}$ and $|\widetilde{S}_a|=0$, then the LHS of \eqref{e:rem351} converges to zero. This uses Rolle's theorem for the real-valued $E_s$, but as explained in Remark~\ref{rem:uniher}, it extends to our unitary setting. Consequently, \eqref{e:rem351} implies that $|\widetilde{S}_a|>0$. By the identity theorem, this implies that $\nabla E_s(\btheta)\cdot T=0$ on $\cO_a$.

 Note that $\nabla E^s(\btheta)\cdot T\equiv 0$ implies $E_s(\btheta+cT)=E_s(\btheta)$ for all $\btheta,\btheta+cT \in \mathcal{O}_a$, where $c\in (0,1)$. Choosing $c$ small enough we find a non-zero $\balpha=cT \in [0,1)^d$ violating \textbf{(A$'$)}.

\smallskip

\noindent \textbf{Case (iii).} $\balpha=\mathbf{0}$ and $s\neq w$. (This is case $1_1$ in \cite[erratum]{Wen1}. Here $\cO_a=\widetilde{\cO}_a$).

Let $S_a=\{\btheta\in \cO_a:E_s(\btheta)=E_w(\btheta)\}$. As before, \cite[erratum]{Wen1}, \eqref{e:rem351} and the identity theorem imply that $|S_a|>0$ and $E_s(\btheta)\equiv E_w(\btheta)$ on $\cO_a$.

In this case, substituting $E_s$ in place of $E_w$ in \eqref{e:rem351}, we get that
\[
  \frac{\# \Big\{\br  \in \LL_n^d: \frac{\br}{N_n} \in\mathcal{O}_{a}^\varepsilon,  E_s\big( \frac{\br+\bm_n}{N_n}\big)-E_s\big( \frac{\br}{N_n}\big)=0\Big\}}{N_n^d}\to \kappa>0\,. 
\]
and therefore, we are back in case (ii), which is already proved. 
\end{proof}

\subsection{On the necessity of (NRG)}
In this subsection we consider a walk which is physically not very interesting (just moving in one direction). The purpose is to show that \eqref{e:flo} is in general not necessary for ergodicity.

\begin{prp}\label{prp:nrgstro}
Consider the walk $U = \begin{pmatrix} S_{-1}&0\\0&-S_{-1}\end{pmatrix}$. Then
\begin{enumerate}[\rm(1)]
\item $U$ violates \eqref{e:flo}.
\item $U$ satisfies $\ell^\infty$-\eqref{eqn:ergodic_general}.
\item $ \lim_{T\to\infty} \mu_{T,\psi}^N(r) = \frac{1}{N}$ is the uniform measure for any $N$, yet the eigenvalues of $U_N$ are not distinct for even $N$.
\end{enumerate}
\end{prp}
\begin{proof}
(1) is clear since the eigenvalues are $E_1(\theta) = \ee^{2\pi\ii\theta}$ and $E_2(\theta)=-\ee^{2\pi\ii\theta} = \ee^{2\pi\ii(\theta+\frac{1}{2})} = E_1(\theta+\frac{1}{2})$ for all $\theta$. So \eqref{e:flo} is violated for even $N$ if $m=\frac{N}{2}$.

Next, we have $(U_N\psi)(r) = \begin{pmatrix} (S_{-1}\psi_1)(r)\\-(S_{-1}\psi_2)(r)\end{pmatrix} = \begin{pmatrix} \psi_1(r+1)\\ -\psi_2(r+1)\end{pmatrix}$. More generally, $(U^k_N\psi)(r) = \begin{pmatrix} \psi_1(r+k)\\ (-1)^k \psi_2(r+k)\end{pmatrix}$. So $|(U^k_N\psi)_i(r)|^2 = |\psi_i(r+k)|^2$. Thus,
\[
\frac{1}{T}\sum_{k=0}^{T-1}\langle U^k_N\psi, a U^k_N\psi \rangle = \sum_{r\in\LL_N} \sum_{i=1}^2 a_i(r) \frac{1}{T}\sum_{k=0}^{T-1} |\psi_i(r+k)|^2 \,.
\]
For $T=N$, we have $\sum_{k=0}^{T-1}|\psi_i(r+k)|^2=\|\psi_i\|^2$. More generally, for $T=nN$ this gives $n\|\psi_i\|^2$, and if $T=nN+t$, $t<N$, then $ \frac{1}{T}\sum_{k=0}^{T-1} |\psi_i(r+k)|^2 = \frac{1}{nN+t}(n\|\psi_i\|^2+c_{\psi_i}(r,t))$ for some $|c_{\psi_i}(r,t)|\le \|\psi_i\|^2\le 1$, so $\frac{1}{T}\sum_{k=0}^{T-1} |\psi_i(r+k)|^2  \to \frac{1}{N}\|\psi_i\|^2$ as $n\to\infty$. Thus,
\[
\lim_{T\to\infty} \frac{1}{T}\sum_{k=0}^{T-1} \langle U^k_N \psi, aU^k_N\psi\rangle = \|\psi_1\|^2\langle a_1\rangle + \|\psi_2\|^2\langle a_2\rangle \,.
\]

On the other hand, here $\widehat{U}(\theta)$ has eigenvectors $\begin{pmatrix}1\\0\end{pmatrix},\begin{pmatrix}0\\1\end{pmatrix}$, so $P_{E_1}(\theta) = \begin{pmatrix}1&0\\0&0\end{pmatrix}$ and $P_{E_2}(\theta)=\begin{pmatrix} 0&0\\0&1\end{pmatrix}$. Thus, $P_{E_1}(\theta)\widehat{\psi}(r) = \begin{pmatrix} \widehat{\psi}_1(r)\\0\end{pmatrix}$ and $P_{E_2}(\theta)\widehat{\psi}(r) = \begin{pmatrix} 0\\ \widehat{\psi}_2(r)\end{pmatrix}$. Hence,
\[
\sum_{r\in\LL_N}\Big(\Big|\Big[P_{E_1}\Big(\frac{r}{N}\Big)\widehat{\psi}(r)\Big]_1\Big|^2 + \Big|\Big[P_{E_2}\Big(\frac{r}{N}\Big)\widehat{\psi}(r)\Big]_1\Big|^2\Big) = \sum_{r\in\LL_N} |\widehat{\psi}_1(r)|^2 = \|\widehat{\psi}_1\|^2 = \|\psi_1\|^2
\]
and analogously if the subscript $1$ above is replaced by $2$. Thus,
\[
\langle a\rangle_\psi = \|\psi_1\|^2\langle a_1\rangle + \|\psi_2\|^2\langle a_2\rangle \,.
\]
In particular, \eqref{eqn:ergodic_general} holds (without needing $N\to\infty$).

Finally, we showed that $\frac{1}{T}\sum_{k=0}^{T-1}|(U_N^k\psi)_i(r)|^2 \to \frac{1}{N}\|\psi_i\|^2$, so $\mu_{T,\psi}^N(r) \to \frac{1}{N}\sum_{i=1}^2\|\psi_i\|^2 = \frac{1}{N}$ as stated. Moreover, by Lemma 2.1, the eigenvalues of $U_N$ are $E_s(\frac{-r}{N})$, $s=1,2$, $r\in \LL_N$. In particular, $E_2(\frac{-r}{N}) = E_1(\frac{-r}{N}+\frac{1}{2})=E_1(\frac{N-2r}{2N})= E_1(\frac{-r'}{N})$ for $r'=r-\frac{N}{2}$ and $r\ge \frac{N}{2}$ showing the eigenvalues are not distinct.
\end{proof}

It may be tempting from Proposition~\ref{prp:nrgstro} to think that the ``right'' condition is \eqref{e:flo} for all $s=w$. But this is not true.

\begin{prp}
\begin{enumerate}[\rm(1)]
\item \eqref{e:flo} for all $s=w$ does not imply $\ell^\infty$-\eqref{e:mainlim}.
\item \eqref{e:flo} for all $s\neq w$ does not imply $\ell^\infty$-\eqref{e:mainlim}.
\end{enumerate}
\end{prp}
\begin{proof}
For (1), take $U = \begin{pmatrix} 0&S_{-2}\\S_4&0\end{pmatrix}$ and use Proposition~\ref{prp:offdia}. For (2), take e.g. $U = \begin{pmatrix} S_{-1}&0\\0&S_{2}\end{pmatrix}$. Then $E_1(\theta+\alpha)=E_2(\theta)$ implies $\ee^{2\pi\ii(\theta+\alpha)}=\ee^{-4\pi\ii\theta}$, so $\ee^{2\pi\ii\alpha} = \ee^{-6\pi\ii\theta}$, which has only finitely many solutions $\theta=\frac{3-\alpha}{3},\frac{2-\alpha}{3},\frac{1-\alpha}{3}$ for $0<\alpha<1$. So \eqref{e:flo} is satisfied for $s\neq w$. However, if $N$ is even, the initial state $\psi = \delta_0\mathop\otimes\begin{pmatrix} 0\\1\end{pmatrix}$ satisfies that $U^k \psi$ is supported only on even integers in $\LL_N$, so taking $\phi = \mathbf{1}_{\text{even}}$ violates \eqref{e:mainlim}. 
\end{proof}

\begin{rem}
Looking further into the example of Proposition~\ref{prp:nrgstro}, we see that the only violation of \eqref{e:flo} comes from $E_1(\theta+\frac{1}{2})=E_2(\theta)$, i.e. $N$ even, $m=\frac{N}{2}$ and $s\neq w$. More precisely, if we look back at the proof of Theorem~\ref{thm:pergra}, in the presence of \eqref{e:flo}, all higher modes of the limiting symbol $b$ vanish and we are left with $b_0$, which is responsible for the uniform average. Here on the other hand, if $N$ is even and $S_r=\{(m,s,t):E_s(\frac{r+m}{N})=E_t(\frac{r}{N})\}$, then $S_r=\big\{(0,1,1),(0,2,2),\big(\frac{N}{2},1,2\big),\big(\frac{N}{2},2,1\big)\big\}$. So the symbol $b$ takes the form $b=b_0+b_{N/2}$. The reason why we don't see $b_{N/2}$ asymptotically in this example is due to a vanishing of eigenprojections, which is quite subtle. In fact, $b_{N/2}(v;i,\ell)=\sum_{j=1}^2P_{E_1}(i,j)\widehat{a}_j(\frac{N}{2})P_{E_2}(j,\ell)e_{N/2}^{(N)}(v) + \sum_{j=1}^2P_{E_2}(i,j)\widehat{a}_j(\frac{N}{2})P_{E_1}(j,\ell)e_{N/2}^{(N)}(v)$. By definition of $P_{E_i}$ here, each of the previous sums contain $\delta_{i=j=1}\delta_{j=\ell=2}$ and $\delta_{i=j=2}\delta_{j=\ell=1}$. Each of these products is zero. Hence, $b_{N/2}=0$.

We will see in Theorem~\ref{thm:fqenrg} that this happens more generally.
\end{rem}

\subsection{A converse to (FQE)}
Our aim here is to show that \eqref{eqn:ergodic_general} implies either \eqref{e:flo}, or a condition on the Floquet spectral projections. More precisely,

\begin{thm}\label{thm:fqenrg}
Suppose $U$ on $\Z^d$ satisfies \eqref{eqn:ergodic_general} for regular observables. Let $\Omega_{\balpha}(s,w) = \{\btheta\in \T^d:E_s(\btheta+\balpha)=E_w(\btheta)\}$. Then for each $s,w$ and each nonzero $\balpha\in\Q^d$, either
\begin{enumerate}[\rm(1)]
\item $\Omega_{\balpha}(s,w)$ has Lebesgue measure zero,
\item or $\Omega_{\balpha}(s,w)$ has positive Lebesgue measure, and $P_s(\btheta+\balpha)(i,j)P_w(\btheta)(j,k)=0$ for a.e. $\btheta\in \Omega_{\balpha}(s,w)$ and all $i,j,k$. In particular, $P_s(\btheta+\balpha)P_w(\btheta)=0$ is the zero matrix for a.e. $\btheta\in \Omega_{\balpha}(s,w)$.
\end{enumerate}
\end{thm}

If \eqref{e:flo} is satisfied, it is case (1) which occurs for all $s$ and $w$ by \S\,\ref{rem:higher_dim_equiv}. For the walk of Proposition~\ref{prp:nrgstro}, $\Omega_\alpha(1,1)=\Omega_\alpha(2,2)=\emptyset$, while $\Omega_{1/2}(1,2)=\Omega_{1/2}(2,1)=\T$, and $P_1P_2=0$. To prove the theorem, we will only use that \eqref{eqn:ergodic_general} holds for qubit initial states $\psi = \delta_{\mathbf{0}}\mathop\otimes f$ for $f\in \C^\nu$.

\begin{proof}
We first note that if $\psi = \delta_{\mathbf{0}}\mathop\otimes f$ for $f\in \C^\nu$ then $\widehat{\psi} = (f_1 \widehat{\delta_{\mathbf{0}}},\dots,f_\nu \widehat{\delta_{\mathbf{0}}})^T = N^{-d/2}(f_1,\dots,f_\nu)^T$, so $\widehat{\psi}_\ell = N^{-d/2}f_\ell$. Also, $\langle \psi, \opn(b')\psi\rangle = \sum_{i=1}^\nu f_i [\opn(b')\psi]_i(\mathbf{0}) = \sum_{i=1}^\nu f_i \sum_{\br}\sum_\ell b'(\mathbf{0},\br;i,\ell)\widehat{\psi}_\ell(\br)e_{\br}^{(N)}(\mathbf{0})$. Thus, 
\begin{multline*}
\langle \psi, \opn(b')\psi\rangle = \frac{1}{N^d}\sum_{i,\ell=1}^\nu f_i \sum_{\br\in\LL_N^d} b'(\mathbf{0},\br;i,\ell)f_\ell \\
= \frac{1}{N^d}\sum_{i,\ell=1}^\nu f_i \sum_{\br\in\LL_N^d} \sum_{\bm\neq 0}\sum_j \sum_{s,t} \mathbf{1}_{A_{\bm}}(\br,s,t)P_{E_s}\Big(\frac{\br+\bm}{N}\Big)(i,j)\widehat{a}_j(\bm)P_{E_t}\Big(\frac{\br}{N}\Big)(j,\ell)e_{\bm}^{(N)}(\mathbf{0})f_\ell.
\end{multline*}
If $a_j(\bk) = c_j \ee^{2\pi\ii \bk \cdot \bm_0/N}$ for some $c_j\in \C$, then $\widehat{a}_j(\bm) = \frac{1}{N^{d/2}}\sum_{\bk} \ee^{-2\pi\ii \bk\cdot \bm/N} a_j(\bk) = c_j N^{d/2} \delta_{\bm_0}(\bm)$. Denoting $c= (c_1,\dots,c_\nu)^T$, we get that $\langle \psi, \opn(b')\psi\rangle$ equals
\begin{equation}\label{e:someterm}
\frac{1}{N^{d}}\sum_{i,\ell=1}^\nu f_i \sum_{\br\in\LL_N^d} \sum_{j=1}^\nu\sum_{s,t}  \mathbf{1}_{A_{\bm_0}}(\br,s,w)P_{E_s}\Big(\frac{\br+\bm_0}{N}\Big)(i,j)c_jP_{E_t}\Big(\frac{\br}{N}\Big)(j,\ell)f_\ell\,,
\end{equation}
where the sums over $s,t$ run up to $\nu'(\frac{\br+\bm_0}{N})$ and $\nu'(\frac{\br}{N})$, respectively, with $\nu'(\btheta)$ the number of distinct eigenvalues of $\widehat{U}(\btheta)$.

\eqref{eqn:ergodic_general} implies that $\lim\limits_{N\to\infty}\langle\psi,\opn(b')\psi\rangle=0$, so \eqref{e:someterm} has a zero limit as $N\to\infty$. Now let $\mathbf{0}\neq \balpha\in \Q^d$. Then $\balpha=(\frac{p_1'}{q_1},\dots,\frac{p_d'}{q_d}) = (\frac{p_1}{q},\dots,\frac{p_d}{q})=\frac{\mathbf{p}}{q}$ for $q=\prod_{i=1}^d q_i$ and $p_i = p_i'\prod_{j\neq i} q_i$. Moreover, \eqref{e:someterm} tends to zero in particular along the subsequence $N = qn$. Let $\bm_0 = (p_1n,\dots,p_dn)\neq\mathbf{0}$. Then $A_{\bm_0} = \{(\br,s,w):E_s(\frac{\br}{N} + \frac{\mathbf{p}}{q})=E_w(\frac{\br}{N})\}$. So $(\br,s,w)\in A_{\bm_0}\iff \br/N\in \Omega_{\balpha}(s,w)$, where $\Omega_{\balpha}(s,w)=\{\btheta\in \T^d: E_s(\btheta+\balpha)=E_w(\btheta)\}$. In other words, $\mathbf{1}_{A_{\bm_0}}(\br,s,w) = \mathbf{1}_{\Omega_{\balpha}(s,w)}(\frac{\br}{N})$. Note that $\Omega_{\balpha}(s,w)$ is a closed set since the eigenvalue functions are continuous.  

We claim that the limit of \eqref{e:someterm} as $N=nq\to\infty$, for this $\mathbf{m}_0=n\mathbf{p}$, is
\begin{equation}\label{e:limterm}
\sum_{i,\ell=1}^\nu f_i \sum_{j,s,w=1}^\nu \int_{\Omega_{\balpha}(s,w)} P_s(\btheta+\balpha)(i,j)c_jP_w(\btheta)(j,\ell)f_\ell\,\dd\btheta \,.
\end{equation} 

To see this, first suppose that $\Omega_{\balpha}(s,w)$ has measure zero for some $s,w$. As it is compact in $\T^d$, we may cover it by a finite union of balls $\cup_{i=1}^f B_i$ of measure $<\delta$. Then using $|P_s(i,j)|\le 1$, the term in \eqref{e:someterm} is bounded by $\frac{c \cdot \# \{\br\in\LL_N^d:\frac{\br}{N}\in \Omega_{\balpha}(s,w)\}}{N^d}\le \frac{c \cdot \# \{\br\in\LL_N^d:\frac{\br}{N}\in \cup B_i\}}{N^d} \to c\, |\cup B_i|<c\delta$. As $\delta$ is arbitrary, this shows that this term in \eqref{e:someterm} converges to zero. But if $|\Omega_{\balpha}(s,w)|=0$, then the corresponding integral in \eqref{e:limterm} is also equal to zero. This shows the claim in this case. 

So suppose that $\Omega_{\balpha}(s,w)$ has positive Lebesgue measure. Recall the subvariety $S\subset\T^d$ (of zero measure) outside of which all eigenvalues are analytic, with $\T^d\setminus S$ a finite union of connected components $\cO_a$, $a\le \mathfrak{f}$. Let $\widetilde{S} = S\cup (S-\balpha)$. Then $|\widetilde{S}|=0$. Moreover, if $\frac{\br}{N}\in\T^d\setminus \widetilde{S}$, then $\frac{\br}{N}\in \cO_a$ and $\frac{\br}{N}+\balpha\in \cO_b$ for some $a,b$. 

Going back to \eqref{e:someterm}, the contribution of $\br$ with $\frac{\br}{N}\in \widetilde{S}$ goes to zero as before. The remaining contribution can be written as
\begin{multline*}
\frac{1}{N^d}\sum_{i,\ell=1}^\nu f_i \sum_{a,b=1}^{\mathfrak{f}}\sum_{\br\in\LL_N^d} \mathbf{1}_{\cO_a}\Big(\frac{\br}{N}\Big) \mathbf{1}_{\cO_b}\Big(\frac{\br}{N}+\balpha\Big)\sum_{j=1}^\nu \sum_{s=1}^{\nu_b'}\sum_{t=1}^{\nu_a'}\mathbf{1}_{\Omega_{\balpha}(s,w)}\Big(\frac{\br}{N}\Big)\\
\cdot P_{E_s}\Big(\frac{\br}{N}+\balpha\Big)(i,j)c_jP_{E_t}\Big(\frac{\br}{N}\Big)(j,\ell)f_\ell
\end{multline*}
where we used that $\nu'$ is constant over each connected component $\cO_a$.

Note that $\partial (\cO_a \cap (\cO_b-\balpha))\subseteq \partial \cO_a \cup \partial (\cO_b-\balpha)$, implying $|\partial (\cO_a \cap (\cO_b-\balpha))| = 0$. On the other hand, by the identity theorem, if $\Omega_{\balpha}(s,w)$ has a positive measure, then the eigenvalues coincide on a full connected component. Consequently, if $\Omega_{\balpha}(s,w)$ is not negligible, then the boundary of $\Omega_{\balpha}(s,w)$ is contained in $\widetilde{S}$, which has measure zero. This and the previous remark show that $ \mathbf{1}_{\cO_a}\big(\frac{\br}{N}\big) \mathbf{1}_{\cO_b}\big(\frac{\br}{N}+\balpha\big)\mathbf{1}_{\Omega_{\balpha}(s,w)}\big(\frac{\br}{N}\big)$ is Riemann integrable. On the other hand, the spectral projections $P_{E_u}(\btheta)$ are continuous in $\btheta\in \cO_c$ (and even holomorphic, see \cite{Tate2019}). Therefore, the previous Riemann sum converges to
\[
\sum_{i,\ell=1}^\nu f_i \sum_{a,b=1}^{\mathfrak{f}}\int_{\T^d} \mathbf{1}_{\cO_a}(\btheta) \mathbf{1}_{\cO_b}(\btheta+\balpha)\sum_{j=1}^\nu \sum_{s=1}^{\nu_b'}\sum_{t=1}^{\nu_a'}\mathbf{1}_{\Omega_{\balpha}(s,w)}(\btheta) P_{E_s}(\btheta+\balpha)(i,j)c_jP_{E_t}(\btheta)(j,\ell)f_\ell\,\dd\btheta.
\]
We finally expand $P_{E_u} = \sum_w P_w$, where the sum runs over $E_w=E_u$. So the sums over $s,t$ now run up to $\nu$. Then we use that $\int_{\T^d} = \int_{\T^d\setminus \widetilde{S}}$ and $\sum_{a,b}\mathbf{1}_{\cO_a}(\btheta) \mathbf{1}_{\cO_b}(\btheta+\balpha)=1$. After rearranging, the integral reduces to \eqref{e:limterm}.

So in summary, \eqref{eqn:ergodic_general} implies that \eqref{e:limterm} equals zero for any $f,c\in \C^\nu$ and any nonzero $\balpha\in \Q^d$.

If $f = \delta_i$ and $c=\delta_j$, this gives
\[
\sum_{s,w=1}^\nu \int_{\Omega_{\balpha}(s,w)} P_s(\btheta+\balpha)(i,j)P_w(\btheta)(j,i)\,\dd\btheta = 0
\]
for any $i,j\le \nu$ and nonzero $\balpha\in\Q^d$. 

Moreover, since $P_s(\boldsymbol{\varphi})(i,j) = (P_s(\boldsymbol{\varphi})\delta_j)(i)=\langle  g^s_{\boldsymbol{\varphi}},\delta_j\rangle g^s_{\boldsymbol{\varphi}}(i) = \overline{g^s_{\boldsymbol{\varphi}}(j)}g^s_{\boldsymbol{\varphi}}(i)$, we have $P_s(\boldsymbol{\varphi})(i,i)\ge 0$ for any $i$, so choosing $j=i$, we get a sum of positive terms vanishing. Hence,
\[
\int_{\Omega_{\balpha}(s,w)} P_s(\btheta+\balpha)(i,i)P_w(\btheta)(i,i)\,\dd\btheta=0
\]
for each $s,w,i$ and nonzero $\balpha\in\Q^d$. And since the integrand is $\ge 0$, we must have for each $s,w$ and nonzero $\balpha\in\Q^d$,
\begin{enumerate}
\item Either $|\Omega_{\balpha}(s,w)|=0$,
\item or $|\Omega_{\balpha}(s,w)|>0$ and $P_s(\btheta+\balpha)(i,i)P_w(\btheta)(i,i)=0$ for a.e. $\btheta\in \Omega_{\balpha}(s,w)$ and all $i$.
\end{enumerate}

Note that for a projector $P$, using $P=P^2$, $P^\ast=P$ and Cauchy-Schwarz, we have $|P(i,j)| = |\langle \delta_i, P\delta_j\rangle| \le \|P\delta_i\|\|P\delta_j\| = \sqrt{P(i,i)P(j,j)}$. This implies that $|P_1(i,j)P_2(j,\ell)|\le \sqrt{P_1(i,i)P_1(j,j)P_2(j,j)P_2(\ell,\ell)}$. It thus follows in case (2) that $P_s(\btheta+\balpha)(i,j)P_w(\btheta)(j,k)=0$ for a.e. $\btheta\in \Omega_{\balpha}(s,w)$ and all $i,j,k$. This in turn implies that $P_s(\btheta+\balpha)P_w(\btheta)=0$ is the zero matrix for a.e. $\btheta\in \Omega_{\balpha}(s,w)$.
\end{proof}

\appendix

\section{Notations and models in the literature}\label{app}

\subsection{Other notations}
The following table gives a dictionary between the different notations found in the literature.

\begin{table}[h!]
\scalebox{0.8}{ 
\begin{tabular}{|p{3cm}|c|c|c|}
Notation&Bra-ket& Tensor& Vector\\
\hline
Hilbert space $\mathscr{H}$ & $\mathcal{H}_C\mathop\otimes\mathcal{H}_P$& $\ell^2(\Z^d)\mathop\otimes \C^\nu$& $(\ell^2(\Z^d))^\nu = \ell^2(\Z^d,\C^\nu)$\\
Basis element& $\mid j,\bk\rangle$ & $\delta_{\bk}\otimes\delta_j$& $\delta_{j,\bk}$ \\
Vector $\psi\in \mathscr{H}$& $\mid \psi\rangle = \sum\limits_{j=1}^\nu\sum\limits_{\bk\in \Z^d} \psi_{j,\bk} \mid j,\bk\rangle$& $\psi = \sum\limits_{j=1}^\nu\sum\limits_{\bk\in\Z^d} \psi_j(\bk)\delta_{\mathbf{k}}\otimes\delta_j$& $\psi = \begin{pmatrix} \psi_1\\ \vdots\\ \psi_\nu\end{pmatrix}$, $\psi_i\in \ell^2(\Z^d)$\\
Coin operator $\mathcal{C}=C\otimes I$, $C=(c_{i,j})$ &$\mathcal{C}\mid\psi\rangle = \sum\limits_{i,j=1}^\nu\sum\limits_{\bk\in\Z^d} c_{i,j}\psi_{j,\bk}\mid i,\mathbf{k}\rangle$&$\mathcal{C}\psi= \sum\limits_{i,j=1}^\nu\sum\limits_{\bk\in\Z^d} c_{i,j}\psi_{j}(\bk)\delta_{\bk}\mathop\otimes\delta_i$& $\mathcal{C}\psi = \begin{pmatrix} \sum_{j=1}^\nu c_{1,j}\psi_j\\ \vdots\\ \sum_{j=1}^\nu c_{\nu,j}\psi_j\end{pmatrix}$\\
Shift operators, $d=1$, $\nu=2$&$\mathcal{S} = \sum\limits_{j=0}^1\sum\limits_{k\in \Z} \mid j,k+(-1)^j\rangle\langle j,k\mid $& $\mathcal{S}\psi = \sum\limits_{k\in \Z}\delta_k\mathop\otimes \begin{pmatrix} \psi_1(k-1)\\ \psi_2(k+1)\end{pmatrix}$& $\mathcal{S}\psi = \begin{pmatrix} S_1\psi_1\\ S_{-1}\psi_2\end{pmatrix}$\\
& $\cS_+=\mid +\rangle \langle +\mid \mathop\otimes S_1 +  \mid -\rangle \langle -\mid\mathop\otimes \mathbf{1}_{\Z}$ & $\cS_+\psi = \sum_{k\in\Z}\delta_k\mathop\otimes \begin{pmatrix} \psi_1(k-1)\\\psi_2(k)\end{pmatrix}$& $\cS_+\psi = \begin{pmatrix} S_1\psi_1\\ \psi_2\end{pmatrix}$\\
& $\cS_-=\mid -\rangle \langle -\mid \mathop\otimes S_{-1} +  \mid +\rangle \langle +\mid\mathop\otimes \mathbf{1}_{\Z}$ & $\cS_-\psi = \sum_{k\in\Z} \delta_k\mathop\otimes\begin{pmatrix} \psi_1(k)\\\psi_2(k+1)\end{pmatrix}$& $\cS_-\psi = \begin{pmatrix} \psi_1\\ S_{-1}\psi_2\end{pmatrix}$\\
\hline 
\end{tabular}
}
\caption{In this paper we mostly follow the convention in the last column.}
\end{table}

In the table, the bra-ket convention is coin-position $\mathcal{H}_C\mathop\otimes \mathcal{H}_P$ as in \cite{Aha,Port}. Some use $\mid j \rangle \mid \bk\rangle :=\, \mid j,\bk\rangle$. Also common is the position-coin $\mathcal{H}_P\mathop\otimes \mathcal{H}_C$ as in \cite{ACSW}. We follow that order with the tensor notation (middle column).

In row 2, $\delta_{j,\bk}:=\begin{pmatrix}0&\cdots&0&\delta_{\bk}&0&\cdots&0 \end{pmatrix}^T$ has $\delta_{\bk}\in \ell^2(\Z^d)$ in the $j$th coordinate.

To clarify the notation further, let us verify two relations. For example, 
\[
\cS_+(\alpha_k \mid 0,k\rangle + \beta_k\mid 1,k\rangle) = \cS_+ \delta_k\mathop\otimes \binom{\alpha_k}{\beta_k} =\delta_{k+1} \mathop\otimes \binom{\alpha_k}{0}+\delta_k\mathop\otimes\binom{0}{\beta_k}
\]
so $\cS_+\psi = \sum_k \cS_+\delta_k\mathop\otimes \binom{\psi_1(k)}{\psi_2(k)}  = \sum_k(\delta_{k+1}\mathop\otimes\binom{\psi_1(k)}{0}+\delta_k\mathop\otimes\binom{0}{\psi_2(k)})=\sum_k\delta_k \mathop\otimes\binom{\psi_1(k-1)}{\psi_2(k)}$.

Similarly, $\mathcal{S} = \sum_k (|0,k+1\rangle \langle 0,k| + |1,k-1\rangle\langle 1,k|)$ satisfies $\mathcal{S} |0,k\rangle = |0,k+1\rangle$ and $\mathcal{S}|1,k\rangle=|1,k-1\rangle$, from which the relation follows.

Note that $\mathcal{S} = \cS_+\cS_-$.

\subsection{Split-step and coined walks}\label{sec:spli}
A \emph{split-step quantum walk} is a walk of the form $U = \cS_+\cC_1\cS_-\cC_2$ for some $2\times 2$ coins $C_i$, see \cite{ACSW}.  If $C_1=(c_{ij})$ and $C_2=(d_{ij})$, this gives
\begin{align*}
    U\psi &= \cS_+\cC_1\cS_-\binom{d_{11}\psi_1+d_{12}\psi_2}{d_{21}\psi_1+d_{22}\psi_2}=\cS_+\cC_1\binom{d_{11}\psi_1+d_{12}\psi_2}{d_{21}S_{-1}\psi_1+d_{22}S_{-1}\psi_2}\\
    &= \binom{c_{11}(d_{11}S_1\psi_1+d_{12}S_1\psi_2)+c_{12}(d_{21}\psi_1+d_{22}\psi_2)}{c_{21}(d_{11}\psi_1+d_{12}\psi_2)+c_{22}(d_{21}S_{-1}\psi_1+d_{22}S_{-1}\psi_2)} \\
    &= \begin{pmatrix} c_{11}d_{11}S_1+c_{12}d_{21}&c_{11}d_{12}S_1+c_{12}d_{22}\\\ c_{21}d_{11}+c_{22}d_{21}S_{-1}&c_{21}d_{12}+c_{22}d_{22}S_{-1}\end{pmatrix}\begin{pmatrix}\psi_1\\\psi_2\end{pmatrix}
\end{align*}
This clearly is a special case of our setting. In particular, $U_{11}= c_{11}d_{11}S_1+c_{12}d_{21}S_0$.

If $C_1=I$ and $C_2=C$, this gives $\cS_+\cS_-\cC=\cS\cC$, which is a \emph{coined walk}. Here,
\[
U\psi = \cS \binom{c_{11}\psi_1+c_{12}\psi_2}{c_{21}\psi_1+c_{22}\psi_2}=\binom{c_{11}S_1\psi_1+c_{12}S_1\psi_2}{c_{21}S_{-1}\psi_1+c_{22}S_{-1}\psi_2}=\begin{pmatrix}c_{11}S_1&c_{12}S_1\\c_{21}S_{-1}&c_{22}S_{-1}\end{pmatrix}\begin{pmatrix}\psi_1\\\psi_2\end{pmatrix}
\]

A popular coined walk is the \emph{Hadamard walk} defined by $C = \frac{1}{\sqrt{2}}\begin{pmatrix}1&1\\1&-1\end{pmatrix}$. There are some slight variations among references, with several authors \cite{NayakVishwanath2000,Konno2002,Konno} substituting the roles of $S_1$ and $S_{-1}$. We do the same in \eqref{e:uhagro}; this changes nothing to the results.

\subsection{Shunt decomposition model} 
 This model is a generalization of coined quantum walks for graphs, introduced in \cite{Aha}. We follow the definition in \cite[Chapter 7]{GodsilZhan2023}. 
 
Let $G$ be a $\nu$-regular graph with adjacency matrix $A$. 
A shunt decomposition of $G$ is a collection of permutation matrices $P_1, P_2,\ldots, P_\nu$ such that
\[
A = P_1 + P_2+ \cdots + P_\nu.
\]
The shift matrix of the walk is defined as
\[
\cS =
\begin{pmatrix}
P_1^{-1} & & & \\
& P_2^{-1} & & \\
& & \ddots & \\
& & & P_\nu^{-1}
\end{pmatrix},
\]
and the unitary matrix of the shunt decomposition walk is
\[
U = \cS (C \otimes I),
\]
where $C$ is a $\nu \times \nu$ coin matrix. Here $U$ is a unitary operator on the space $\ell^2(G) \otimes \C^\nu$.

Let us specialize to shunt decompositions of the Cayley graphs $\Cay(\Z_n,\cX)$, where the set of generators $\cX$ is independent of $n$. Examples of such graphs include circulant graphs, which were the main topic of interest of \cite[Chapter 7]{GodsilZhan2023}. For a Cayley graph $\Cay(\Z_n,\cX)$, with generator set $\cX=\{a_1,a_2,\ldots, a_\nu\}$, the natural choice of permutation matrices are those given by the permutation maps $\pi_i(x)= x+a_i$ on $n$ points. In this case, $(P_i)_{xy}=\delta_y(x+a_i)$; consequently, $(P_i^{-1})_{xy}=\delta_x(y+a_i)$, and the shift operator is
\[(\cS \psi (k))_j= \psi_j(k-a_j) 
\]
Note that $\cS\psi = \sum_{j,k}\psi_j(k-a_j)\delta_k\mathop\otimes\delta_j =\begin{pmatrix} S_{a_1}\psi_1\\ \vdots\\ S_{a_\nu}\psi_\nu\end{pmatrix}$. This generalizes the Hadamard walk where $\nu=2$, $a_j\in \{-1,1\}$ and the \emph{Grover walk} where $\nu=3$, $a_1=-1$, $a_2=0$ and $a_3=1$ with a Grover coin.

\subsection{Arc reversal model}\label{sec:arc}
This model was first introduced by Watrous \cite{Watrous'01} for regular graphs, and later generalized by Kendon \cite{Kendon'06}. 

Consider an undirected graph $G$, and associate to each edge two arcs. The quantum walk takes place in the Hilbert space spanned by the characteristic vectors $\delta_{u,v}$ of the arcs $(u,v)$. Let $R$ be the permutation that reverses the arcs, i.e., $R\delta_{u,v}=\delta_{v,u}$. A second unitary matrix is obtained by considering orderings $f_u: \{1,2,\ldots, \deg(u)\} \rightarrow \{v:v \sim u\}$, where $f_u(j)$ is the $j$-th neighbour of the vertex $u$. Now, define a unitary operator $C_u$ acting on all the outgoing arcs of $u$, and define the unitary matrix $C$ as the block diagonal matrix with diagonal blocks $C_j$. The transition matrix for the arc-reversal quantum walk is defined as
\[
U=RC\,.
\]
Roughly speaking, each vertex $u$ has $\deg(u)$ spins in this walk.

We concentrate on arc-reversal quantum walks on $\Z^d$. In this case, all the vertices have degree $2d$ and we look at the case where $C_j$ is the same for all $j$ and $f_u(j)=u\pm e_j$. Identifying $\delta_{u,f_{u}(j)}$ with $\delta_u \mathop\otimes \delta_j$, the Hilbert space for the quantum walk can be identified  as $\ell^2(\Z^d)\mathop\otimes \C^{2d}$ and $R\delta_u\otimes\delta_j=\delta_{f_u(j)}\otimes\delta_k$, where $k$ is chosen so that $u$ is the $k$th neighbor of $f_u(j)$. Now, taking the shift matrix as in the shunt-decomposition model, we have $R=PS$, where $S$ is the shift matrix in the shunt-decomposition model and $P$ is an appropriate $2d \times 2d$ permutation matrix. In the particular case of $d=1$, the arc-reversal quantum walk can be expressed as 
\begin{equation}\label{eqn:unitary_arcreversal}
 U=\begin{pmatrix}
	0 & 1 \\
	1 &0
\end{pmatrix}
\begin{pmatrix}
	S_{-1} & 0 \\
	0 & S_1
\end{pmatrix}
C,   
\end{equation}
where $C$ is the homogeneous coin matrix and $\begin{pmatrix}
	0 & 1 \\
	1 &0
\end{pmatrix}$ flips the spin from $\delta_j$ to $\delta_k$ in the previous notation $R\delta_u\otimes\delta_j=\delta_{f_u(j)}\otimes\delta_k$.

\subsection{PUTO Model}\label{sec:puto}
The periodic unitary transition operators (PUTO) is a model of higher-dimensional quantum walks on $\mathscr{H}=\ell^2(\Z^d,\C^\nu)$ introduced in \cite{Tate2019,KomatsuTate2019}. The walk takes the form
\[
U = \sum_{\balpha\in F} \tau^{\balpha} P^{\balpha} \cC \,,
\]
where $F\subset \Z^d$ is finite, the $(P^{\balpha})_{\balpha\in F}$ form a resolution of the identity of $\C^\nu$ and $\tau^{\balpha}$ is defined by $(\tau^{\balpha} \psi)(\bk)=\psi(\bk-\balpha)$ for $\psi\in \mathscr{H}$.

Our framework encompasses such walks. In fact, \[
\tau^{\balpha} P^{\balpha} \cC\psi = \tau^{\balpha} P^{\balpha} \begin{pmatrix} \sum_j c_{1,j}\psi_j\\ \vdots\\ \sum_j c_{\nu,j} \psi_j\end{pmatrix} = \tau^{\balpha} \begin{pmatrix} \sum_k p_{1,k}^{\balpha} \sum_j c_{kj}\psi_j \\ \vdots\\ \sum_k p_{\nu,k}^{\balpha} \sum_j c_{k,j}\psi_j\end{pmatrix} = \begin{pmatrix} \sum_{k,j} c_{k,j}p_{1,k}^{\balpha} S_{\balpha} \psi_j\\ \vdots\\ \sum_{k,j} p_{\nu,k}^{\balpha} c_{k,j}S_{\balpha} \psi_j \end{pmatrix}.
\]
This shows that $U\psi = \sum_{\balpha\in F}\tau^{\balpha} P^{\balpha} \cC$ has the form \eqref{e:ugen} with 
\begin{equation}\label{e:puto}
U_{i,j} = \sum_{k=1}^\nu \sum_{\balpha\in F} c_{k,j} p_{i,k}^{\balpha} S_{\balpha} \,.
\end{equation}
The papers \cite{Tate2019,KomatsuTate2019} derive several properties about these operators, particularly the spectrum, the presence of flat bands, and its relation to the time evolution of $U$ on $\mathscr{H}$. The Floquet matrix $\widehat{U}(\btheta)$ here takes the form $\widehat{U}(\btheta)=V(\btheta)C$, where $V(\btheta)=\sum_{\balpha} \ee^{2\pi \ii\btheta\cdot \balpha}P^{\balpha}$. The main results in \cite{Tate2019} however focus on more specific operators $(P^{\balpha})$ defined as follows:
\begin{itemize}
\item $\nu = 2d$ and $F = F_{std} = \{\pm e_1,\dots,\pm e_d\}$, where $(e_j)$ is the standard basis of $\Z^d$, and $P^{e_j}:= P_{2j-1}$, $P^{-e_j}:= P_{2j}$, with $P_j \phi = \phi_j \delta_j$, the orthogonal projection in $\C^\nu$ onto $\delta_j$. The matrix $P_r=E_{rr}$ has all entries zero except at $rr$, where it is one. That is, $p_{i,k}^{e_p}=1$ if $i=k= 2p-1$, $p_{i,k}^{e_p}=0$ otherwise, $p_{i,k}^{-e_p}=1$ if $i=k=2p$ and $p_{i,k}^{-e_p}=0$ otherwise. So \eqref{e:puto} simplifies to
\begin{equation}\label{e:putosta}
U_{i,j} = \sum_{\balpha} c_{i,j}p_{i,i}^{\balpha} S_{\balpha} = \begin{cases} c_{i,j} S_{e_{\frac{i+1}{2}}}& \text{if } i \text{ is odd,}\\ c_{i,j}S_{-e_{\frac{i}{2}}}&\text{if } i \text{ if even.} \end{cases}
\end{equation}
For example, if $d=1$, then $U = \begin{pmatrix} c_{11}S_{e_1}&c_{12} S_{e_1}\\ c_{21}S_{-e_1}& c_{22} S_{-e_1}\end{pmatrix}$. Here $e_1=1$, so we get the usual coined walk. If $d=2$, $U$ is $4\times 4$.
\item $\nu=2d+1$ and $F= F_{lazy}= F_{std} \cup \{\mathbf{0}\}$ and $P^{e_j}=P_j$, $P^{\mathbf{0}}=P_{d+1}$ and $P^{-e_j} = P_{d+1+j}$. This gives 
\[
U_{i,j} = \sum_{\balpha} c_{i,j}p_{i,i}^{\balpha} S_{\balpha} = \begin{cases} c_{i,j} S_{e_i}& \text{if } i\le d\\ c_{i,j}& \text{if } i=d+1,\\ c_{i,j}S_{-e_{i-d-1}}&\text{if } i>d+1 \end{cases}
\]
\end{itemize}
A popular coin $C=(c_{i,j})$ that works for both cases is the \emph{Grover coin} on $\C^\nu$ given by
\[
C_{G}=\frac{2}{\nu} J-I\,,
\]
where $J$ is the matrix with all $1$. The most common choice for this coin is $d=1$ and $\nu=3$ with $F_{lazy}$. This coin features localization (the operator $U$ has an eigenvalue) and is thus less relevant for our study, but is otherwise interesting. The limiting distribution of the two-dimensional Grover walk with shifts along the $x-$ and $y-$axes was studied in \cite{Watabe2dgroverlimit}.

A more interesting coin for us is the \emph{Fourier coin} 
\[
C_F = \frac{1}{\sqrt{\nu}} \left(\omega^{(p-1)(q-1)}\right)_{p,q=1}^\nu
\]
for $\omega = \ee^{2\pi \ii /\nu}$. We study the case $d=2$ and $\nu=4$ in detail in Section~\ref{sec:fou}.

\section{RAGE theorem for quantum walks}\label{app:rage}

We adapt here the RAGE theorem to homogeneous discrete-time quantum walks. This result is not used in the text and is given here for completeness.

We recall that the homogeneous quantum walks $U$ of finite range studied here, just like periodic Schr\"odinger operators on crystals, exhibit no singularly continuous spectrum \cite{Tate2019}, and that $U$ can only have finitely many flat bands (eigenvalues). Let us denote these by $\lambda_1,\dots,\lambda_f$ and let $P_{\lambda_i}$ be the corresponding eigenprojectors. Then we let 
\[
\mathscr{H}_{pp} = \mathrm{Ran} P \qquad \text{and} \qquad \mathscr{H}_c = \mathscr{H}_{pp}^\bot \,,
\]
where $P=\sum_{i=1}^f P_{\lambda_i}$. Then we have:

\begin{thm}\label{thm:rage}
Let $U$ be a quantum walk \eqref{e:ugen}-\eqref{eqn:homogen_unitary}. Let $[\chi_{\Lambda}(\br)]_i := \mathbf{1}_{\Lambda}(\br)$ for all $i\le \nu$.
\begin{enumerate}[\rm (1)]
    \item $\psi\in\mathscr{H}_c$ iff for any finite $\Lambda\subset \Z^d$, we have $\|\chi_\Lambda U^n\psi\|\to 0$.
    \item $\psi\in \mathscr{H}_{pp}$ iff for any $\varepsilon>0$ there is a finite $\Lambda\subset \Z^d$ such that $\sup_n \|\chi_{\Lambda^c}U^n\psi\|<\varepsilon$.
\end{enumerate}
\end{thm}

Thus, a state is in the pure point space iff at all times, most of its mass lies within a fixed compact set, and it is in the continuous space iff it escapes from any compact set, after sufficient time has passed. 
Theorem~\ref{thm:rage} strengthens some results in \cite{Tate2019}.

\begin{proof}
By the RAGE theorem for general unitary operators $U$ on a Hilbert space, we have for any compact operator $K$, $\frac{1}{n}\sum_{m=0}^{n-1}\|KU^m\varphi\|^2\to \|KP\varphi\|^2$ for any $\varphi\in \mathscr{H}$, see \cite[Thm. 5.5.6]{Simon4}. 

 Take $K=\chi_{\Lambda}$ and let $\varphi=\varphi_1+\varphi_2$, where $\varphi_1\in \mathscr{H}_{pp}$ and $\varphi_2\in \mathscr{H}_c$. Then this limit reduces to 
 $\|\chi_\Lambda \varphi_1\|^2$. 
  This is zero for all finite $\Lambda$ iff $\varphi_1=0$, i.e. iff $\varphi\in \mathscr{H}_c$.

    So far we showed that $\psi\in \mathscr{H}_c \iff \frac{1}{n}\sum_{m=0}^{n-1} \|\chi_{\Lambda} U^n \psi\|^2 \to 0$. In particular, if $\|\chi_{\Lambda} U^n \psi\|\to 0$, then it vanishes in Ces\`aro sense, so $\psi \in \mathscr{H}_c$. 

    By the spectral theorem for unitary operators (see e.g. \cite[Thm. 5.1] {BSU}), for any $\varphi,\psi\in \mathscr{H}$, there exists a complex measure $\mu_{\varphi,\psi}$ such that $\langle \psi, U^n\varphi\rangle = \int_0^{2\pi} \ee^{\ii n\theta}\,\dd\mu_{\varphi,\psi}(\theta)$. If $\psi\in \mathscr{H}_{c}=\mathscr{H}_{ac}$ here, the measure $\mu_{\psi,\psi}$ is absolutely continuous with respect to the Lebesgue, hence so is $\mu_{\varphi,\psi}$ for any $\varphi\in\mathscr{H}$ by the Cauchy-Schwarz inequality $|\mu_{\varphi,\psi}(J)|^2\le \mu_{\varphi,\varphi}(J)\mu_{\psi,\psi}(J)$ for Borel $J$. In other words, $\langle \varphi, U^n\psi\rangle =\int\ee^{\ii n\theta}g_{\varphi,\psi}(\theta)\,\dd\theta$ with $g_{\varphi,\psi}\in\ell^1$. By the Riemann-Lebesgue lemma, this goes to zero as $n\to\infty$. Specializing to $\varphi=\delta_{\br}\mathop\otimes\delta_i$, we deduce that $|[U^n\psi]_i(\br)|\to 0$ for any $i,\br$, hence $\|\chi_\Lambda U^n\psi\|\to 0$.

    We have proved (1). For (2), if $\psi\in \mathscr{H}_{pp}$, then $U^n \psi = \sum_{i=1}^f \lambda_i^n P_{\lambda_i} \psi$, so $\|\chi_{\Lambda^c} U^n\psi\|\le \sum_{i=1}^f \|\chi_{\Lambda^c} P_{\lambda_i}\psi\|$. Because each $P_{\lambda_i}\psi$ is $\ell^2$, we may find finite sets $\Lambda_i$ such that $\|\chi_{\Lambda^c_i} U^n\psi\|<\varepsilon/f$, so the conclusion of (2) follows by taking $\Lambda = \cup_{i=1}^f\Lambda_i$. Conversely, if the conclusion of (2) holds, write $\psi=\psi_1+\psi_2$ with $\psi_1\in \mathscr{H}_{pp}$ and $\psi_2\in \mathscr{H}_c$. By hypothesis, $\exists \Lambda$ with $\|\chi_{\Lambda^c}U^n\psi\|<\varepsilon$. And by the $(\implies)$ part we proved, $\exists \Lambda_1$ such that $\|\chi_{\Lambda_1^c}U^n\psi_1\|<\varepsilon$. By the triangle inequality, if $\Lambda_2=\Lambda\cup\Lambda_1$, this gives $\|\chi_{\Lambda_2^c}U^n\psi_2\|<2\varepsilon$ for all $n$. Therefore, $\|\psi_2\|^2 = \|\chi_{\Lambda_2}U^n\psi_2\|^2+\|\chi_{\Lambda_2^c}U^n\psi_2\|^2<3\varepsilon$ by taking $n$ large enough and using part (1). As $\varepsilon$ is arbitrary, this shows that $\psi_2=0$, hence $\psi\in \mathscr{H}_{pp}$.
\end{proof}

\begin{rem}
    Theorem~\ref{thm:rage} holds for general quantum walks (not necessarily homogeneous) if one replaces $\|\chi_{\Lambda}U^n\psi\|$ by $\frac{1}{n}\sum_{m=0}^{n-1}\|\chi_{\Lambda}U^m\psi\|^2$ in (1), as can be seen from the previous proof.
\end{rem}

\providecommand{\bysame}{\leavevmode\hbox to3em{\hrulefill}\thinspace}
\providecommand{\MR}{\relax\ifhmode\unskip\space\fi MR }
\providecommand{\MRhref}[2]{%
}
\providecommand{\href}[2]{#2}

\end{document}